\documentclass[letterpaper,11pt]{article}

\usepackage[margin=1in]{geometry}
\usepackage{setspace}

\usepackage{amsmath,amssymb,amsthm,mathtools,mathrsfs}
\usepackage{nicefrac,units}
\usepackage{forest}

\usepackage{algorithm}
\usepackage{algpseudocode}
\usepackage{thm-restate}

\usepackage{array,booktabs,multirow,colortbl}
\usepackage{xcolor}
\definecolor{lightblue}{RGB}{230, 240, 255}
\definecolor{lightgray}{RGB}{240, 240, 240}

\usepackage[bookmarksnumbered=true,hidelinks,colorlinks,linkcolor=blue,citecolor=purple]{hyperref}
\usepackage[capitalize,sort]{cleveref}

\usepackage{subcaption}

\usepackage{tikz,pgfplots}
\usetikzlibrary{fit,positioning}
\pgfplotsset{compat=1.18}
\usetikzlibrary{
	positioning,
	shapes.geometric,
	arrows.meta,
	decorations.pathreplacing,
	patterns,
	patterns.meta,
	calc
}
\tikzstyle{startstop} = [rectangle, rounded corners, minimum width=3cm, minimum height=1cm, text centered, draw=black, fill=gray!20]
\tikzstyle{process}   = [rectangle, minimum width=3cm, minimum height=1cm, text centered, draw=black, fill=blue!10]
\tikzstyle{decision}  = [diamond, minimum width=3cm, minimum height=1cm, text centered, draw=black, fill=red!10]
\tikzstyle{arrow}     = [thick,->,>=stealth]

\usepackage{tabularx}

\usepackage{booktabs}
\usepackage{array}
\usepackage{ragged2e}
\usepackage{graphicx}
\usepackage{colortbl}
\usepackage{amsmath}
\usepackage{amssymb}
\definecolor{headerblue}{RGB}{25,55,92}
\definecolor{rowblue1}{RGB}{240,245,255}
\definecolor{rowblue2}{RGB}{225,235,250}
\definecolor{accentblue}{RGB}{65,105,225}

\usepackage{enumitem}
\usepackage[symbol]{footmisc}

\usepackage{multirow}
\renewcommand{\arraystretch}{1.5}

\usepackage{natbib}

\usepackage{titling}
\pretitle{%
	\begin{center}\LARGE\bfseries\vspace{-1em}%
		\textsc
	}
	\posttitle{\par\end{center}\vspace{1em}}

\usepackage[draft]{changes} 
\definechangesauthor[name={Masoud}, color=teal]{ms}

\definecolor{airforceblue}{rgb}{0.1, 0.31, 0.59}
\newcommand{\CComment}[1]{\hfil \textcolor{airforceblue}{\Comment{#1}}}

\newtheorem{lemma}{Lemma}
\newtheorem{theorem}{Theorem}[section]
\newtheorem{observation}{Observation}

\newtheorem{definition}{Definition}

\newtheorem*{claim*}{Claim}
\newtheorem{example}{Example}

\newtheorem{remark}{Remark}

\usepackage{apxproof}
\newtheoremrep{lem}[lemma]{Lemma}

\crefname{ineq}{Inequality}{Inequalities}
\creflabelformat{ineq}{#2{\upshape(#1)}#3}

\makeatletter
\def\input@path{{../Dropbox/MMS-10-13/Draft 5.1.1/main}}
\graphicspath{{../Dropbox/MMS-10-13/Draft 5.1.1/main}}
\makeatother

\definecolor{darkblue}{rgb}{0.03, 0.27, 0.49}
\definecolor{olivegreen}{rgb}{0.33, 0.42, 0.18}

\newcommand{\etal}{\textit{et al.~}}
\newcommand{\SharifUniversity}{Sharif University of Technology, Tehran, Iran}
\newcommand{\TeIAS}{Tehran Institute for Advanced Studies (TeIAS), Tehran, Iran}
\newcommand{\lbnone}{\frac{\Pnumb}{\sqrt2}}

\newcommand{\riki}[5]{(#1, #2, #3, #4, #5)}
\newcommand{\Srik}[2]{S(#2)}

\newcommand{\mms}{\Psi}
\newcommand{\MMS}{\textsf{MMS}}
\newcommand{\allmms}[1]{\Psi^n_{#1}(M)}
\newcommand{\dotmms}[1]{\dot{\Psi}_{#1}}
\newcommand{\ddotmms}[1]{\ddot{\Psi}_{#1}}
\newcommand{\hatmms}[1]{\hat{\Psi}_{#1}}
\newcommand{\checkmms}[1]{\check{\Psi}_{#1}}

\newcommand{\valu}{v}

\newcommand{\prinstance}{\hat{\mathcal I}}
\newcommand{\pragents}{\hat{N}}
\newcommand{\pragent}{\hat{a}}
\newcommand{\prgoods}{\hat{M}}

\newcommand{\prvalu}{\hat{v}}

\newcommand{\prnumb}{\hat{n}}
\newcommand{\prgood}{\hat{g}}
\newcommand{\prgoodrat}[1]{\hat{g}_{#1}}
\newcommand{\prit}{\hat{m}}
\newcommand{\prreduc}{\mathfrak{R}}

\newcommand{\prauxmms}[3]{\mms^{#3}_{\prauxval{#1}}(#2)}
\newcommand{\prauxval}[1]{({#1}\star \prvalu)}
\newcommand{\prauxfun}[2]{({#1}\star \prvalu_{#2})}

\newcommand{\praux}[2]{({#1}\star \prvalu)(#2)}

\newcommand{\seinstance}{\check{\mathcal I}}
\newcommand{\seagents}{\check{N}}
\newcommand{\segoods}{\check{M}}

\newcommand{\senumb}{\check{n}}

\newcommand{\Pinstance}{\mathcal I}
\newcommand{\Pagents}{N}
\newcommand{\Pgoods}{M}
\newcommand{\Pgood}{g}

\newcommand{\Pnumb}{n}
\newcommand{\Pit}{m}

\newcommand{\Jvalu}{{v_i}^{\textsf{norm}}}
\newcommand{\Jnumb}{n'}
\newcommand{\Jgoods}{M'}

\newcommand{\Gvalu}{\auxwal{\vF_{s_\agenti}}^{\textsf{norm}}}
\newcommand{\Gnumb}{n'}
\newcommand{\Ggoods}{M'}

\newcommand{\instance}{\dot{\mathcal I}}
\newcommand{\agents}{\dot{N}}
\newcommand{\goods}{\dot{M}}
\newcommand{\good}{\dot{g}}
\newcommand{\goodrat}[1]{{{\good}_{#1}}}
\newcommand{\numb}{\dot{n}}
\newcommand{\It}{\dot{m}}
\newcommand{\Nyek}{\Pagents^g}
\newcommand{\Ndo}{\Pagents^r}

\newcommand{\Sinstance}{\ddot{\mathcal I}}
\newcommand{\Sagents}{\ddot N}
\newcommand{\Sgoods}{\ddot M}
\newcommand{\Sgood}{\ddot{g}}
\newcommand{\Sgoodrat}[1]{{{\Sgood}_{#1}}}
\newcommand{\Snumb}{\ddot{n}}
\newcommand{\Sit}{\ddot{m}}

\newcommand{\reduction}{\rho}
\newcommand{\reductiontype}{\mathcal R}

\newcommand{\bundle}{B}
\newcommand{\allocation}{A}
\newcommand{\agenti}{i}
\newcommand{\agent}{a}

\newcommand{\vF}{f}
\newcommand{\vG}{h}
\newcommand{\vGF}{\vG\star\FJ}
\newcommand{\vH}{w}
\newcommand{\vHA}{w}
\newcommand{\vHB}{z}
\newcommand{\FJ}{\mathring \vF}

\newcommand{\LM}{\vF_\lambda}
\newcommand{\rstar}{\widetilde\reductiontype}

\newcommand{\AN}[1]{\vHA_{#1}}

\newcommand{\auxval}[1]{({#1}\star \valu)}

\newcommand{\auxagent}[3]{({#1}\star \valu_{#3})(#2)}
\newcommand{\auxfun}[2]{({#1}\star \valu_{#2})}

\newcommand{\auxwal}[1]{({#1}\star \valu_\agenti)}

\newcommand{\ff}[1]{\auxagent{\FJ}{#1}{\agenti}}
\newcommand{\gf}[1]{\auxagent{\vGF}{#1}{\agenti}}
\newcommand{\fs}[1]{\auxagent{\vF_{s_\agenti}}{#1}{\agenti}}

\newcommand{\mmsfrac}[2]{\tfrac{#1}{#2}}

\newcommand{\finib}{B}

\newcommand{\halt}{\gamma}

\newcounter{MasoudCommentCounter}
\newcommand{\MasoudCommentList}{}
\newcommand{\MasoudComment}[1]{%
	\stepcounter{MasoudCommentCounter}%
	\expandafter\gdef\expandafter\MasoudCommentList\expandafter{%
		\MasoudCommentList
		\par \textcolor{magenta}{\textbf{Masoud~\theMasoudCommentCounter:}} #1 (See \hyperlink{MasoudComment\theMasoudCommentCounter}{here}.)%
	}%
	\textcolor{magenta}{\textbf{\hypertarget{MasoudComment\theMasoudCommentCounter}{}Masoud~\theMasoudCommentCounter:} #1}%
}

\newcounter{AlirezaCommentCounter}
\newcommand{\AlirezaCommentList}{}
\newcommand{\AlirezaComment}[1]{%
	\stepcounter{AlirezaCommentCounter}%
	\expandafter\gdef\expandafter\AlirezaCommentList\expandafter{%
		\AlirezaCommentList
		\par \textcolor{cyan}{\textbf{Alireza~\theAlirezaCommentCounter:}} #1 (See \hyperlink{AlirezaComment\theAlirezaCommentCounter}{here}.)%
	}%
	\textcolor{cyan}{\textbf{\hypertarget{AlirezaComment\theAlirezaCommentCounter}{}Alireza~\theAlirezaCommentCounter:} #1}%
}

\newcounter{AmShZCommentCounter}
\newcommand{\AmShZCommentList}{}
\newcommand{\AmShZComment}[1]{%
	\stepcounter{AmShZCommentCounter}%
	\expandafter\gdef\expandafter\AmShZCommentList\expandafter{%
		\AmShZCommentList
		\par \textcolor{olivegreen}{\textbf{AmShZ~\theAmShZCommentCounter:}} #1 (See \hyperlink{AmShZComment\theAmShZCommentCounter}{here}.)%
	}%
	\textcolor{olivegreen}{\textbf{\hypertarget{AmShZComment\theAmShZCommentCounter}{}AmShZ~\theAmShZCommentCounter:} #1}%
}

\newcounter{EhsanCommentCounter}
\newcommand{\EhsanCommentList}{}
\newcommand{\EhsanComment}[1]{%
	\stepcounter{EhsanCommentCounter}%
	\expandafter\gdef\expandafter\EhsanCommentList\expandafter{%
		\EhsanCommentList
		\par \textcolor{brown}{\textbf{Ehsan~\theEhsanCommentCounter:}} #1 (See \hyperlink{EhsanComment\theEhsanCommentCounter}{here}.)%
	}%
	\textcolor{brown}{\textbf{\hypertarget{EhsanComment\theEhsanCommentCounter}{}Ehsan~\theEhsanCommentCounter:} #1}%
}

\title{Improved Maximin Share Guarantee for Additive Valuations}
\author{
		   Ehsan Heidari\thanks{\SharifUniversity}\\ 
		   \texttt{\small eh3an1383@gmail.com} 
		   \and 
		   Alireza Kaviani\footnotemark[1]\\ 
		   \texttt{\small akaviani05@gmail.com} 
		   \and 
		   Masoud Seddighin\thanks{\TeIAS}\\ 
		   \texttt{\small m.seddighin@teias.institute} 
		   \and 
		   AmirMohammad Shahrezaei\footnotemark[1]\\ 
		   \texttt{\small s.a.m.shahrezaei@gmail.com}
}
\date{}

\begin{document}
	\maketitle
	
	\begin{toappendix}
		\section{Table of Frequently Used Notation} \label{sec:notationstable}

\begin{table}[H]
	\centering
	\renewcommand{\arraystretch}{1.5}
	\begin{tabular}{|>{\centering\arraybackslash}p{2.25cm}|>{\centering\arraybackslash}p{2.25cm}|>{\centering\arraybackslash}p{2.25cm}|>{\centering\arraybackslash}p{2.25cm}|>{\centering\arraybackslash}p{2.25cm}|>{\centering\arraybackslash}p{2.25cm}|}
		\hline
		\textbf{Label} & \textbf{Instance} & \textbf{Agents} & \textbf{Goods} & \textbf{Good} & \textbf{MMS} \\
		\hline
		Initial    & $\Pinstance$ & $\Pagents$ & $\Pgoods$ & $\Pgood_k$ & $\allmms{v}$ \\
		\hline
		Primary    & $\instance$ & $\agents$ & $\goods$ & $\good_k$ & $\dotmms{v}$ \\
		\hline
		Secondary  & $\Sinstance$ & $\Sagents$ & $\Sgoods$ & $\Sgood_k$ & $\ddotmms{v}$ \\
		\hline
		$\Nyek$     & \multicolumn{5}{>{\centering\arraybackslash}p{13cm}|}{Green agents. $\{ \agent_i \in \Pagents \mid \valu_\agenti(\{\good_{2\numb+1}\}) \ge 1 - \alpha \}$} \\
		\hline
		$\Ndo$      & \multicolumn{5}{>{\centering\arraybackslash}p{13cm}|}{Red agents. $\{ \agent_i \in \Pagents \mid \valu_\agenti(\{\good_{2\numb+1}\}) < 1 - \alpha \}$} \\
		\hline
		$\finib_k$  & \multicolumn{5}{>{\centering\arraybackslash}p{13cm}|}{Bags. Initialized with $\{\Sgood_k, \Sgood_{\Snumb+k}, \Sgood_{3\Snumb+1-k}\}$ in Case~1, and $\{\good_k, \good_{\numb+k}\}$ in Case~2.} \\
		\hline
		$\FJ$      & \multicolumn{5}{>{\centering\arraybackslash}p{13cm}|}{$f_{\tfrac{4}{3} \alpha - 1}$} \\
		\hline
		$\reductiontype^1$      & \multicolumn{2}{>{\centering\arraybackslash}p{5cm}|}{$\{\prnumb, x\}$} & $\reductiontype^2$      & \multicolumn{2}{>{\centering\arraybackslash}p{5cm}|}{$\{2\prnumb - 1, 2\prnumb, x\}$}\\
		\hline
		$\reductiontype^3$      & \multicolumn{2}{>{\centering\arraybackslash}p{5cm}|}{$\{3\prnumb-2, 3\prnumb - 1, 3\prnumb, x\}$} & $\reductiontype^4$      & \multicolumn{2}{>{\centering\arraybackslash}p{5cm}|}{$\{4\prnumb - 3, 4\prnumb-2, 4\prnumb-1, 4\prnumb, x\}$}\\
		\hline
		$\rstar^1$      & \multicolumn{2}{>{\centering\arraybackslash}p{5cm}|}{$\{1, x\} \quad x \ge 2\prnumb + 1$} & $\rstar^2$      & \multicolumn{2}{>{\centering\arraybackslash}p{5cm}|}{$\{1, x\} \quad x \ge 2$}\\
		\hline
	\end{tabular}
	\caption{Notation Table}
	\label{tab:mms-blank}
\end{table}
		\newpage
		\newpage
\section{Examples} \label{sec:examples}

\newcommand{\partitionA}[2]{\textcolor{red}{\nicefrac{#1}{#2}}}
\newcommand{\partitionB}[2]{\textcolor{blue}{\nicefrac{#1}{#2}}}
\newcommand{\partitionC}[2]{\textcolor{teal}{\nicefrac{#1}{#2}}}
\newcommand{\partitionD}[2]{\textcolor{olive}{\nicefrac{#1}{#2}}}
\newcommand{\partitionES}[1]{\textcolor{violet}{#1}}
\newcommand{\partitionE}[2]{\textcolor{violet}{\nicefrac{#1}{#2}}}

\begin{example}[Reductions]
	We visualize three of our modified reduction rules \( \reductiontype^1 \), \( \reductiontype^2 \), and \( \rstar^1 \) and illustrate how their outcomes differ from the corresponding classical reductions. Consider an instance with $2$ agents and $7$ goods, whose valuations are given in the table below:
	
	\[
	\begin{array}{c|ccccccc}
		\text{Good} & g_1 & g_2 & g_3 & g_4 & g_5 & g_6 & g_7 \\
		\hline
		\text{Agent 1} & \partitionA{7}{13} & \partitionB{7}{13} & \partitionB{4}{13} & \partitionA{3}{13} & \partitionA{3}{13} & \partitionB{1}{13} & \partitionB{1}{13} \\
		\text{Agent 2} & \partitionA{8}{13} & \partitionA{5}{13} & \partitionB{5}{13} & \partitionB{3}{13} & \partitionB{2}{13} & \partitionB{2}{13} & \partitionB{1}{13}
	\end{array}
	\]
	
	For \( \alpha = \nicefrac{10}{13} \):
	\begin{itemize}
		\item The classical rule $\mathsf{R}^1$ assigns the bundle \( \{g_2, g_3\} \), while our modified rule selects \( \{g_2 , g_5\} \).
		\item The classical rule $\mathsf{R}^2$ assigns \( \{g_3, g_4, g_5\} \), whereas our version picks \( \{g_3, g_4, g_6\} \).
		\item The classical rule $\widetilde{\mathsf{R}}^1$ assigns \( \{g_1, g_5\} \), while our modified rule selects \( \{g_1, g_6\} \).
	\end{itemize}
	
	Note that unlike the classical reductions, we do not allocate these bundles immediately and defer the matching process.
\end{example}

\begin{figure}[ht]
	\centering
	
	\begin{tikzpicture}[x=1.3cm, y=1.3cm, font=\small]
		\foreach \i/\label in {1/$g_1$,2/$g_2$,3/$g_3$,4/$g_4$,5/$g_5$,6/$g_6$,7/$g_7$} {
			\node[draw, rectangle, rounded corners=3pt, minimum width=1cm, minimum height=1cm, inner sep=2pt] (good\i) at (\i,0) {\label};
		}
		
		\node[font=\footnotesize\bfseries] at ([shift={(-45pt,10pt)}]good1.north) {Agent 1};
		\foreach \i/\val in {1/7,2/7,3/4,4/3,5/3,6/1,7/1} {
			\node[above=4pt, font=\scriptsize] at (good\i.north) {$\nicefrac{\val}{13}$};
		}
		
		\node[font=\footnotesize\bfseries] at ([shift={(-45pt,-13pt)}]good1.south) {Agent 2};
		\foreach \i/\val in {1/8,2/5,3/5,4/3,5/2,6/2,7/1} {
			\node[below=4pt, font=\scriptsize] at (good\i.south) {$\nicefrac{\val}{13}$};
		}
		
		\draw[fill=blue!20, fill opacity=0.3, draw=blue!60, thick, rounded corners=4pt]
		(1.5,0.9) -- (5.45,0.9) -- (5.45,-0.45) -- (4.55,-0.45) -- (4.55,0.45) -- (2.45,0.45) -- (2.45,-0.45) -- (1.5,-0.45) -- cycle;
		
		\draw[fill=red!20, fill opacity=0.3, draw=red!60, thick, rounded corners=4pt]
		(1.55,0.45) -- (3.45,0.45) -- (3.45,-0.45) -- (1.55,-0.45) -- cycle;
	\end{tikzpicture}
	
	\hfill
	
	\begin{tikzpicture}[x=1.3cm, y=1.3cm, font=\small]
		\foreach \i/\label in {1/$g_1$,2/$g_2$,3/$g_3$,4/$g_4$,5/$g_5$,6/$g_6$,7/$g_7$} {
			\node[draw, rectangle, rounded corners=3pt, minimum width=1cm, minimum height=1cm, inner sep=2pt] (good\i) at (\i,0) {\label};
		}
		
		\node[font=\footnotesize\bfseries] at ([shift={(-45pt,10pt)}]good1.north) {Agent 1};
		\foreach \i/\val in {1/7,2/7,3/4,4/3,5/3,6/1,7/1} {
			\node[above=4pt, font=\scriptsize] at (good\i.north) {$\nicefrac{\val}{13}$};
		}
		
		\node[font=\footnotesize\bfseries] at ([shift={(-45pt,-13pt)}]good1.south) {Agent 2};
		\foreach \i/\val in {1/8,2/5,3/5,4/3,5/2,6/2,7/1} {
			\node[below=4pt, font=\scriptsize] at (good\i.south) {$\nicefrac{\val}{13}$};
		}
		
		\draw[fill=blue!20, fill opacity=0.3, draw=blue!60, thick, rounded corners=4pt]
		(2.5,-0.9) -- (6.45,-0.9) -- (6.45,0.45) -- (5.55,0.45) -- (5.55,-0.45) -- (4.45,-0.45) -- (4.45,0.45) -- (2.5,0.45) -- cycle;
		
		\draw[fill=red!20, fill opacity=0.3, draw=red!60, thick, rounded corners=4pt]
		(2.55,0.45) -- (5.45,0.45) -- (5.45,-0.45) -- (2.55,-0.45) -- cycle;
	\end{tikzpicture}
	
	\hfill
	
	\begin{tikzpicture}[x=1.3cm, y=1.3cm, font=\small]
		\foreach \i/\label in {1/$g_1$,2/$g_2$,3/$g_3$,4/$g_4$,5/$g_5$,6/$g_6$,7/$g_7$} {
			\node[draw, rectangle, rounded corners=3pt, minimum width=1cm, minimum height=1cm, inner sep=2pt] (good\i) at (\i,0) {\label};
		}
		
		\node[font=\footnotesize\bfseries] at ([shift={(-45pt,10pt)}]good1.north) {Agent 1};
		\foreach \i/\val in {1/7,2/7,3/4,4/3,5/3,6/1,7/1} {
			\node[above=4pt, font=\scriptsize] at (good\i.north) {$\nicefrac{\val}{13}$};
		}
		
		\node[font=\footnotesize\bfseries] at ([shift={(-45pt,-13pt)}]good1.south) {Agent 2};
		\foreach \i/\val in {1/8,2/5,3/5,4/3,5/2,6/2,7/1} {
			\node[below=4pt, font=\scriptsize] at (good\i.south) {$\nicefrac{\val}{13}$};
		}
		
		\draw[fill=blue!20, fill opacity=0.3, draw=blue!60, thick, rounded corners=4pt]
		(0.5,0.9) -- (6.45,0.9) -- (6.45,-0.45) -- (5.55,-0.45) -- (5.55,0.45) -- (1.45,0.45) -- (1.45,-0.45) -- (0.5,-0.45) -- (0.5,0.45) -- cycle;
		
		\draw[fill=red!20, fill opacity=0.3, draw=red!60, thick, rounded corners=4pt]
		(0.5,-0.9) -- (5.45,-0.9) -- (5.45,0.45) -- (4.55,0.45) -- (4.55,-0.45) -- (1.45,-0.45) -- (1.45,0.45) -- (0.5,0.45) -- (0.5,-0.45) -- cycle;
	\end{tikzpicture}
	\caption{Comparison of classical and our modified reductions. Red boxes indicate the bundle chosen by the classical reductions, while blue boxes indicate the bundle chosen by our reductions.}
\end{figure}

\newpage

\begin{example}[Primary Reductions]
	Consider an instance with $5$ agents and $17$ goods. We illustrate the primary reductions on this instance step by step. In this example, we consider $\alpha = \nicefrac{10}{13}$. The valuation functions are given in the table below: 
	\footnote{Given $\alpha = \nicefrac{10}{13}$ \textsf{(i) } $\nicefrac{13}{17} < \alpha < \nicefrac{14}{17}$, \textsf{(ii) } $\nicefrac{14}{19} < \alpha < \nicefrac{15}{19}$, \textsf{(iii) } $\nicefrac{16}{21} < \alpha < \nicefrac{17}{21}$.}
	\[
	\hspace{-0.7cm}
	\begin{array}{c|ccccccccccccccccc}
		\text{Good} & g_1 & g_2 & g_3 & g_4 & g_5 & g_6 & g_7 & g_8 & g_9 & g_{10} & g_{11} & g_{12} & g_{13} & g_{14} & g_{15} & g_{16} & g_{17}\\
		\hline
		\text{Agent 1} & \partitionA{9}{13} & \partitionB{8}{13} & \partitionC{7}{13} & \partitionE{6}{13} & \partitionD{5}{13} & \partitionA{4}{13} & \partitionB{4}{13} & \partitionC{4}{13} & \partitionD{4}{13} & \partitionD{4}{13} & \partitionC{2}{13} & \partitionE{2}{13} & \partitionE{2}{13} & \partitionE{1}{13} & \partitionE{1}{13} & \partitionE{1}{13} & \partitionB{1}{13}\\
		\text{Agent 2} & \partitionA{9}{17} & \partitionB{9}{17} & \partitionA{8}{17} & \partitionB{8}{17} & \partitionC{8}{17} & \partitionC{5}{17} & \partitionC{4}{17} & \partitionD{4}{17} & \partitionD{4}{17} & \partitionE{4}{17} & \partitionE{4}{17} & \partitionD{3}{17} & \partitionD{3}{17} & \partitionD{3}{17} & \partitionE{3}{17} & \partitionE{3}{17} & \partitionE{3}{17}\\
		\text{Agent 3} & \partitionA{10}{19} & \partitionB{10}{19} & \partitionA{9}{19} & \partitionB{9}{19} & \partitionC{9}{19} & \partitionC{5}{19} & \partitionC{5}{19} & \partitionD{4}{19} & \partitionD{4}{19} & \partitionD{4}{19} & \partitionD{4}{19} & \partitionE{4}{19} & \partitionE{4}{19} & \partitionE{4}{19} & \partitionE{4}{19} & \partitionD{3}{19} & \partitionE{3}{19}\\ 
		\text{Agent 4} & \partitionA{11}{21}& \partitionB{11}{21} & \partitionC{11}{21} & \partitionD{11}{21} & \partitionE{11}{21} & \partitionA{5}{21} & \partitionA{5}{21} & \partitionB{5}{21} & \partitionB{5}{21} & \partitionC{5}{21} & \partitionC{5}{21} & \partitionD{4}{21} & \partitionD{4}{21} & \partitionE{4}{21} & \partitionE{3}{21} & \partitionE{3}{21} & \partitionD{2}{21}\\ 
		\text{Agent 5} & \partitionA{7}{13} & \partitionB{7}{13} & \partitionC{5}{13} & \partitionC{5}{13} & \partitionD{5}{13} & \partitionE{4}{13} & \partitionA{3}{13} & \partitionA{3}{13} & \partitionB{3}{13} & \partitionB{3}{13} & \partitionC{3}{13} & \partitionD{3}{13} & \partitionD{3}{13} & \partitionE{3}{13} & \partitionE{3}{13} & \partitionE{3}{13} & \partitionD{2}{13} \\
	\end{array}
	\]

Initially, rule $\reductiontype^1$ is not applicable since no one values the bundle $\{g_5, g_6\}$ at a value of at least $\alpha$. However, rule $\reductiontype^2$ is applicable, and the modified version of this rule selects the goods $\{g_9, g_{10}, g_{13}\}$. Note that the bundle is not allocated and is only considered in the matching.

\begin{figure}[h!]
	\centering
	
	\begin{tikzpicture}[good node/.style={draw, rectangle, rounded corners=3pt,
			minimum width=0.8cm, minimum height=0.8cm}]
		\foreach \i/\label in {1/$g_1$,2/$g_2$,3/$g_3$,4/$g_4$,5/$g_5$,
			6/$g_6$,7/$g_7$,8/$g_8$,9/$g_9$,10/$g_{10}$,11/$g_{11}$,12/$g_{12}$,13/$g_{13}$, 14/$g_{14}$, 15/$g_{15}$, 16/$g_{16}$,17/$g_{17}$} {
			\node [good node] (good\i) at (\i,0) {\label};
		}

		\node[good node, fill=cyan!30] at (9,0) {$g_9$};
		\node[good node, fill=cyan!30] at (10,0) {$g_{10}$};
		\node[good node, fill=cyan!30] at (13,0) {$g_{13}$};
		
		\draw[dashed] (5.5,-0.6) -- (5.5,0.6);
		\draw[dashed] (10.5,-0.6) -- (10.5,0.6);
		\draw[dashed] (15.5,-0.6) -- (15.5,0.6);
	\end{tikzpicture}
	
\end{figure}

In the second step, we consider the following reductions: $\reductiontype^1$ with the bundle $\{g_4, g_5\}$, $\reductiontype^2$ with $\{g_7, g_8, g_{11}\}$, and $\rstar^1$ with $\{g_1, g_{11}\}$. Rule $\reductiontype^1$ is applicable, and adding the bundle ${g_4, g_5}$ results in a matching of size two. However, we cannot select the bundle ${g_4, g_6}$ since, although Agent 1 still values this bundle, no matching of size two exists for it. \footnote{A saturating matching is shown with blue edges; otherwise, the Hall-violating set is marked in yellow.}

\begin{figure}[h!]
	\centering
	
	\begin{tikzpicture}[good node/.style={draw, rectangle, rounded corners=3pt,
			minimum width=0.8cm, minimum height=0.8cm}]
		\foreach \i/\label in {1/$g_1$,2/$g_2$,3/$g_3$,4/$g_4$,5/$g_5$,
			6/$g_6$,7/$g_7$,8/$g_8$,9/$g_{11}$,10/$g_{12}$,11/$g_{14}$, 12/$g_{15}$, 13/$g_{16}$,14/$g_{17}$} {
			\node [good node] (good\i) at (\i,0) {\label};
		}

		\node[good node, fill=cyan!30] at (4,0) {$g_4$};
		\node[good node, fill=cyan!30] at (5,0) {$g_5$};
		
		\draw[dashed] (4.5,-0.6) -- (4.5,0.6);
		\draw[dashed] (8.5,-0.6) -- (8.5,0.6);
		\draw[dashed] (12.5,-0.6) -- (12.5,0.6);
	\end{tikzpicture}
	
\end{figure}

\begin{figure}[h!]
	\centering
	\begin{subfigure}[t]{0.31\textwidth}
		\centering
		\resizebox{\textwidth}{!}{%
\begin{tikzpicture}[vertex/.style={draw, rectangle, rounded corners=3pt, font=\small}, top vertex/.style={vertex, minimum width=12mm, minimum height=8mm}, bottom vertex/.style={vertex, minimum width=8mm, minimum height=8mm}]
  \node[top vertex, fill=white] (t1) at (-1.00,1.50) {$g_{9}, g_{10}, g_{13}$};
  \node[top vertex, fill=white] (t2) at (1.00,1.50) {$g_{4}, g_{5}$};
  \node[bottom vertex, fill=white] (b1) at (-3.00,-1.50) {1};
  \node[bottom vertex, fill=white] (b2) at (-1.50,-1.50) {2};
  \node[bottom vertex, fill=white] (b3) at (0.00,-1.50) {3};
  \node[bottom vertex, fill=white] (b4) at (1.50,-1.50) {4};
  \node[bottom vertex, fill=white] (b5) at (3.00,-1.50) {5};
  \draw[blue, ultra thick] (t1) -- (b1);
  \draw (t2) -- (b1);
  \draw (t2) -- (b2);
  \draw (t2) -- (b3);
  \draw (t2) -- (b4);
  \draw[blue, ultra thick] (t2) -- (b5);
\end{tikzpicture}
		}
		\caption{Rule $\reductiontype^1$ is applicable.}
	\end{subfigure}
	\hspace{1cm} 
	\begin{subfigure}[t]{0.31\textwidth}
		\centering
		\resizebox{\textwidth}{!}{%
\begin{tikzpicture}[vertex/.style={draw, rectangle, rounded corners=3pt, font=\small}, top vertex/.style={vertex, minimum width=12mm, minimum height=8mm}, bottom vertex/.style={vertex, minimum width=8mm, minimum height=8mm}]
  \node[top vertex, fill=yellow!30] (t1) at (-1.00,1.50) {$g_{9}, g_{10}, g_{13}$};
  \node[top vertex, fill=yellow!30] (t2) at (1.00,1.50) {$g_{4}, g_{6}$};
  \node[bottom vertex, fill=yellow!30] (b1) at (-3.00,-1.50) {1};
  \node[bottom vertex, fill=white] (b2) at (-1.50,-1.50) {2};
  \node[bottom vertex, fill=white] (b3) at (0.00,-1.50) {3};
  \node[bottom vertex, fill=white] (b4) at (1.50,-1.50) {4};
  \node[bottom vertex, fill=white] (b5) at (3.00,-1.50) {5};
  \draw (t1) -- (b1);
  \draw (t2) -- (b1);
\end{tikzpicture}
		}
		\captionsetup{width=1.2\textwidth} 
		\caption{Bundle $\{g_4, g_6\}$ cannot be selected.}
	\end{subfigure}
\end{figure}

In the third step, we consider the following reductions: $\reductiontype_1$ with the bundle $\{g_3, g_6\}$, $\reductiontype_2$ with $\{g_7, g_8, g_{11}\}$, and $\rstar^1$ with $\{g_1, g_{11}\}$. The rules $\reductiontype^1$ and $\reductiontype^2$ are not applicable, and the matching does not exist for these rules. However, $\rstar^1$ is applicable and can select the goods $g_1$ and $g_{16}$. For the bundle $\{g_1, g_{17}\}$, no matching exists.

\begin{figure}[H]
	\centering
	
	\begin{tikzpicture}[good node/.style={draw, rectangle, rounded corners=3pt,
			minimum width=0.8cm, minimum height=0.8cm}]
		\foreach \i/\label in {1/$g_1$,2/$g_2$,3/$g_3$,
			4/$g_6$,5/$g_7$,6/$g_8$,7/$g_{11}$,8/$g_{12}$,9/$g_{14}$, 10/$g_{15}$, 11/$g_{16}$,12/$g_{17}$} {
			\node [good node] (good\i) at (\i,0) {\label};
		}

		\node[good node, fill=cyan!30] at (1,0) {$g_1$};
		\node[good node, fill=cyan!30] at (11,0) {$g_{16}$};
		
		\draw[dashed] (3.5,-0.6) -- (3.5,0.6);
		\draw[dashed] (6.5,-0.6) -- (6.5,0.6);
		\draw[dashed] (9.5,-0.6) -- (9.5,0.6);
	\end{tikzpicture}
	
\end{figure}

\begin{figure}[H]
	\centering
	\begin{subfigure}[t]{0.31\textwidth}
		\centering
		\resizebox{\textwidth}{!}{%
\begin{tikzpicture}[vertex/.style={draw, rectangle, rounded corners=3pt, font=\small}, top vertex/.style={vertex, minimum width=12mm, minimum height=8mm}, bottom vertex/.style={vertex, minimum width=8mm, minimum height=8mm}]
  \node[top vertex, fill=yellow!30] (t1) at (-2.00,1.50) {$g_{9}, g_{10}, g_{13}$};
  \node[top vertex, fill=white] (t2) at (0.00,1.50) {$g_{4}, g_{5}$};
  \node[top vertex, fill=yellow!30] (t3) at (2.00,1.50) {$g_{3}, g_{6}$};
  \node[bottom vertex, fill=yellow!30] (b1) at (-3.00,-1.50) {1};
  \node[bottom vertex, fill=white] (b2) at (-1.50,-1.50) {2};
  \node[bottom vertex, fill=white] (b3) at (0.00,-1.50) {3};
  \node[bottom vertex, fill=white] (b4) at (1.50,-1.50) {4};
  \node[bottom vertex, fill=white] (b5) at (3.00,-1.50) {5};
  \draw (t1) -- (b1);
  \draw (t2) -- (b1);
  \draw (t2) -- (b2);
  \draw (t2) -- (b3);
  \draw (t2) -- (b4);
  \draw (t2) -- (b5);
  \draw (t3) -- (b1);
\end{tikzpicture}
		}
		\caption{Rule $\reductiontype^1$ is not applicable.}
	\end{subfigure}
	\hspace{1cm} 
	\begin{subfigure}[t]{0.31\textwidth}
		\centering
		\resizebox{\textwidth}{!}{%
\begin{tikzpicture}[vertex/.style={draw, rectangle, rounded corners=3pt, font=\small}, top vertex/.style={vertex, minimum width=12mm, minimum height=8mm}, bottom vertex/.style={vertex, minimum width=8mm, minimum height=8mm}]
  \node[top vertex, fill=yellow!30] (t1) at (-2.00,1.50) {$g_{9}, g_{10}, g_{13}$};
  \node[top vertex, fill=white] (t2) at (0.00,1.50) {$g_{4}, g_{5}$};
  \node[top vertex, fill=yellow!30] (t3) at (2.00,1.50) {$g_{7}, g_{8}, g_{11}$};
  \node[bottom vertex, fill=yellow!30] (b1) at (-3.00,-1.50) {1};
  \node[bottom vertex, fill=white] (b2) at (-1.50,-1.50) {2};
  \node[bottom vertex, fill=white] (b3) at (0.00,-1.50) {3};
  \node[bottom vertex, fill=white] (b4) at (1.50,-1.50) {4};
  \node[bottom vertex, fill=white] (b5) at (3.00,-1.50) {5};
  \draw (t1) -- (b1);
  \draw (t2) -- (b1);
  \draw (t2) -- (b2);
  \draw (t2) -- (b3);
  \draw (t2) -- (b4);
  \draw (t2) -- (b5);
  \draw (t3) -- (b1);
\end{tikzpicture}
		}
		\caption{Rule $\reductiontype^2$ is not applicable.}
	\end{subfigure}
\end{figure}

\begin{figure}[H]
	\centering
	\begin{subfigure}[t]{0.31\textwidth}
		\centering
		\resizebox{\textwidth}{!}{%
\begin{tikzpicture}[vertex/.style={draw, rectangle, rounded corners=3pt, font=\small}, top vertex/.style={vertex, minimum width=12mm, minimum height=8mm}, bottom vertex/.style={vertex, minimum width=8mm, minimum height=8mm}]
  \node[top vertex, fill=white] (t1) at (-2.00,1.50) {$g_{9}, g_{10}, g_{13}$};
  \node[top vertex, fill=white] (t2) at (0.00,1.50) {$g_{4}, g_{5}$};
  \node[top vertex, fill=white] (t3) at (2.00,1.50) {$g_{1}, g_{11}$};
  \node[bottom vertex, fill=white] (b1) at (-3.00,-1.50) {1};
  \node[bottom vertex, fill=white] (b2) at (-1.50,-1.50) {2};
  \node[bottom vertex, fill=white] (b3) at (0.00,-1.50) {3};
  \node[bottom vertex, fill=white] (b4) at (1.50,-1.50) {4};
  \node[bottom vertex, fill=white] (b5) at (3.00,-1.50) {5};
  \draw[blue, ultra thick] (t1) -- (b1);
  \draw (t2) -- (b1);
  \draw (t2) -- (b2);
  \draw (t2) -- (b3);
  \draw[blue, ultra thick] (t2) -- (b4);
  \draw (t2) -- (b5);
  \draw (t3) -- (b1);
  \draw[blue, ultra thick] (t3) -- (b5);
\end{tikzpicture}
		}
		\caption{Rule $\rstar^1$ is applicable.}
	\end{subfigure}
	\hfill
	\begin{subfigure}[t]{0.31\textwidth}
		\centering
		\resizebox{\textwidth}{!}{%
\begin{tikzpicture}[vertex/.style={draw, rectangle, rounded corners=3pt, font=\small}, top vertex/.style={vertex, minimum width=12mm, minimum height=8mm}, bottom vertex/.style={vertex, minimum width=8mm, minimum height=8mm}]
  \node[top vertex, fill=white] (t1) at (-2.00,1.50) {$g_{9}, g_{10}, g_{13}$};
  \node[top vertex, fill=white] (t2) at (0.00,1.50) {$g_{4}, g_{5}$};
  \node[top vertex, fill=white] (t3) at (2.00,1.50) {$g_{1}, g_{16}$};
  \node[bottom vertex, fill=white] (b1) at (-3.00,-1.50) {1};
  \node[bottom vertex, fill=white] (b2) at (-1.50,-1.50) {2};
  \node[bottom vertex, fill=white] (b3) at (0.00,-1.50) {3};
  \node[bottom vertex, fill=white] (b4) at (1.50,-1.50) {4};
  \node[bottom vertex, fill=white] (b5) at (3.00,-1.50) {5};
  \draw[blue, ultra thick] (t1) -- (b1);
  \draw (t2) -- (b1);
  \draw (t2) -- (b2);
  \draw (t2) -- (b3);
  \draw[blue, ultra thick] (t2) -- (b4);
  \draw (t2) -- (b5);
  \draw (t3) -- (b1);
  \draw[blue, ultra thick] (t3) -- (b5);
\end{tikzpicture}
		}
		\captionsetup{width=1.2\textwidth} 
		\caption{Bundle $\{g_1, g_{16}\}$ is valid.}
	\end{subfigure}
	\hfill
	\begin{subfigure}[t]{0.31\textwidth}
		\centering
		\resizebox{\textwidth}{!}{%
\begin{tikzpicture}[vertex/.style={draw, rectangle, rounded corners=3pt, font=\small}, top vertex/.style={vertex, minimum width=12mm, minimum height=8mm}, bottom vertex/.style={vertex, minimum width=8mm, minimum height=8mm}]
  \node[top vertex, fill=yellow!30] (t1) at (-2.00,1.50) {$g_{9}, g_{10}, g_{13}$};
  \node[top vertex, fill=white] (t2) at (0.00,1.50) {$g_{4}, g_{5}$};
  \node[top vertex, fill=yellow!30] (t3) at (2.00,1.50) {$g_{1}, g_{17}$};
  \node[bottom vertex, fill=yellow!30] (b1) at (-3.00,-1.50) {1};
  \node[bottom vertex, fill=white] (b2) at (-1.50,-1.50) {2};
  \node[bottom vertex, fill=white] (b3) at (0.00,-1.50) {3};
  \node[bottom vertex, fill=white] (b4) at (1.50,-1.50) {4};
  \node[bottom vertex, fill=white] (b5) at (3.00,-1.50) {5};
  \draw (t1) -- (b1);
  \draw (t2) -- (b1);
  \draw (t2) -- (b2);
  \draw (t2) -- (b3);
  \draw (t2) -- (b4);
  \draw (t2) -- (b5);
  \draw (t3) -- (b1);
\end{tikzpicture}
		}
		\captionsetup{width=1.2\textwidth} 
		\caption{Bundle $\{g_1, g_{17}\}$ cannot be selected.}
	\end{subfigure}
\end{figure}

In the final step, none of the remaining reduction rules is applicable, so we complete the primary reduction phase. At this stage, we divide all agents into two groups, $N^g$ and $N_r$, based on their valuation of the good $g_8$, where $N^g$ consists of agents who value $g_8$ at least $\nicefrac{3}{13}$, and $N^r$ includes the remaining agents. In this case, $N^g = \{1, 2, 4, 5\}$ and $N^r = \{3\}$ and $|N^g| \ge \nicefrac{n}{\sqrt{2}}$. We then select a matching that maximizes the number of matched agents from $N^r$, and prioritize them in both primary reductions and the subsequent steps of the algorithm.

\begin{figure}[H]
	\centering
	
	\begin{tikzpicture}[good node/.style={draw, rectangle, rounded corners=3pt,
			minimum width=0.8cm, minimum height=0.8cm}]
		\foreach \i/\label in {1/$g_2$,2/$g_3$,
			3/$g_6$,4/$g_7$,5/$g_8$,6/$g_{11}$,7/$g_{12}$,8/$g_{14}$, 9/$g_{15}$, 10/$g_{17}$} {
			\node [good node] (good\i) at (\i,0) {\label};
		}
		
		\draw[dashed] (2.5,-0.6) -- (2.5,0.6);
		\draw[dashed] (4.5,-0.6) -- (4.5,0.6);
		\draw[dashed] (6.5,-0.6) -- (6.5,0.6);
	\end{tikzpicture}
	
\end{figure}

\begin{figure}[H]
	\centering
	\begin{subfigure}[t]{0.31\textwidth}
		\centering
		\resizebox{\textwidth}{!}{%
\begin{tikzpicture}[vertex/.style={draw, rectangle, rounded corners=3pt, font=\small}, top vertex/.style={vertex, minimum width=12mm, minimum height=8mm}, bottom vertex/.style={vertex, minimum width=8mm, minimum height=8mm}]
  \node[top vertex, fill=yellow!30] (t1) at (-3.00,1.50) {$g_{9}, g_{10}, g_{13}$};
  \node[top vertex, fill=white] (t2) at (-1.00,1.50) {$g_{4}, g_{5}$};
  \node[top vertex, fill=white] (t3) at (1.00,1.50) {$g_{1}, g_{16}$};
  \node[top vertex, fill=yellow!30] (t4) at (3.00,1.50) {$g_{3}, g_{6}$};
  \node[bottom vertex, fill=yellow!30] (b1) at (-3.00,-1.50) {1};
  \node[bottom vertex, fill=white] (b2) at (-1.50,-1.50) {2};
  \node[bottom vertex, fill=white] (b3) at (0.00,-1.50) {3};
  \node[bottom vertex, fill=white] (b4) at (1.50,-1.50) {4};
  \node[bottom vertex, fill=white] (b5) at (3.00,-1.50) {5};
  \draw (t1) -- (b1);
  \draw (t2) -- (b1);
  \draw (t2) -- (b2);
  \draw (t2) -- (b3);
  \draw (t2) -- (b4);
  \draw (t2) -- (b5);
  \draw (t3) -- (b1);
  \draw (t3) -- (b5);
  \draw (t4) -- (b1);
\end{tikzpicture}
		}
		\caption{Rule $\reductiontype^1$ is not applicable.}
	\end{subfigure}
	\hfill
	\begin{subfigure}[t]{0.31\textwidth}
		\centering
		\resizebox{\textwidth}{!}{%
\begin{tikzpicture}[vertex/.style={draw, rectangle, rounded corners=3pt, font=\small}, top vertex/.style={vertex, minimum width=12mm, minimum height=8mm}, bottom vertex/.style={vertex, minimum width=8mm, minimum height=8mm}]
  \node[top vertex, fill=yellow!30] (t1) at (-3.00,1.50) {$g_{9}, g_{10}, g_{13}$};
  \node[top vertex, fill=white] (t2) at (-1.00,1.50) {$g_{4}, g_{5}$};
  \node[top vertex, fill=white] (t3) at (1.00,1.50) {$g_{1}, g_{16}$};
  \node[top vertex, fill=yellow!30] (t4) at (3.00,1.50) {$g_{6}, g_{7}, g_{8}$};
  \node[bottom vertex, fill=yellow!30] (b1) at (-3.00,-1.50) {1};
  \node[bottom vertex, fill=white] (b2) at (-1.50,-1.50) {2};
  \node[bottom vertex, fill=white] (b3) at (0.00,-1.50) {3};
  \node[bottom vertex, fill=white] (b4) at (1.50,-1.50) {4};
  \node[bottom vertex, fill=white] (b5) at (3.00,-1.50) {5};
  \draw (t1) -- (b1);
  \draw (t2) -- (b1);
  \draw (t2) -- (b2);
  \draw (t2) -- (b3);
  \draw (t2) -- (b4);
  \draw (t2) -- (b5);
  \draw (t3) -- (b1);
  \draw (t3) -- (b5);
  \draw (t4) -- (b1);
\end{tikzpicture}
		}
		\caption{Rule $\reductiontype^2$ is not applicable.}
	\end{subfigure}
	\hfill
	\begin{subfigure}[t]{0.31\textwidth}
		\centering
		\resizebox{\textwidth}{!}{%
\begin{tikzpicture}[vertex/.style={draw, rectangle, rounded corners=3pt, font=\small}, top vertex/.style={vertex, minimum width=12mm, minimum height=8mm}, bottom vertex/.style={vertex, minimum width=8mm, minimum height=8mm}]
  \node[top vertex, fill=yellow!30] (t1) at (-3.00,1.50) {$g_{9}, g_{10}, g_{13}$};
  \node[top vertex, fill=white] (t2) at (-1.00,1.50) {$g_{4}, g_{5}$};
  \node[top vertex, fill=yellow!30] (t3) at (1.00,1.50) {$g_{1}, g_{16}$};
  \node[top vertex, fill=yellow!30] (t4) at (3.00,1.50) {$g_{2}, g_{8}$};
  \node[bottom vertex, fill=yellow!30] (b1) at (-3.00,-1.50) {1};
  \node[bottom vertex, fill=white] (b2) at (-1.50,-1.50) {2};
  \node[bottom vertex, fill=white] (b3) at (0.00,-1.50) {3};
  \node[bottom vertex, fill=white] (b4) at (1.50,-1.50) {4};
  \node[bottom vertex, fill=yellow!30] (b5) at (3.00,-1.50) {5};
  \draw (t1) -- (b1);
  \draw (t2) -- (b1);
  \draw (t2) -- (b2);
  \draw (t2) -- (b3);
  \draw (t2) -- (b4);
  \draw (t2) -- (b5);
  \draw (t3) -- (b1);
  \draw (t3) -- (b5);
  \draw (t4) -- (b1);
  \draw (t4) -- (b5);
\end{tikzpicture}
		}
		\caption{Rule $\rstar^1$ is not applicable.}
	\end{subfigure}
\end{figure}

\begin{figure}[H]
	\centering
	\begin{subfigure}[t]{0.45\textwidth}
		\centering
		\resizebox{\textwidth}{!}{%
\begin{tikzpicture}[vertex/.style={draw, rectangle, rounded corners=3pt, font=\small}, top vertex/.style={vertex, minimum width=12mm, minimum height=8mm}, bottom vertex/.style={vertex, minimum width=8mm, minimum height=8mm}]
  \node[top vertex, fill=white] (t1) at (-2.00,1.50) {$g_{9}, g_{10}, g_{13}$};
  \node[top vertex, fill=white] (t2) at (0.00,1.50) {$g_{4}, g_{5}$};
  \node[top vertex, fill=white] (t3) at (2.00,1.50) {$g_{1}, g_{16}$};
  \node[bottom vertex, fill=green!30] (b1) at (-3.00,-1.50) {1};
  \node[bottom vertex, fill=green!30] (b2) at (-1.50,-1.50) {2};
  \node[bottom vertex, fill=red!30] (b3) at (0.00,-1.50) {3};
  \node[bottom vertex, fill=green!30] (b4) at (1.50,-1.50) {4};
  \node[bottom vertex, fill=green!30] (b5) at (3.00,-1.50) {5};
  \draw[blue, ultra thick] (t1) -- (b1);
  \draw (t2) -- (b1);
  \draw (t2) -- (b2);
  \draw[blue, ultra thick] (t2) -- (b3);
  \draw (t2) -- (b4);
  \draw (t2) -- (b5);
  \draw (t3) -- (b1);
  \draw[blue, ultra thick] (t3) -- (b5);
\end{tikzpicture}
		}
	\end{subfigure}
	\caption{The final result of primary reductions.}
\end{figure}

\end{example}


		\section{Bounds on $\MMS$ Values for Calibrated Valuations} \label{apx:adjusting}

\begin{lemma}\label{lem:virt-mms-before}
	Let $\prgoods$ be a set of goods, $d$ be a constant, and let $\prvalu$ be a valuation function such that $\mms^d_{\prvalu}(\prgoods) \ge 1$, and for all $\prgood \in \prgoods$ we have $\prvalu(\{\prgood\})\le 1$. Then for every $0\le\lambda \le  \tfrac{4\alpha}{3}-1$ we have $\prauxmms{\LM}{\prgoods}{d}  \ge 1 - 3\lambda.$
\end{lemma}
\begin{proof}
	Since $\mms^d_{\prvalu}(\prgoods) \geq 1$, we can partition $\prgoods$ into $(P_1, \dots, P_d)$  such that each subset $P_j$ satisfies $\prvalu(P_j) \geq 1$. It suffices to show that  for every $1\leq j \leq d$,  
	$
	\praux{\LM}{P_j} \geq 1 - 3\lambda,
	$
	which directly implies $\prauxmms{\LM}{\hat M}{d} \geq 1 - 3\lambda$.
	Let $S \;=\;\{\,\prgood\in P_j\ \mid\,\prvalu(\{\prgood\})\ge\tfrac{\alpha}{3}-\lambda\}.$\\
	If ${\lvert S\rvert \ge4}$, we have 
	\begin{align*}
		\praux{\LM}{P_j} 
		&\;\ge\; 4\cdot\biggl(\frac{\alpha}{3} - \lambda\biggr) \\
		&\;=\; \frac{4\alpha}{3} - 4\lambda \\
		&\;\ge\; 1 - 3\lambda  & \lambda \le \frac{4\alpha}{3} - 1.
	\end{align*}
	Therefore, assume  ${\lvert S\rvert \le3}$. Note that for every good $\prgood \in P_j \setminus S$, we have $\prvalu(\{\prgood\})<\tfrac{\alpha}{3}-\lambda$, and thus $\LM(\prvalu(\{\prgood\}))=\prvalu(\{\prgood\})$. We consider two cases.
	\begin{itemize}
		\item
		\textbf{At least one good \( \prgood \) in $S$ has value at least $1 - \tfrac{\alpha}{3} - \tfrac{\lambda}{2}$:} 
		In particular, for this good, we have
		$
		\LM(\prvalu(\{\prgood\}))=\max(1 - \tfrac{\alpha}{3} - 2\lambda, \prvalu(\{\prgood\}) - 3\lambda)
		$,  
		therefore, the transformation \(\LM\) reduces the original value by at most  \( 3\lambda \).  
		Now If \( \prgood \) is the only good in \( S \), we get \( \praux{\LM}{P_j} \geq 1 - 3\lambda \).  Otherwise, if there is another good in \( P_j \) with value at least \( \tfrac{\alpha}{3} - \lambda \), then combining both goods ensures:  
		\begin{align*}
			\praux{\LM}{P_j} &\geq \Bigl(1 - \frac{\alpha}{3} - 2\lambda\Bigr) + \Bigl(\frac{\alpha}{3} - \lambda\Bigr) \\
			&= 1 - 3\lambda.
		\end{align*}
		
		Thus, in both cases, we have \( \praux{\LM}{P_j} \geq 1 - 3\lambda \).
		\item
		\textbf{All goods in $S$ have values below \( 1 - \tfrac{\alpha}{3} - \tfrac{\lambda}{2} \):} 	
		If there are at most two such goods in \( S \), the transformation \(\LM\) reduces their original values by at most  \( \tfrac{3\lambda}{2} \).
		Therefore the calibrated value satisfies \( \praux{\LM}{P_j} \geq 1 - 3\lambda \).  
		Now, suppose there are exactly three such goods with value at least  $\tfrac{\alpha}{3} - \lambda$. If at least one of them has a value of at least \( 1 - \tfrac{2\alpha}{3} \), then grouping it with the other two ensures:  
		\begin{align*}
			\praux{\LM}{P_j} 
			&\geq \Bigl(1 - \frac{2\alpha}{3} - \lambda\Bigr) + 2\Bigl(\frac{\alpha}{3} - \lambda\Bigr) \\
			&= 1 - 3\lambda.
		\end{align*}

		Otherwise, if all three goods have values below \( 1 - \tfrac{2\alpha}{3} \), the transformation \(\LM\) reduces their original values by at most \( \lambda \), leading to a total loss of at most \( 3\lambda \), which again guarantees \( \praux{\LM}{P_j} \geq 1 - 3\lambda \).  
	\end{itemize}
	Thus, in all cases, the bound holds.
\end{proof}

\begin{lemma}\label{lem:virtG-mms}
	Let $\prvalu$ be a valuation function on $\prgoods$ with $\mms^d_{\prvalu}(\prgoods)\ge 4(1-\alpha)$ and for all $\prgood \in \prgoods$ we have $\prvalu(\{\prgood\})\le 1$. Then $\prauxmms{\vG}{\prgoods}{d} \;\;\ge\; 4(2-\tfrac{7\alpha}{3}).$
\end{lemma}
\begin{proof}
	Since $\mms^d_{\prvalu}(\prgoods) \geq 4(1-\alpha)$, there exists a partition $(P_1, \dots, P_d)$ of $\prgoods$ such that $\prvalu(P_j) \geq 4(1-\alpha)$ for each $j$. We aim to show that for every $P_j$, 
	$
	\praux{\vG}{P_j} \ge  4(2-\tfrac{7\alpha}{3})
	$
	which directly implies $\prauxmms{\vG}{\prgoods}{d} \geq 4(2-\tfrac{7\alpha}{3})$.
	Let $S \;=\;\{\,\prgood\in P_j\ \mid\,\prvalu(\{\prgood\})\ge2-\tfrac{7\alpha}{3}\}.$\\
	Note that if $\boldsymbol{\lvert S\rvert \ge4}$, we have:
	$
		\praux{\vG}{P_j} 
		\ge 4(2-\frac{7\alpha}{3}).
	$
	For the case that  $\boldsymbol{\lvert S\rvert \le3}$, by definition of $S$, there are at most three goods in $P_j$ with $\prvalu(\{\prgood\})\ge 2-\tfrac{7\alpha}{3}$, all other goods in $P_j$ have $\prvalu(\{\prgood\})<2-\tfrac{7\alpha}{3}$, and thus $\vG(\prvalu(\{\prgood\}))=\prvalu(\{\prgood\})$. If $\boldsymbol{\lvert S\rvert \le2}$, then the transformation \(\vG\) reduces their original values by at most $\tfrac{8\alpha}{3}-2$, therefore
	\begin{align*}
		\sum_{\prgood\in P_j} \praux{\vG}{\{\prgood\}}
		&\ge \sum_{\prgood\in P_j} v(\{\prgood\}) - 2\Bigl(\frac{8\alpha}{3}-2\Bigr) \\
		&\ge 4(1-\alpha) - 4\Bigl(\frac{4\alpha}{3}-1\Bigr) \\
		&=4\Bigl(2-\frac{7\alpha}{3}\Bigr).
	\end{align*}

	For $\boldsymbol{\lvert S\rvert =3}$ we consider two cases.
	\begin{itemize}
		\item
		\textbf{At most one good in $S$ has value at least $2-\tfrac{13\alpha}{6}$:}  
		Then by definition of $\vG$,the transformation \(\vG\) reduces the original value of one good by at most $\tfrac{8\alpha}{3}-2$, and two goods by at most $\tfrac{4\alpha}{3}-1$. Hence:
		\begin{align*}
			\sum_{\prgood\in P_j} \praux{\vG}{\{\prgood\}}
			&\ge \sum_{g\in P_j} v(\{\prgood\}) - \Bigl(\frac{8\alpha}{3}-2\Bigr)-2\Bigl(\frac{4\alpha}{3}-1\Bigr)\\
			&\ge 4(1-\alpha) - 4\Bigl(\frac{4\alpha}{3}-1\Bigr) \\
			&=4\Bigl(2-\frac{7\alpha}{3}\Bigr).
		\end{align*}
		
		\item
		\textbf{At least two goods in $S$ have value at least $2-\tfrac{13\alpha}{6}$:} 		
		Therefore there are two goods with value at least  $2-\tfrac{13\alpha}{6}$ and one good with value at least $2-\tfrac{7\alpha}{3}$, ensures: 
		\begin{align*}
			\sum_{\prgood\in P_j} \praux{\vG}{\{\prgood\}}
			&\ge 2\;\vG\Bigl(2-\frac{13\alpha}{6}\Bigr) + \vG\Bigl(2-\frac{7\alpha}{3}\Bigr)\\[1ex]
			&\ge 2\Bigl(3-\frac{7\alpha}{2}\Bigr) + \Bigl(2-\frac{7\alpha}{3}\Bigr) \\
			&= 4\Bigl(2-\frac{7\alpha}{3}\Bigr).
		\end{align*}
		
	\end{itemize}
	Thus $\prauxmms{\vG}{\prgoods}{d} \;\;\ge\; 4(2-\tfrac{7\alpha}{3})$, as desired.
\end{proof}


\begin{lemma}\label{lem:HA-mms-before}
	Let $\prgoods$ be a set of goods, $d$ be a constant, and let $\prvalu$ be a valuation function such that $\mms^d_{\prvalu}(\prgoods) \ge 1$ and for all $\prgood \in \prgoods$ we have $\prvalu(\{\prgood\})\le 1$. Then we have $\prauxmms{\vHA_\lambda}{\prgoods}{d}  \ge 1 - 2\lambda.$
\end{lemma}
\begin{proof}
	Since $\mms^d_{\prvalu}(\prgoods) \geq 1$, there exists a partition $(P_1, \dots, P_d)$ of $\prgoods$ such that $\prvalu(P_j) \geq 1$ for each $1\leq j \leq d$. We want to show that for every $1\leq j \leq d$,  
	$
	\praux{\vHA_\lambda}{P_j} \geq 1 - 2\lambda,
	$
	which directly implies $\prauxmms{\vHA_\lambda}{\prgoods}{d} \geq 1 - 2\lambda$.
	Let $S \;=\;\{\,\prgood\in P_j\ \mid\, \prvalu(\{\prgood\})\ge \tfrac{1}{2} - \lambda.\}$\\
	We consider two cases:
	\begin{itemize}
		\item $|S| < 2$: By definition of $\vHA_\lambda$, the transformation $\vHA_\lambda$ reduces the original value of one good by at most $2\lambda$. Hence:
		\begin{align*}
			\sum_{\prgood\in P_j} \praux{\vHA_\lambda}{\{\prgood\}}
			&\ge \sum_{\prgood\in P_j} v(\{\prgood\}) - 2\lambda\\
			&\ge 1 - 2\lambda.
		\end{align*}
		
		\item $|S| \ge 2$: In this case for each $\prgood \in S$  we have $\vHA_\lambda(\prvalu(\{\prgood\})) \ge \tfrac{1}{2} - \lambda$. Therefore:
		\begin{align*}
			\sum_{\prgood\in P_j} \praux{\vHA_\lambda}{\{\prgood\}}
			&\ge 2(\frac{1}{2} - \lambda)\\
			&= 1 - 2\lambda.
		\end{align*}
	\end{itemize}
\end{proof}

\begin{lemma}\label{lem:HB-mms-before}
	Let $\prgoods$ be a set of goods, $d$ be a constant, and let $\prvalu$ be a valuation function such that $\mms^d_{\prvalu}(\prgoods) \ge 4(1-\alpha)$ and for all $\prgood \in \prgoods$ we have $\prvalu(\{\prgood\})\le 1$. Then we have $\prauxmms{\vHB_\lambda}{\prgoods}{d}  \ge 4(1-\alpha) - 2\lambda.$
\end{lemma}
\begin{proof}
	Since $\mms^d_{\prvalu}(\prgoods) \geq 4(1-\alpha)$, there exists a partition $(P_1, \dots, P_d)$ of $\prgoods$ such that $\prvalu(P_j) \geq 4(1-\alpha)$ for each $1\leq j \leq d$. We want to show that for every $1\leq j \leq d$,  
	$
	\praux{\vHB_\lambda}{P_j} \geq 4(1-\alpha) - 2\lambda,
	$
	which directly implies $\prauxmms{\vHB_\lambda}{\prgoods}{d} \geq 4(1-\alpha) - 2\lambda$.
	Let $S \;=\;\{\,\prgood\in P_j\ \mid\,  \prvalu(\{\prgood\})\ge 2(1-\alpha) - \lambda \}.$\\
	We consider two cases:
	\begin{itemize}
		\item $|S| < 2$: By definition of $\vHB_\lambda$, the transformation $\vHB_\lambda$ reduces the original value of one good by at most $2\lambda$. Hence:
		\begin{align*}
			\sum_{\prgood\in P_j} \praux{\vHB_\lambda}{\{\prgood\}}
			&\ge \sum_{\prgood\in P_j} v(\{\prgood\}) - 2\lambda\\
			&\ge 4(1-\alpha) - 2\lambda.
		\end{align*}
		
		\item $|S| \ge 2$: In this case for each $\prgood \in S$  we have $\vHB_\lambda(\prvalu(\{\prgood\})) \ge 2(1-\alpha) - \lambda$. Therefore:
		\begin{align*}
			\sum_{\prgood\in P_j} \praux{\vHB_\lambda}{\{\prgood\}}
			&\ge 2(2(1-\alpha) - \lambda)\\
			&= 4(1-\alpha) - 2\lambda.
		\end{align*}
	\end{itemize}
\end{proof}
	\end{toappendix}
	
	\begin{abstract}
	The maximin share ($\MMS$) is the most prominent share-based fairness notion in the fair allocation of indivisible goods. Recent years have seen significant efforts to improve the approximation guarantees for $\MMS$ for different valuation classes, particularly for additive valuations. For the additive setting, it has been shown that for some instances, no allocation can guarantee a factor better than $1-\tfrac{1}{n^4}$ of maximin share value to all agents. However, the best currently known algorithm achieves an approximation guarantee of \( \tfrac{3}{4} + \tfrac{3}{3836} \) for $\MMS$.  
	In this work, we narrow this gap and improve the best-known approximation guarantee for $\MMS$ to \( \tfrac{10}{13} \).
\end{abstract}
	{\newpage\onehalfspacing\tableofcontents\newpage}
	\section{Introduction}

Fair allocation is a fundamental problem that spans multiple disciplines, including mathematics, social sciences, economics, and computer science. Given \( m \) goods and \( n \) agents, each agent has a valuation function \( \valu_{i} \) that assigns a non-negative value to every subset of goods. The goal is to allocate the goods fairly. In this paper, we focus on the setting where the valuations are \textit{additive}. 

What does it mean for an allocation to be fair? How can fairness be measured and ensured? These questions have been extensively studied for over eight decades. The foundation of modern fair division theory dates back to Hugo Steinhaus’ seminal work \cite{steinhaus1949division} in 1949, where he provided a mathematically rigorous definition of the \textit{cake-cutting problem}—a fair allocation problem involving a continuous, heterogeneous resource. Since then, numerous fairness criteria have been proposed, which can be broadly classified into two main categories:

\begin{itemize}
\item \textbf{Envy-based}: 	
Agents evaluate fairness by comparing their own bundle to either the entire bundle or a subset of another agent’s bundle. Examples include \emph{envy-freeness}~\cite{foley1966resource}, \emph{envy-freeness up to one good}~\cite{budish2011combinatorial}, and \emph{envy-freeness up to any good}~\cite{caragiannis2019unreasonable}.
	\item \textbf{Share-based}: An agent evaluates fairness based on the value they receive, independently of others' allocations. Examples include \emph{maximin share}~\cite{budish2011combinatorial} and \emph{proportionality}~\cite{steinhaus1949division}.  
\end{itemize}

In this paper, we focus on one of the most well-studied share-based fairness notions in recent years: the \textit{maximin share} (\(\MMS\)) \cite{budish2011combinatorial}.
Suppose  we aim to define a share-based notion of fairness by setting a threshold $\tau_i$ for each agent $\agent_i$ to determine whether their share is fair. A reasonable expectation is \( \tau _i\leq \tfrac{\valu_{i}(\Pgoods)}{n} \), since if all agents have similar valuation functions, guaranteeing a larger value to every agent would be impossible. This quantity, \( \valu_{i}(\Pgoods)/n \), is called the \textit{proportional share} and has been extensively studied in the literature of fair allocation \cite{steinhaus1949division}. 
When goods are \textit{divisible}, proportionality can always be guaranteed \cite{dubins1961cut}. However, with \textit{indivisible} goods, this is no longer the case. Consider a simple example: if there are two agents and one indivisible good, one agent will receive the good while the other gets nothing—far below their proportional share.

A natural alternative is the \textit{maximin share} (\(\MMS\)), which provides a more flexible fairness benchmark. To define \(\MMS\), consider a different way to set an upper bound for \( \tau_i \). We ask agent \( \agent_i \) to divide the goods into \( n \) bundles in a way that maximizes the value of the least valuable bundle. The value of this least valuable bundle is called the \textit{maximin share} (\(\MMS\))  value of agent \( \agent_i \).

By definition, an agent's maximin share value is always at most their proportional share. They coincide when an agent can partition the goods into \( n \) bundles of equal value.  Moreover, \(\MMS\) value serves as an upper bound for \( \tau_i \); if all agents have similar valuations, at least one agent receives a bundle worth at most their maximin share value. This naturally leads to a question:  
	\textit{Can we guarantee that every agent receives a bundle which she values as much as her maximin share value?}

Unfortunately, the answer to this question is negative; there exist instances where no allocation can ensure that every agent receives a bundle with value at least as their maximin share value \cite{procaccia2014fair}. However, unlike proportionality, there always exist allocations that guarantee each agent a constant fraction of their maximin share value. 

Over the past decade, significant efforts have been made to improve approximation guarantees for the maximin share problem in the additive setting  \cite{procaccia2014fair,barman2020approximation,ghodsi2018fair,amanatidis2017approximation,garg2020improved,akrami2023simplification,akrami2024breaking}. A \( (\mmsfrac{1}{2}) \)-approximation guarantee is easy to achieve, but the first nontrivial bound of \( \mmsfrac{2}{3} \) was introduced by Procaccia and Wang \cite{procaccia2014fair}. This was later improved to \( \mmsfrac{3}{4} \) by Ghodsi \etal~\cite{ghodsi2018fair}. Subsequent work slightly improved this bound to $ (\mmsfrac{3}{4} + \mmsfrac{1}{12n})$-$\MMS$ and $(\mmsfrac{3}{4}+\min(\mmsfrac{1}{36},\mmsfrac{3}{16n-4}))$-$\MMS$ \cite{garg2020improved,akrami2023simplification}, but no breakthrough occurred until the recent work of Akrami and Garg \cite{akrami2024breaking}, who improved the approximation guarantee to \( \mmsfrac{3}{4} + \mmsfrac{3}{3836} \). 

Although the result of Akrami and Garg \cite{akrami2024breaking} breaks the \( (\mmsfrac{3}{4}) \)-approximation barrier, the improvement remains small. In this paper, we  advance the frontier by proving that a \( (\mmsfrac{10}{13}) \)-approximation of \(\MMS\) can always be guaranteed for all agents (See \Cref{fig:figure1}).

\begin{figure}
	\centering
	\includegraphics[width=0.75\linewidth]{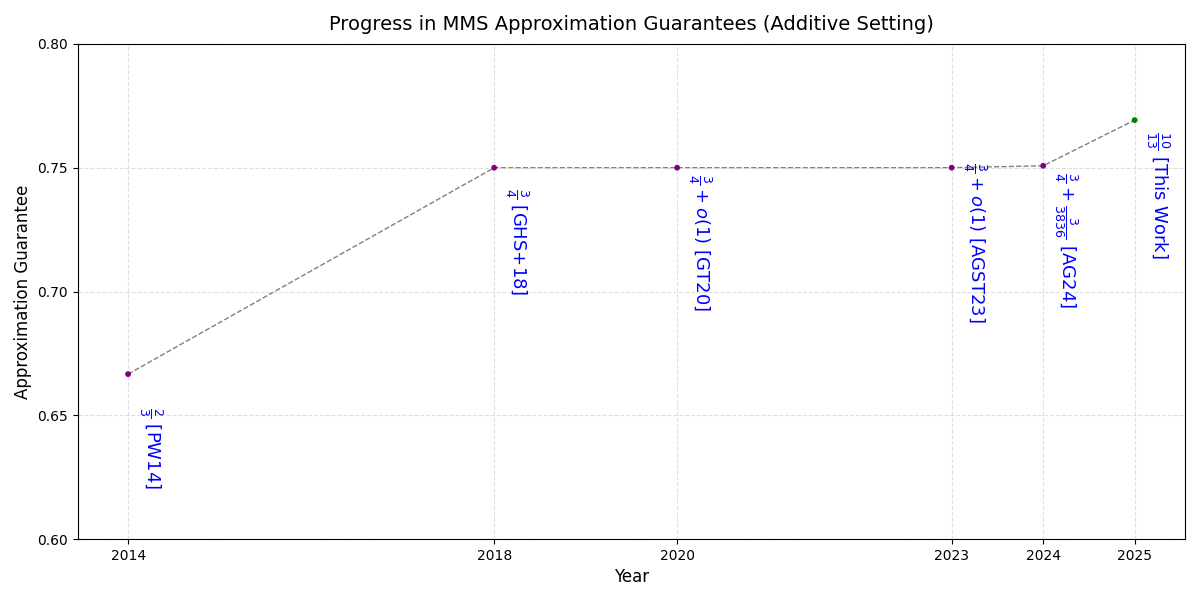}
	\caption{Recent progress on approximation guarantees for $\MMS$ in the additive setting.}
	\label{fig:figure1}
\end{figure}

	\subsection{Further Related Work}  

Much of the work on $\MMS$ approximations for \textit{additive valuations} has been discussed in the introduction. Here, we briefly mention additional results and focus on broader settings. 

On the impossibility side, $\MMS$ allocations do not always exist~\cite{procaccia2014fair}. 
Kurokawa et al.~\cite{kurokawa2016can} showed this even when $m \le 3n + 4$, 
and Feige et al.~\cite{feige2021tight} proved an upper bound of $1 - \mmsfrac{1}{n^4}$ 
and also established a bound of $\mmsfrac{39}{40}$ for three agents.

Beyond additive valuations, $\MMS$ has been studied under \textit{submodular, fractionally subadditive, and subadditive valuations}.  
For submodular valuations, Barman and Krishnamurthy~\cite{barman2020approximation} initiated this direction with a $0.21$-approximation. Ghodsi et al.~\cite{ghodsi2018fair} later improved the bound to $\mmsfrac{1}{3}$, and more recently, Uziahu and Feige~\cite{uziahu2023fair} achieved a $\mmsfrac{10}{27}$-approximation. The best-known upper bound remains $\mmsfrac{3}{4}$~\cite{ghodsi2018fair}.  For fractionally subadditive valuations, Ghodsi et al.~\cite{ghodsi2018fair} provided an initial $\mmsfrac{1}{5}$ approximation with an upper bound of $\mmsfrac{1}{2}$. Seddighin and Seddighin~\cite{seddighin2024improved} improved this to $\mmsfrac{1}{4.6}$, and Akrami et al.~\cite{akrami2023randomized} further improved the bound to $\mmsfrac{3}{13}$.  
For subadditive valuations, Ghodsi et al.~\cite{ghodsi2018fair} proved an $\Omega(\mmsfrac{1}{ \log m })$ approximation. Seddighin and Seddighin~\cite{seddighin2024improved} improved the lower bound to $\Omega(\mmsfrac{1}{\log n \log \log n})$. Subsequently, Feige and Huang~\cite{feige2025concentration} improved the approximation to $\Omega(\mmsfrac{1}{\log n})$, which was further improved  to $\Omega(\mmsfrac{1}{\log \log n^2})$ by Seddighin and Seddighin~\cite{seddighin2025beating}, and to $\Omega(\mmsfrac{1}{\log \log n})$ by Feige~\cite{feige2025multi}.

For a \textit{small number of agents or goods}, $\MMS$ allocations exist for two agents. For three agents, successive improvements have raised the best-known approximation guarantee from $\mmsfrac{7}{8}$~\cite{amanatidis2017approximation} to $\mmsfrac{8}{9}$~\cite{gourves2019maximin} and later to $\mmsfrac{11}{12}$~\cite{feige2022improved}. Ghodsi et al.~\cite{ghodsi2018fair} established a $\mmsfrac{4}{5}$ guarantee for four agents. Further existence results hold when $m \leq n+3$~\cite{amanatidis2017approximation} or $m \leq n+5$~\cite{feige2021tight}. Recently, Garg and Shahkar~\cite{garg2025improved} provided better  guarantees for two and three types of agents.

For \textit{chores} (undesirable goods), Aziz et al.~\cite{aziz2017algorithms} extended the definition of \(\MMS\) and provided a \(2\)-approximation, which was improved to \(\mmsfrac{4}{3}\) by Barman and Krishnamurthy~\cite{barman2020approximation} and further to \(\mmsfrac{11}{9}\) by Huang and Lu~\cite{huang2021algorithmic}, and Huang et al.~\cite{huang2023reduction} subsequently provided a \(\mmsfrac{13}{11}\)-approximation for chores.
\(\MMS\) guarantees have also been studied in \textit{ordinal settings}~\cite{hosseini2022ordinal,akrami2024improving, garg2024improved} and in \textit{weighted} \(\MMS\)~\cite{aziz2019weighted, farhadi2019fair}, for which Wang, Li, and Lu~\cite{ijcai2024p334} achieved an \(O(\log n)\)-approximation.

	\section{Basic Notations} 
\label{sec:prelim}
We denote the input instance of our $\MMS$ allocation Algorithm by $\Pinstance =  (\Pagents, \Pgoods )$, where $\Pagents$ is the set of agents and $\Pgoods$ is the set of goods. Also, we have $|\Pagents| = \Pnumb$ and $|\Pgoods| = \Pit$.  For each agent \( \agent_i \), we denote their valuation function by \( \valu_\agenti: {2^{\Pgoods}} \to \mathbb{R}^{\geq 0} \), which assigns a non-negative value to every subset of goods. We assume valuations are additive, meaning that for any disjoint subsets $S, T \subseteq \Pgoods$,  
$
\valu_\agenti(S \cup T) = \valu_\agenti(S) + \valu_\agenti(T).
$  
Thus, the valuation of a subset $S$ simplifies to  
$
\valu_\agenti(S) = \sum_{\Pgood \in S} \valu_\agenti(\{\Pgood\}).
$

A key assumption we make is that the input instance is \textbf{ordered}, meaning all agents rank the goods in the same order according to their values. Barman and Krishnamurthy \cite{barman2020approximation} showed that any $\MMS$ allocation instance with additive valuations can be reduced to an ordered instance.

\begin{theorem}[Theorem 3.2 of \cite{barman2020approximation}(Restated)]
	\label{order_reduction}
	For every instance $\prinstance$, there exists an ordered instance $\seinstance$ such that any $\alpha$-$\MMS$ allocation for $\seinstance$ can be converted into an $\alpha$-$\MMS$ allocation for $\prinstance$.
\end{theorem}
By \Cref{order_reduction}, we assume all agents rank the goods in a common order. Hence, we denote $\Pgoods = \langle \Pgood_1, \Pgood_2, \ldots, \Pgood_{\Pit} \rangle$, where the goods are sorted in non-increasing order of their values for all agents, i.e.,  for every agent $\agent_i \in \Pagents$ and every $ j < k \leq \Pit$, we have $\valu_\agenti(\{\Pgood_j\}) \geq \valu_\agenti(\{\Pgood_k\}).
$

Given a constant $d$, a valuation function $\valu$, and a set $S$ of goods, the {maximin share value} of $\valu$ with respect to $d$ and $S$ is defined as  
$$
\mms_{\valu}^d(S) = \max_{\langle \pi_1,\pi_2, \ldots, \pi_d \rangle  \in \Pi_d(S)} \min_{j = 1}^d \valu(\pi_j),
$$  
where $\Pi_d(S)$ is the set of all partitionings of $S$ into $d$ bundles. For an agent $\agent_i \in \Pagents$, we refer to $\allmms{\valu_\agenti}$ as her maximin share value. Our goal is to compute an allocation that guarantees each agent a constant-factor approximation of her maximin share value.
 In \Cref{def:appxMMS}, we formally define approximate maximin share allocations.

 \begin{definition}\label{def:appxMMS}
 	For a constant $\alpha$, we say
 	an allocation that allocates a distinct bundle $\allocation_\agenti$ to each agent $\agent_i$ is $\alpha$-$\MMS$, if for every agent $\agent_i$, $\valu_\agenti(\allocation_\agenti)\geq \alpha \allmms{\valu_{\agenti}}$.
 \end{definition}

Our goal in this paper is to prove the existence of a $(\tfrac{10}{13})$-$\MMS$ allocation. For this purpose, we set $\alpha = \tfrac{10}{13}$. A key property of the maximin share is that it is {scale-free}. That is, an agent’s maximin share depends only on her valuations, so multiplying or dividing all values by a constant factor does not affect the approximation guarantee of an allocation. Hence, we suppose without loss of generality that for every agent $\agent_{\agenti}$, their maximin share is scaled such that $\allmms{\valu_{\agenti}} = 1$. The goal is then to allocate each agent a bundle with value at least $\alpha$ according to their valuation.

\subsection{Algorithm Structure and Notation}
In \Cref{fig:flowchart}, we present a flowchart of our algorithm. As mentioned, the input instance is denoted by $\Pinstance = (\Pagents, \Pgoods)$. Our algorithm proceeds as follows:

\begin{itemize}
 \item \textbf{Primary Reductions:} We apply a set of primary reductions to the input instance $\Pinstance$. The output of this step is denoted by $\instance = (\agents, \goods)$, where $\agents \subseteq \Pagents$ is the remaining set of agents,  and $\goods \subseteq \Pgoods$ is the remaining set of goods.  

\item \textbf{Secondary Reductions for Case 1:}  After primary reductions, the algorithm branches into two cases. For the first case, we apply a set of secondary reductions. The output of these reductions is denoted by $\Sinstance = (\Sagents, \Sgoods)$, where
$\Sagents \subseteq \agents$ is the set of agents after the reductions, and  
$\Sgoods \subseteq \goods $ is the set of goods after the reductions.  
\end{itemize}

Afterwards, the algorithm executes a Bag-filling process tailored to each case. Further details on these reductions and Bag-filling procedures are provided in \Cref{sec:resultstech}.
For convenience, we assume the following notations for the instances:  
$|\Pagents| = \Pnumb$, $|\agents| = \numb$, and $|\Sagents| = \Snumb$ for the agents,  
$|\Pgoods| = \Pit$, $|\goods| = \It$, and $|\Sgoods| = \Sit$ for the goods.  

All three instances ($\Pinstance$, $\instance$, and $\Sinstance$) are assumed to be ordered. Specifically $\Sgoods = \langle \Sgood_1, \Sgood_2, \ldots, \Sgood_{\Sit} \rangle$,   $\goods = \langle \good_1, \good_2, \ldots, \good_{\It} \rangle$,  and $\Pgoods = \langle \Pgood_1, \Pgood_2, \ldots, \Pgood_{\Pit} \rangle$,  where the goods are sorted in non-increasing order of their values for all agents. Recall that for every agent $\agent_i$, we have $\mms_{\valu_i}^{\Pnumb}(\Pgoods)=1$. 
For convenience, for a function $\valu$, we denote  $\mms_{\valu}^{\numb}(\goods)$ by $\dotmms{\valu}$, and $\mms_{\valu}^{\Snumb}(\Sgoods)$ by $\ddotmms{\valu}$. 

To state and prove some of our lemmas and theorems in a general setting, we occasionally consider arbitrary instances, which we denote by \(\prinstance = (\pragents, \prgoods)\) and \(\seinstance = (\seagents, \segoods)\). Following our notational convention, we let \(|\pragents| = \prnumb\), \(|\prgoods| = \prit\), \(|\seagents| = \senumb\), and define \(\hatmms{\prvalu} = \mms_{\prvalu}^{\prnumb}(\prgoods)\) and \(\checkmms{\valu} = \mms_{\valu}^{\senumb}(\segoods)\).

\begin{figure}[t]
	\centering
	\tikzset{every picture/.style={line width=0.75pt}} 
	\scalebox{0.7}{
		\begin{tikzpicture}[x=0.75pt,y=0.75pt,yscale=-0.8,xscale=0.8]
			\draw  [color={rgb, 255:red, 128; green, 128; blue, 128 } ][fill={green!10}] (330,38) .. controls (330,33.58) and (333.58,30) .. (338,30) -- (486,30) .. controls (493.42,30) and (494,33.58) .. (494,38) -- (494,62) .. controls (494,66.42) and (490.42,70) .. (486,70) -- (338,70) .. controls (333.58,70) and (330,66.42) .. (330,62) -- cycle ;
			\draw    (420,67) -- (420,105) ;
			\draw [shift={(420,107)}, rotate = 270] [color={rgb, 255:red, 0; green, 0; blue, 0 }  ][line width=0.75]    (10.93,-3.29) .. controls (6.95,-1.4) and (3.31,-0.0) .. (0,0) .. controls (3.31,0.3) and (6.95,1.4) .. (10.93,3.29)   ;
			\draw  [color={rgb, 255:red, 128; green, 128; blue, 128 }  ][fill={blue!5} ] (330,117) .. controls (330,112.58) and (333.58,109) .. (338,109) -- (490,109) .. controls (494.42,109) and (498,112.58) .. (498,117) -- (498,187) .. controls (498,191.42) and (494.42,195) .. (490,195) -- (338,195) .. controls (333.58,195) and (330,191.42) .. (330,187) -- cycle ;
			\draw    (420,195) -- (420,235) ;
			\draw [shift={(420,237)}, rotate = 270] [color={rgb, 255:red, 0; green, 0; blue, 0 }  ][line width=0.75]    (10.93,-3.29) .. controls (6.95,-1.4) and (3.31,-0.3) .. (0,0) .. controls (3.31,0.3) and (6.95,1.4) .. (10.93,3.29)   ;
			\draw  [color={rgb, 255:red, 128; green, 128; blue, 128 } ][fill={green!10} ] (330,247) .. controls (330,242.58) and (333.58,239) .. (338,239) -- (490,239) .. controls (494.42,239) and (498,242.58) .. (498,247) -- (498,287) .. controls (498,291.42) and (494.42,295) .. (490,295) -- (338,295) .. controls (333.58,295) and (330,291.42) .. (330,287) -- cycle ;
			\draw    (420,295) -- (420,340) ;
			\draw [shift={(420,340)}, rotate = 270] [color={rgb, 255:red, 0; green, 0; blue, 0 }  ][line width=0.75]    (10.93,-3.29) .. controls (6.95,-1.4) and (3.31,-0.3) .. (0,0) .. controls (3.31,0.3) and (6.95,1.4) .. (10.93,3.29)   ;
			\draw  [fill={red!10}  ] (420,340) -- (510,399.5) -- (420,459) -- (330,399.5) -- cycle ;
			\draw    (510,399.5) -- (578,399.99) ;
			\draw [shift={(580,400)}, rotate = 180.41] [color={rgb, 255:red, 0; green, 0; blue, 0 }  ][line width=0.75]    (10.93,-3.29) .. controls (6.95,-1.4) and (3.31,-0.3) .. (0,0) .. controls (3.31,0.3) and (6.95,1.4) .. (10.93,3.29)   ;
			\draw    (670,458) -- (670,519) ;
			\draw [shift={(670,521)}, rotate = 270] [color={rgb, 255:red, 0; green, 0; blue, 0 }  ][line width=0.75]    (10.93,-3.29) .. controls (6.95,-1.4) and (3.31,-0.3) .. (0,0) .. controls (3.31,0.3) and (6.95,1.4) .. (10.93,3.29)   ;
			\draw    (670,450) -- (670,505) ;
			\draw [shift={(670,505)}, rotate = 270] [color={rgb, 255:red, 0; green, 0; blue, 0 }  ][line width=0.75]    (10.93,-3.29) .. controls (6.95,-1.4) and (3.31,-0.0) .. (0,0) .. controls (3.31,0.3) and (6.95,1.4) .. (10.93,3.29)   ;
			\draw    (580,615) -- (520,615) -- (450,615) -- (450,675) ;
			\draw [shift={(450,675)}, rotate = 270] [color={rgb, 255:red, 0; green, 0; blue, 0 }  ][line width=0.75]    (10.93,-3.29) .. controls (6.95,-1.4) and (3.31,-0.3) .. (0,0) .. controls (3.31,0.3) and (6.95,1.4) .. (10.93,3.29)   ;
			\draw    (280,615) -- (380,615) -- (380,650) -- (380,675) ;
			\draw [shift={(380,675)}, rotate = 270] [color={rgb, 255:red, 0; green, 0; blue, 0 }  ][line width=0.75]    (10.93,-3.29) .. controls (6.95,-1.4) and (3.31,-0.3) .. (0,0) .. controls (3.31,0.3) and (6.95,1.4) .. (10.93,3.29)   ;
			\draw  [color={rgb, 255:red, 128; green, 128; blue, 128 }  ][fill={green!10}  ] (331,685) .. controls (331,680.58) and (334.58,677) .. (339,677) -- (497,677) .. controls (501.42,677) and (505,680.58) .. (505,685) -- (505,739) .. controls (505,743.42) and (501.42,747) .. (497,747) -- (339,747) .. controls (334.58,747) and (331,743.42) .. (331,739) -- cycle ;
			\draw  [color={rgb, 255:red, 128; green, 128; blue, 128 }  ][fill={blue!5} ] (580,372) .. controls (580,360) and (588.95,352) .. (600,352) -- (740,352) .. controls (751.05,352) and (760,360) .. (760,372) -- (760,422) .. controls (760,443) and (751.05,452) .. (740,452) -- (600,452) .. controls (588.95,452) and (580,443) .. (580,432) -- cycle ;
			\draw    (190,580) -- (190,400) -- (330,399.5) ;
			\draw [shift={(190,580)}, rotate = 270] [color={rgb, 255:red, 0; green, 0; blue, 0 }  ][line width=0.75]    (10.93,-3.29) .. controls (6.95,-1.4) and (3.31,-0.3) .. (0,0) .. controls (3.31,0.3) and (6.95,1.4) .. (10.93,3.29)   ;
			\draw  [color={rgb, 255:red, 128; green, 128; blue, 128 } ][fill={green!10}] (585,515) .. controls (585,510) and (588.58,507) .. (593,507) -- (741,507) .. controls (745.42,507) and (749,511) .. (749,515) -- (749,529) .. controls (749,533) and (745.42,537) .. (741,537) -- (593,537) .. controls (588.58,537) and (585,533) .. (585,529) -- cycle ;
			\draw    (670,537) -- (670,580) ;
			\draw [shift={(670,580)}, rotate = 270] [color={rgb, 255:red, 0; green, 0; blue, 0 }  ][line width=0.75]    (10.93,-3.29) .. controls (6.95,-1.4) and (3.31,-0.0) .. (0,0) .. controls (3.31,0.3) and (6.95,1.4) .. (10.93,3.29)   ;
			\draw  [color={rgb, 255:red, 128; green, 128; blue, 128 }  ][fill={blue!5} ] (580,600) .. controls (580,588) and (588.95,580) .. (600,580) -- (740,580) .. controls (751.05,580) and (760,589) .. (760,600) -- (760,660) .. controls (760,671) and (751.05,680) .. (740,680) -- (600,680) .. controls (588.95,680) and (580,681) .. (580,660) -- cycle ;
			\draw  [color={rgb, 255:red, 128; green, 128; blue, 128 }  ][fill={blue!5}  ] (100,600) .. controls (100,588) and (108.95,580) .. (120,580) -- (260,580) .. controls (271.05,580) and (280,588) .. (280,600) -- (280,660) .. controls (280,671) and (271.05,680) .. (260,680) -- (120,680) .. controls (108.95,680) and (100,671) .. (100,660) -- cycle ;
			\draw (360,40) node [anchor=north west][inner sep=0.75pt]    {$\Pinstance \ =\ (\Pagents,\ \Pgoods)$};
			\draw (365,140) node [anchor=north west][inner sep=0.75pt]  [font=\small]  {$\reductiontype^0, \reductiontype^1, \reductiontype^2, \rstar^1{\displaystyle \textcolor[rgb]{0.82,0.01,0.11}{}}$};
			\draw (350,120) node [anchor=north west][inner sep=0.75pt]  [font=\small]  {${\textstyle \mathsf{Primary\ Reductions}}$};
			\draw (365,245) node [anchor=north west][inner sep=0.75pt]  [font=\small]  {$\instance \  = \  (\agents,\  \goods)$};
			\draw (350,270) node [anchor=north west][inner sep=0.75pt]  [font=\small]  {Define $\Nyek$ and $\Ndo$};
			\draw (375,390) node [anchor=north west][inner sep=0.75pt]    {$|\Nyek|\ \geq \frac{\Pnumb}{\sqrt{2}} \ $};
			\draw (600,397) node [anchor=north west][inner sep=0.75pt]  [font=\small]  {$\reductiontype^1, \reductiontype^2,\reductiontype^3, \reductiontype^4, \rstar^2{\displaystyle \textcolor[rgb]{0.82,0.01,0.11}{}}$};
			\draw (596,370) node [anchor=north west][inner sep=0.75pt]  [font=\small]  {${\textstyle \mathsf{Secondary\ Reductions}}$};
			\draw (147,595) node [anchor=north west][inner sep=0.75pt]  [font=\small]  {Bag-filling};
			\draw (632,595) node [anchor=north west][inner sep=0.75pt]  [font=\small]  {Bag-filling};
			\draw (261,380) node [anchor=north west][inner sep=0.75pt]   [align=left] {No};
			\draw (523,380) node [anchor=north west][inner sep=0.75pt]   [align=left] {Yes};
			\draw (580,620) node [anchor=north west][inner sep=0.75pt]  [font=\small]  {$B_{k}\!=\!\!\{\Sgood_{k}, \Sgood_{\Snumb+k},\Sgood_{3\Snumb-k+1}\}$};
			\draw (335,687) node [anchor=north west][inner sep=0.75pt][font=\small]    {$(\nicefrac{10}{13})$-$\mathsf{MMS}$ $ \mathsf{ Allocation}$};
			\draw (370,717) node [anchor=north west][inner sep=0.75pt][font=\small]    {\cref{thm:mms}};
			\draw (631,430) node [anchor=north west][inner sep=0.75pt]  [font=\small]  {\cref{algo:N11}};
			\draw (631,650.4) node [anchor=north west][inner sep=0.75pt]  [font=\small]  {\cref{algo:1}};
			\draw (143,650) node [anchor=north west][inner sep=0.75pt]  [font=\small]  {\cref{algo:N2}};
			\draw (121,620) node [anchor=north west][inner sep=0.75pt]  [font=\small]  {$ \ B_{k} \ =\ \{\good_{k} ,\ \good_{\numb+k}\}$};
			\draw (375,170) node [anchor=north west][inner sep=0.75pt]  [font=\small]  {\cref{thm:main}};
			\draw (618,512) node [anchor=north west][inner sep=0.75pt]  [font=\small]  {$\Sinstance \  = \  (\Sagents, \  \Sgoods)$};
		\end{tikzpicture}
	}  
	\caption{A flowchart of our algorithm.}
	\label{fig:flowchart}
\end{figure}

Finally, to simplify the analysis and presentation, we make a few standard assumptions about the input. First, to ensure that the indices of goods used in reductions or during the Bag-filling process do not exceed the total number of goods, we assume that the number of goods is at least \( 5n \).
 This can always be ensured by adding dummy goods that have value $0$ for all agents. Second, we assume that no good has a value greater than $1$ to any agent. This assumption is common in prior work and has been shown to be without loss of generality~\cite{ghodsi2018fair}. In fact, it is immediately justified by the first reduction we introduce.
  
\begin{remark}For the reader's convenience, we provide \cref{tab:mms-blank}, which summarizes notations that are frequently used throughout the paper.\end{remark}

\section{Highlights of Techniques} \label{sec:resultstech}
\subsection{Algorithmic Overview}

 Our approach follows the classical framework adopted in much of the prior work. This framework consists of two main components: \textbf{a set of reductions} and \textbf{a Bag-filling process}.  
The reduction phase focuses on allocating \emph{large} goods. Roughly speaking, a reduction identifies a small subset of goods (typically of size fewer than four) each valued at least $\alpha$ to an agent\footnote{We present this section under the assumption that the goal is to find an \(\alpha\)-\(\MMS\) allocation. While this is not the objective in prior studies, the assumption does not affect the description of the underlying ideas.} and (almost) preserving the maximin share value of the remaining agents for the remaining goods. If such a subset exists, the algorithm allocates it to that agent and reduces the problem to a smaller instance with fewer goods and agents. 

As a simple example, the most basic form of a reduction checks whether there exists a good valued at least \(\alpha\) by some agent. If so, we allocate it to that agent, remove both the good and the agent, and recurse on the rest. This works because removing one good and one agent does not decrease the maximin share values of the remaining agents for the remaining goods.
 Moreover, this reduction ensures that in the residual instance—where this reduction is no longer applicable—each good has a value of at most \( \alpha \) to every agent.  A similar---though more complex---principle applies to other reductions. These reductions impose a rich collection of upper bounds and structural constraints on the values that agents assign to various goods.

When no further reductions apply, the instance consists of goods that hold relatively low value for the agents. At this stage, the algorithm invokes a Bag-filling process to allocate the remaining goods. While the Bag-filling procedure can be intricate, its core idea draws inspiration from the classic moving-knife method in cake-cutting: starting with an empty bag, goods are added one by one until an agent among those still participating calls “STOP!”—indicating that the current bundle meets the targeted approximation guarantee of the agent’s maximin share.
The bundle is then given to that agent, and they are removed from the process. 

The core idea of the Bag-filling process is as follows: when a bundle is allocated to an agent who shouts "STOP!", this bundle must have value less than $\alpha$ for any agent who has not yet shouted. Our goal is to upper bound the value that each allocated bundle holds for the remaining agents. 

However, complications arise when multiple agents shout ``STOP!'' at the same time. If we give the bundle to one of them, the value of that bundle to the others may exceed $\alpha$ (or sometimes even $1$). Note that since no one shouted before the last good was added, this excess value is bounded by the value of a single remaining good.

To handle this, several key strategies can be employed. First, bags may be initialized with higher-value remaining goods to ensure these goods are evenly distributed and do not end up among the last goods added. Second, a priority order among agents can be used to break ties when multiple agents shout simultaneously. This priority typically favors agents who are more likely to face difficulties later in the Bag-filling process.

Together, these strategies—along with the bounds established through the reductions and a careful analysis of bundle values—yield \cref{thm:mms}, which is the main result of the paper. We note that deriving some of these bounds is highly nontrivial and relies heavily on the structure of the agents' maximin share partitions. 

\begin{restatable}{theorem}{maintheorem}\label{thm:mms}
	The allocation returned by \cref{alg:main} is $(\mmsfrac{10}{13})$-$\MMS$. 
\end{restatable}
\subsection{Techniques}

 While the overall approach shares similarities with that of Akrami and Garg~\cite{akrami2024breaking}, our method introduces several key improvements and novel insights, which can be summarized as follows. We note that the analysis of Akrami and Garg \cite{akrami2024breaking} is tight for their algorithm.  

\begin{itemize}  
	\item \textbf{Dynamic reductions:} In previous approaches, a reduction is typically defined by fixing the indices of the goods considered in the ordered list of goods. These indices remain static for each reduction. In contrast, we introduce more flexible reductions: we allow the index of the smallest allocated good (i.e., the good with the largest index) to be determined dynamically based on the input—specifically, we let it be as large as possible. Although this flexibility might seem minor, it plays a crucial role in uncovering useful patterns in the valuations of agents whose $\MMS$ values decrease during the reduction process, which in turn allow us to partially offset this decrease in later steps. 
	 Further details are provided in \Cref{sec:ideas-reduction}.
	
	\item \textbf{Deferred matching:} We introduce a more flexible bundle allocation strategy in the reduction phase. When multiple agents are eligible for a bundle, we initially assign it to one agent temporarily, but keep the option to reassign it later, after the reduction phase is complete. This deferred matching approach allows us to make more informed allocation decisions, based on agents’ valuations over the remaining goods. We elaborate on this in \Cref{sec:ideas-reduction}.
	
	\item \textbf{Calibration functions:} A key conceptual innovation in our approach is the introduction of \emph{calibration functions}. These functions streamline analysis and enable a more precise approach that improves the approximation guarantee. Unlike prior work, which often treated reductions as a black box—focusing only on the maximin share values after reduction—our method preserves additional structural details through these calibration functions. They capture how good values change during a reduction, allowing us to analyze complex cases with greater accuracy that would otherwise be challenging to handle. More detail is given in Section \ref{sec:contribute-adjfunc}.

	\item \textbf{Bundle initialization:} Another key difference from the approach of Akrami and Garg~\cite{akrami2024breaking} lies in how we initialize the bundles in the Bag-filling process. This modification leads to a stronger approximation guarantee in both cases we consider, especially in the second case. More detail is given in \Cref{sec:ideas-bag}.
\end{itemize}  

 Below, we discuss our techniques in more detail and highlight how they compare to previous approaches. We emphasize that, for clarity and ease of presentation, the notation and arguments in this section have been simplified and are not fully rigorous.

\subsection{Reductions}\label{sec:ideas-reduction}
As discussed earlier, a reduction simplifies the problem by allocating large goods. Previous studies have introduced several useful types of reductions, which here we denote by $\textsf{R}^0$ to $\textsf{R}^3$. 

Let us first review these reductions. Consider an ordered instance $\prinstance =  (\pragents, \prgoods )$ where $\prgoods = \langle\prgood_1,\prgood_2,\ldots, \prgood_{\prit} \rangle$. Reduction $\textsf{R}^0$ checks whether $\prgoodrat{1}$ is valued at least $\alpha$ by some agent. If so, we allocate $\prgoodrat{1}$ to that agent and recursively solve the problem for the remaining agents and goods. As mentioned earlier, it has been shown that allocating in this manner does not reduce the maximin share values of the remaining agents for the remaining goods. Consequently, any approximation guarantee achieved for the new instance also applies to the original instance. 

Reductions $\textsf{R}^1$, $\textsf{R}^2$, and $\textsf{R}^3$ follow a similar approach for specific subsets of goods. $\textsf{R}^1$ considers subset $\{\prgoodrat{\prnumb}, \prgoodrat{\prnumb+1}\}$, $\textsf{R}^2$ considers $\{\prgoodrat{2\prnumb-1}, \prgoodrat{2\prnumb}, \prgoodrat{2\prnumb+1}\}$, and $\textsf{R}^3$ considers
  $\{\prgoodrat{3\prnumb-2}, \prgoodrat{3\prnumb-1}, \prgoodrat{3\prnumb}, \prgoodrat{3\prnumb+1}\}$. Each rule checks whether these goods together have value at least $\alpha$ for some agent and, if so, allocates them accordingly and solves the problem recursively for the remaining goods and agents.

There are also two special reductions $\widetilde{\mathsf{R}}^1$ and $\widetilde{\mathsf{R}}^2$ introduced respectively by Garg and Taki \cite{garg2020improved}, and  Akrami and Garg \cite{akrami2024breaking}. These reductions check whether bundles \( \{\prgoodrat{1}, \prgoodrat{2\prnumb+1}\} \) and \( \{\prgoodrat{1}, \prgoodrat{2}\} \), respectively, are worth more than \( \alpha \) to some agent. If so, the bundle is allocated to that agent, and the agent is removed from the instance.  
What sets these reductions apart from the previous ones (which is why they are denoted by tilde) is that they may slightly decrease the maximin share value for some agents. However, they show that this decrease is limited and remains within a tolerable range.

\paragraph{Our reductions.} 
In this paper, we introduce a new reduction and refine the existing ones---\( \textsf{R}^0, \textsf{R}^1, \textsf{R}^2, \textsf{R}^3, \widetilde{\mathsf{R}}^1 \), and \( \widetilde{\mathsf{R}}^2 \). We refer to our set of reductions as \( \reductiontype^0, \reductiontype^1, \reductiontype^2, \reductiontype^3, \reductiontype^4, \rstar^1 \), and \( \rstar^2 \).
The modification we make to obtain \( \reductiontype^0, \reductiontype^1, \reductiontype^2, \reductiontype^3, \rstar^1 \), and \( \rstar^2 \) is simple yet effective: rather than fixing in advance which goods are used in each reduction, we determine the rightmost index dynamically, based on the instance: we allow its index to be shifted as far to the right as possible. For example, rule \( \reductiontype^1 \) identifies the largest index \( x > \prnumb \) (if it exists) such that set \( \{\prgoodrat{\prnumb}, \prgoodrat{x}\} \) has value \( \alpha \) for some agent.   
Or reduction \( \rstar^2 \) identifies the largest index \( x \) for which the bundle \( \{\prgoodrat{1}, \prgoodrat{x}\} \) is worth at least \( \alpha \) to some agent. If such an index exists, the instance is updated accordingly. Note that $\rstar^1,\rstar^2$ may slightly decrease the maximin share value for the remaining agents.  

Also, reduction \( \reductiontype^4 \) follows the same pattern as the previous rules. It identifies the largest index \( x > 4\prnumb \) (if it exists) such that set \( \{\prgoodrat{4\prnumb-3}, \prgoodrat{4\prnumb-2}, \prgoodrat{4\prnumb-1}, \prgoodrat{4\prnumb}, \prgoodrat{x}\} \) values at least $\alpha$ to some agent.  

Shifting the last index to the right as far as possible serves two main purposes.  For reductions \( \reductiontype^0 \) to \( \reductiontype^4 \), this modification helps establish a tighter bound on the value of the allocated bundle during the reduction process.   For $\rstar^1$ and $\rstar^2$, this adjustment reveals a key structural property: If applying one of these reductions causes an agent's maximin share value to drop below $1$, it must be because the agent values the selected bundle significantly—that is, $\prvalu_i(\{\prgoodrat{1}\}) + \prvalu_i(\{\prgoodrat{x}\}) > 1$. At the same time, since $\prgoodrat{x}$ is the rightmost good that could be included in the reduction, replacing it with $\prgoodrat{x+1}$ would not satisfy the reduction condition; thus, we also have $\prvalu_i(\{\prgoodrat{1}\}) + \prvalu_i(\{\prgoodrat{x+1}\}) < \alpha$. This implies a value gap of at least $1 - \alpha$ between $\prgoodrat{x}$ and $\prgoodrat{x+1}$ from the agent's perspective.

This value gap leads to an important consequence: If applying $\rstar^1$ causes an agent's maximin share value to drop below $1$, and considering that $\prvalu_i(\{\prgoodrat{1}\}) < \alpha$ (since $\reductiontype^0$ is not applicable) and $\prvalu_i(\{\prgoodrat{x}\}) < \alpha/3$ (since $\reductiontype^2$ is not applicable), we can deduce:
$
\prvalu(\{\prgoodrat{x}\}) > 1 - \alpha$ and $\prvalu(\{\prgoodrat{x+1}\}) < \tfrac{4\alpha}{3} - 1.
$
Therefore, the agent has no good valued in the interval $[\tfrac{4\alpha}{3} - 1, 1 - \alpha].$ A similar argument holds for reduction $\rstar^2$. These insights plays critical role in our analysis.

\paragraph{Deferred matching.} In the primary reductions, we apply the reductions with priority $\reductiontype^0 \succ \reductiontype^1 \succ \reductiontype^2 \succ \rstar^1$ until none is applicable. The obtained irreducible instance has $\numb$ agents.  We color the agents in $\Pagents$ (the original set of agents before reductions) into two categories: green or red. An agent is green if their value for $\goodrat{2\numb+1}$ (the good ranked $2\numb+1$ in the reduced instance) is at least $1 - \alpha$; otherwise, they are red. We then prioritize either green  or red agents based on their sizes:  if the number of green  agents in $\Pagents$ is at least $\tfrac{\Pnumb}{\sqrt{2}}$, we proritize red agents; otherwise, we prioritize green agents. This prioritization guides the Bag-filling process.

The intuition behind deferred matching is as follows: To achieve a better approximation guarantee, the prioritization must also be considered during the reduction process itself. At first glance, this may seem paradoxical since prioritization depends on the outcome of reduction sequence. However, we show that it is indeed achievable.
To address this, the allocations in the reduction sequence remain temporary—we may reallocate the bundle to another agent once the sequence of reductions is complete. Suppose we reach an irreducible instance  after \( \Pnumb -\numb \) reductions. At this point, we can determine whether each agent in $\Pagents$ is green or red based on their value for $\goodrat{2\numb+1}$. Now, we finalize the reduction by constructing a bipartite graph \( G \) as follows:  
\begin{itemize}
	\item For each bundle allocated during the reduction phase, we add a vertex to the first part of \( G \).  
	\item For each agent in $\Pagents$, we add a vertex to the second part of \( G \).  
	\item We draw an edge between the corresponding vertices of an agent and a bundle if the agent values that bundle at least \( \alpha \).  
\end{itemize}
Note that, by construction, the reduction sequence corresponds to a matching that saturates all bundle vertices. On the other hand we show that, by selecting a specific sequence of reductions, called perfect sequence of reductions, each maximum matching in \( G \) corresponds to a valid reduction sequence. Among all possible maximum matchings, we select one that maximizes the number of agents from the prioritized color (either red or green) . We prove that such a matching incorporates the prioritization into the reduction process. See \Cref{fig:mathc}.
\begin{figure}
	\begin{center}
\usetikzlibrary{positioning, shapes.geometric, arrows.meta}

	\begin{tikzpicture}[
		x=1.2cm,  
		y=1cm,
		bundle/.style={
			rectangle, 
			draw=black, 
			fill=blue!5, 
			minimum width=1.2cm,
			minimum height=1cm,
			font=\small
		},
		agent/.style={
			circle, 
			draw=black,
			minimum size=0.9cm,
			font=\small
		},
		agent1/.style={
			agent,
			fill=green!10
		},
		agent2/.style={
			agent,
			fill=red!10
		},
		edge/.style={
			->,
			>=Stealth,
			semithick,
			gray!60,
			shorten >=1pt,
			shorten <=1pt
		},
		reduction-edge/.style={
			->,
			>=Stealth,
			very thick,
			blue!50,
			shorten >=1pt,
			shorten <=1pt
		},
		optimal-edge/.style={
			->,
			>=Stealth,
			very thick,
			magenta!70,
			densely dashed,
			shorten >=1pt,
			shorten <=1pt
		}
		]
		
		\node[bundle] (B1) at (0,4) {$B_1$};
		\node[bundle] (B2) at (3,4) {$B_2$};
		\node[bundle] (B3) at (6,4) {$B_3$};
		\node[bundle] (B4) at (9,4) {$B_4$};
		\node[bundle] (B5) at (12,4) {$B_5$};
		
		\node[agent1] (A1) at (0,0) {$\agent_1$};
		\node[agent1] (A2) at (1.5,0) {$\agent_2$};
		\node[agent2] (A3) at (3,0) {$\agent_3$};
		\node[agent1] (A4) at (4.5,0) {$\agent_4$};
		\node[agent1] (A5) at (6,0) {$\agent_5$};
		\node[agent1] (A6) at (7.5,0) {$\agent_6$};
		\node[agent2] (A7) at (9,0) {$\agent_7$};
		\node[agent1] (A8) at (10.5,0) {$\agent_8$};
		\node[agent2] (A9) at (12,0) {$\agent_9$};
		
		\foreach \b/\a in {B1/A1, B1/A2, B1/A3,
			B2/A2, B2/A3, B2/A4,
			B3/A3, B3/A5, B3/A6,
			B4/A4, B4/A7, B4/A8,
			B5/A5, B5/A7, B5/A9}
		\draw[edge] (\b) -- (\a);
		
		\draw[reduction-edge] (B1) -- (A1);
		\draw[reduction-edge] (B2) -- (A4);
		\draw[reduction-edge] (B3) -- (A6);
		\draw[reduction-edge] (B4) -- (A8);
		\draw[reduction-edge] (B5) -- (A9);
		
		\draw[optimal-edge] (B1) -- (A1);  
		\draw[optimal-edge] (B2) -- (A3);  
		\draw[optimal-edge] (B3) -- (A5);  
		\draw[optimal-edge] (B4) -- (A7);  
		\draw[optimal-edge] (B5) -- (A9);  
		
	\end{tikzpicture}
\end{center}
\caption{Bipartite graph demonstrating: (1) Reduction matching (solid blue) covering all bundles $\{B_i\}_{i=1}^5$, and (2) A matching (dashed magenta) prioritizing red over green agents. In both matchings, $B_1$ is paired with agent $\agent_1$ and $B_5$ is paired with agent $\agent_9$.}
\label{fig:mathc}
\end{figure}

\subsection{Bag-filling}\label{sec:ideas-bag}
As mentioned earlier, there are two key components we can adjust in a Bag-filling process: how we initialize the bundles and how we prioritize agents when multiple agents simultaneously shout ``STOP." Based on the number of green and red agents, we run two different versions of Bag-filling:

\begin{itemize}
	\item \textbf{Case (i): The number of green agents is at least $\tfrac{\Pnumb}{\sqrt{2}} $.}  
After applying the primary reductions, we obtain the instance $\instance = (\agents, \goods)$. We then apply the secondary reductions with priority $\reductiontype^1 \succ \reductiontype^2 \succ \reductiontype^3 \succ \reductiontype^4 \succ \rstar^2 $ until we reach an irreducible instance $\Sinstance = (\Sagents, \Sgoods)$. Next, we create $\Snumb$ bags, where the $k$-th bag initially contains the goods $\{\Sgoodrat{k}, \Sgoodrat{\Snumb+k}, \Sgoodrat{3\Snumb-k+1}\}$.  
We then perform a Bag-filling process, giving priority to red agents.

	\item \textbf{Case (ii): Otherwise.}  
		We create \( \numb \) bags, where the \( k \)-th bag initially contains goods \( \{\good_k, \good_{\numb+k}\} \).  
		We then perform a Bag-filling process, giving priority to green agents.  
	\end{itemize}
	Note that some bags may already exceed the \( \alpha \) threshold for some agents before any additional goods are added.		

A key innovation of our algorithm  lies in how we initialize the bags in both cases. This distinction is especially crucial in the second case: while previous studies including ~\cite{akrami2024breaking} pair goods as \( \{g_k, g_{2\numb+1-k}\} \), we pair them as \( \{g_k, g_{\numb+k}\} \). Our pairing ensures a consistent ordinal ranking of bundles for all agents—for instance, \( \{g_1, g_{\numb+1}\} \) is more valuable than \( \{g_2, g_{\numb+2}\} \) for every agent. Though subtle, this modification plays an important role in the guarantee of our algorithm.

\subsection{Calibration Functions} \label{sec:contribute-adjfunc}
To facilitate and advance our analysis, we introduce a set of functions, which we call calibration functions, designed to systematically modify agents valuations. These functions do not affect the actual allocation process but serve as theoretical tools to simplify our arguments and also handle more complex cases.

One key challenge in our framework is that reductions $\rstar^1$ and $\rstar^2$ may reduce an agent’s maximin share, making it harder to guarantee that they receive a sufficiently valuable bundle during the Bag-filling steps. Calibration functions help address this issue by carefully modifying valuations so that the agent’s maximin share, when computed under the calibrated valuation, does not decrease. Moreover, these functions are designed to ensure that the maximin share value under the calibrated valuation stays sufficiently close to its original counterpart. This allows us to establish a meaningful lower bound on the value each agent receives in the final allocation.

These functions play an important role in providing more in-depth analysis of what happens in  reduction steps. Since calibration functions comprise a family of functions parameterized differently, we can define a suitable calibrated value for each agent. This allows us to perform more accurate analysis suitable for an agents valuations.
 Moreover, because the definition of calibration functions is deterministic, we always have access to both the original and the calibrated values, allowing us to employ both simultaneously for a more refined analysis.

For a formal definition and detailed properties of these functions, see \cref{sec:virtual-value}. Also, see \Cref{fig:plot-fh} and \Cref{fig:plot-wz} for a visual representation of these functions. 

\subsection{Organization of the Paper}
The remainder of the paper is organized as follows. In \Cref{sec:reduc}, we describe our reduction steps. \Cref{sec:virtual-value} introduces the calibration functions, along with key properties that are used throughout the paper. The formal proofs of the $\MMS$ bounds for calibrated valuations are deferred to the appendix.
\Cref{sec:mainalg} outlines the overall structure of our algorithm, including the initial reductions and introduces the two main cases. The first case is presented in \Cref{sec:alg-N1}, while \Cref{sec:alg-N2} addresses the second case.

To keep the presentation focused and accessible, we defer full proofs of Lemmas~\ref{plus-one}, \ref{minus-one}, \ref{plus-two}, \ref{minus-two}, \ref{lem:N21J}, and \ref{lem:N22J} to the appendix, and include only brief proof sketches in the main text. Furthermore, to aid intuition and readability, we include several supporting tables, figures, and examples. \Cref{table:main} summarizes which lemmas establish $\MMS$ guarantees for different subsets of agents. Detailed examples illustrating key steps of the algorithm and its analysis can be found in \Cref{sec:examples}. We also provide a table of notations in \Cref{sec:notationstable}, listing essential definitions used throughout the algorithm.

\section{Reductions} \label{sec:reduc}
Similar to most previous studies on maximin share, especially in the additive setting, our algorithm begins with a series of reductions. However, in the previous studies on maximin share in the additive setting, the reduction process was often treated as an unstructured procedure, where a set of reduction rules were applied sequentially with some priority over the reductions to the input until no further reduction is possible. In contrast, our approach introduces a more clever way to apply these reductions. Hence, here we need a more rigorous definition of a reduction. 

In this section, we consider a fixed ordered instance \( \prinstance = (\pragents, \prgoods) \) where the set of agents is \( \pragents = \{\pragent_1, \pragent_2, \ldots, \pragent_{\prnumb}\} \) and the set of goods is \( \prgoods = \{\prgood_1, \prgood_2, \ldots, \prgood_{\prit}\} \). For each agent \( \pragent_i \), let \( \prvalu_i \) denote the valuation function of $\pragent_i$. In what follows, we will provide our definitions based on this setup.

\textbf{Reduction patterns} are typically defined as subsets of indices in the sorted sequence of goods. Previous studies have introduced several fixed reduction patterns, such as:
\begin{align*}
	& \{1\},\\
	& \{\prnumb,\, \prnumb+1\},\\
	& \{2\prnumb-1,\, 2\prnumb,\, 2\prnumb+1\},\\
	& \{3\prnumb-2,\, 3\prnumb-1,\, 3\prnumb,\, 3\prnumb+1\},\\
	& \{4\prnumb-3,\, 4\prnumb-2,\, 4\prnumb-1,\, 4\prnumb,\, 4\prnumb+1\}.
\end{align*}
For instance, the reduction pattern \(\{\prnumb,\, \prnumb+1\}\) refers to the goods at positions \(\prnumb\) and \(\prnumb+1\) in the sorted list.
As seen above, these patterns are fully determined by the number of agents. In contrast, our approach introduces a more flexible notion of reduction patterns by allowing the last index to extend further to the right, depending on the values of the goods.
Specifically, we define a reduction pattern as \(\reductiontype^k\) where it has \textbf{static part}  
\[
\Srik{\prnumb}{\reductiontype^k}=\left\{k(\prnumb - 1) + 1,\, k(\prnumb - 1) + 2,\, \dots,\, k\prnumb\right\}
\]
and \textbf{dynamic index} \( x \geq k\prnumb+1 \), which is the largest index satisfying the condition that some agent values the set of goods  
$
\{\prgood_{k(\prnumb-1)+1}, \prgood_{k(\prnumb-1)+2}, \dots, \prgood_{k\prnumb}, \prgood_x\}
$
at least \(\alpha\). Note that, it might be the case that no such index exists. In that case, we say that $\reductiontype^k$ is \textbf{not applicable}.

In this paper, we introduce two additional reduction patterns, denoted by \(\rstar^1\) and \(\rstar^2\). These reductions extend the following fixed reduction patterns originally proposed by Garg and Taki~\cite{garg2020improved}, and Akrami and Garg~\cite{akrami2024breaking}:
$
\{1,\, 2\prnumb+1\},$ and  $\{1,\, 2\}.
$

As with earlier reduction patterns, we allow the last index in each set to shift dynamically based on the values of the goods. Specifically, we define \(\rstar^1,\rstar^2 \) as follows:

\begin{itemize}
	\item  For \(\rstar^1\), we set \textbf{static part} $\Srik{\prnumb}{\rstar^1}=\{1\}$, and \textbf{dynamic index} \(x \geq 2\prnumb+1\), which is the largest index such that some agent values the set \(\{\prgoodrat{1},\, \prgoodrat{x}\}\) at least \(\alpha\).
	\item For \(\rstar^2\), we set \textbf{static part} $\Srik{\prnumb}{\rstar^2}=\{1\}$, and \textbf{dynamic index} \(x \geq 2\), which is the largest index such that some agent values the set \(\{\prgoodrat{1},\, \prgoodrat{x}\}\) at least \(\alpha\).
\end{itemize}

Note that each reduction comes with a lower bound on its dynamic index.
These reductions play a central role in improving the approximation guarantee. The notion of \emph{inapplicability} naturally extends to \(\rstar^1\) and \(\rstar^2\) as well.

Now, we define a reduction in \cref{def:red}.
\begin{definition}\label{def:red}
	A \emph{reduction} is denoted by  
	$
	\reduction = \riki{\prinstance}{\reductiontype}{x}{\pragent_i}{\seinstance}
	$, where $\reductiontype \in \{\reductiontype^0,\,\reductiontype^1,\ldots,\,\reductiontype^4,\,\rstar^1,\,\rstar^2\}$ is a \emph{reduction pattern} and $x$ is dynamic index of $\reductiontype$ and \( \pragent_i \in \pragents \) is the agent to whom the reduction is applied.
	The resulting instance \( \seinstance = (\seagents, \segoods) \) is defined as:
	$
	\seagents = \pragents \setminus \{\pragent_i\},$ and $
	\segoods = \prgoods \setminus \{\prgood_k \mid k \in \Srik{\prnumb}{\reductiontype}\cup\{x\}\} .
	$
	Reduction \( 	\reduction \) is \textbf{valid} if the following conditions hold:
	\begin{itemize}
		\item 	$x$ satisfies the lower bound determined for the dynamic index of   $\reductiontype$,
		\item 	$
		\prvalu_\agenti\left(\{\prgood_k \mid k \in \Srik{\prnumb}{\reductiontype} \cup \{x\}\}\right) \geq \alpha,
		$ 
		\item $x=\prit$ or there is no agent  $\pragent_j \in \pragents $ such that	$\prvalu_j\left(\{\prgood_k \mid k \in \Srik{\prnumb}{\reductiontype}\cup\{x+1\}\}\right) \geq \alpha.
		$
	\end{itemize} 
\end{definition}

\Cref{reductionlem} follow from standard arguments used in previous studies~\cite{akrami2024breaking,akrami2023simplification,garg2019approximating,garg2020improved}, which show that $\reductiontype^0, \reductiontype^1, \reductiontype^2, \reductiontype^3$ and $\reductiontype^4$ do not decrease the maximin share value of the remaining agents over the remaining goods.

\begin{observation}\label{reductionlem}
	Let $\reduction=\riki{\prinstance}{\reductiontype}{x}{\pragent_i}{\seinstance}$ be a valid reduction, such that $\reductiontype \in \{\reductiontype^0,\reductiontype^1,\ldots,\reductiontype^4\}$. Then, for every agent $\pragent_j \in \seagents$, we have
	$
	\checkmms{\prvalu_{j}} \geq \hatmms{\prvalu_{j}}.
	$ 
\end{observation}

In contrast to rules \( \reductiontype^0 \) to \( \reductiontype^4 \),  reductions \( \rstar^1,\rstar^2 \) may decrease an agent's maximin share value. However, under certain conditions, these reductions also preserve the agents’ maximin share values. \cref{reductionlem2} introduces one of these conditions.

\begin{observation}\label{reductionlem2}
	Let $\prgoods$ be a set of goods, let $\prvalu$ be a valuation function on $\prgoods$, and let $\prgoodrat{x},\prgoodrat{y}$ be two distinct goods in $\prgoods$ such that $\prvalu(\{\prgoodrat{x},\prgoodrat{y}\}) \le \mms^{d}_{\prvalu}(\prgoods)$. Then $\mms^{d-1}_{\prvalu}\!\bigl(\prgoods \setminus \{\prgoodrat{x},\prgoodrat{y}\}\bigr) \ge \mms^{d}_{\prvalu}(\prgoods)$.
\end{observation}

We define a total order \(\succ\) over reduction patterns based on their static and dynamic indices as follows. We have
$
\reductiontype^0 \;\succ\; \reductiontype^1 \;\succ\; \reductiontype^2 \;\succ\; \reductiontype^3 \;\succ\; \reductiontype^4 \;\succ\; \rstar^1 \;\succ\; \rstar^2,
$
meaning that $\reductiontype^0$ has the highest priority and $\rstar^2$ the lowest.

\subsection{Reduction Sequence}
Typically, a reduction is viewed as an independent process, where the order and choice of reductions do not matter—only that the instance eventually becomes irreducible, meaning no further reductions can be applied. However, in this paper, we take a different approach by carefully considering the sequence of reductions in our analysis. Among the various ways to reduce the problem to an irreducible instance, we select a specific sequence that follows a structured pattern. This allows us to improve the performance of our algorithm.

\begin{definition}\label{validseq}
	Let \(\prinstance = (\pragents, \prgoods)\) be an ordered instance, and let $\prreduc \subseteq \{\reductiontype^0,\,\reductiontype^1,\ldots,\,\reductiontype^4,\,\rstar^1,\,\rstar^2\}$. 
	A sequence of valid reductions 
	$
	\boldsymbol{\rho} = (\rho_1, \dots, \rho_r)
	$
	on \(\prinstance\), where each reduction uses a pattern from \(\prreduc\), is called a \textbf{perfect sequence of reductions} (with respect to $\prreduc$), if the corresponding sequence 
	$
	\bigl\langle \reductiontype_1, x_1, \reductiontype_2, x_2, \ldots, \reductiontype_r, x_r \bigr\rangle
	$
	is lexicographically maximum (over all such sequence of valid reductions). \footnote{Here lexicographically maximum means comparing each reduction pattern \(\reductiontype_\ell\) according to the order \(\succ\), and comparing each index \(x_\ell\) by its numerical value.}  Here, \(\reductiontype_\ell\) denotes the reduction pattern of $\rho_\ell$, and $x_\ell$ denotes the dynamic index of $\rho_\ell$. 
\end{definition}

\begin{observation}\label{obs:perseq-bundles}
	Let \(\prinstance = (\pragents, \prgoods)\) be an ordered instance, and let \(\prreduc \subseteq \{\reductiontype^0, \reductiontype^1, \ldots, \reductiontype^4, \rstar^1, \rstar^2\}\).  
	Suppose \(\boldsymbol{\rho} = (\rho_1, \dots, \rho_r)\) is a perfect sequence of reductions (with respect to \(\prreduc\)). Then the following hold:
	\begin{enumerate}
		\item For every \(1 \le \ell \le r\), there does not exist a reduction \(\rho'\) such that \((\rho_1, \rho_2, \ldots, \rho_{\ell-1}, \rho')\) is a sequence of valid reductions on \(\prinstance\) and the reduction type of \(\rho'\) has strictly higher priority than that of \(\rho_\ell\).
		\item If \(\boldsymbol{\rho'} = (\rho'_1, \dots, \rho'_{r'})\) is another perfect sequence of reductions (with respect to \(\prreduc\)), then \(r = r'\), and for every \(1 \le \ell \le r\), the reduction \(\rho_\ell\) allocates exactly the same bundle as \(\rho'_\ell\).
	\end{enumerate}
\end{observation}

\begin{proof}
	By \cref{validseq}, the sequence \(\boldsymbol{\rho}\) corresponds to a sequence
	$
	\langle \reductiontype_1, x_1, \reductiontype_2, x_2, \ldots, \reductiontype_r, x_r \rangle, 
	$
	that is lexicographically maximum among all sequences of valid reductions.  
	
	For the first part, if a reduction \(\rho'\) as described existed, then replacing \(\rho_\ell\) with \(\rho'\) would yield a sequence whose corresponding sequence is lexicographically larger, contradicting maximality.  
	
For the second part, suppose $\boldsymbol{\rho'}$ is another perfect sequence. Since both sequences correspond to lexicographically maximum sequences, their lengths must coincide ($r = r'$). Moreover, by the uniqueness of the lexicographically maximum sequence, the tuples coincide entry by entry. Hence, for every position $\ell$, the reduction types and dynamic indices of $\rho_\ell$ and $\rho'_\ell$ must be the same, and therefore $\rho_\ell$ and $\rho'_\ell$ allocate the same bundle.
\end{proof}


Intuitively, after applying a sequence of reductions, we obtain  a reduced instance of the problem. At this point, we analyze how the reductions have affected the agents' values for the good in rank $2\numb+1$ and classify them into two groups. We then choose between two algorithms based on the sizes of these two groups. Each of these algorithms prioritizes one of the groups.
The main challenge, however, is that for our approximation guarantee to hold, this prioritization must be considered not only in the second stage but also during the reduction process itself.  In other words, when multiple reduction choices are available, we must select agents based on this prioritization. At first glance, this may seem paradoxical, as the prioritization depends on the reduction sequence. Surprisingly, we demonstrate that a carefully designed selection strategy allows us to achieve this goal. This idea is built upon \Cref{thm:main}, which we state below and prove in this section.

\begin{lemma}\label{thm:main}
	Let $\prinstance = (\pragents, \prgoods)$ be an ordered instance and $\prreduc \subseteq \{\reductiontype^0,\,\reductiontype^1,\ldots,\,\reductiontype^4,\,\rstar^1,\,\rstar^2\}$. Let 
	$
	\boldsymbol{\rho} = (\rho_1, \rho_2, \dots, \rho_r)
	$
	be a perfect sequence of reductions (with respect to $\prreduc$) on $\prinstance$.
	Additionally, let $\pragents^1$ and $\pragents^2$ be a partition of $\pragents$ into two subsets. Then, there exists another perfect sequence of reductions (with respect to $\prreduc$), 
	$
	\boldsymbol{\rho'} = (\rho'_1, \rho'_2, \dots, \rho'_r),
	$
	such that, after applying $\boldsymbol{\rho'}$, for any \(\pragent_x\in\pragents^1\) who has not received a bundle and any \(\pragent_y\in\pragents^2\) who has, we have:
	$
	{\prvalu}_{x}(\bundle^{y}) < \alpha,
	$
	where $\bundle^{y}$ is the bundle allocated to agent $\pragent_y$ during the reduction process.
\end{lemma}
\begin{proof}
For each $1 \le j \le r$, let $\reductiontype_j$ be the reduction pattern, $x_j$ the dynamic index, and $\bundle_j$ the bundle associated with reduction $\rho_j$. 
We construct a bipartite graph where the first set of nodes represents these bundles, and the second set of nodes represents the agents in $\pragents$. An edge exists between a bundle node \( \bundle_j \) and an agent node \( \pragent_i \) if \(\prvalu_\agenti(\bundle_j) \geq \alpha \).  
	
A perfect sequence of reductions \( \boldsymbol{\rho} \) corresponds to a matching that saturate all nodes in the first part of the graph. 
We now show that any matching that saturate all nodes in the first part of the graph, yields a perfect sequence of reductions. Let \( \boldsymbol{\rho}' = (\rho'_1, \dots, \rho'_r) \) be the sequence obtained from this matching. By \cref{validseq}, it is enough to show that each reduction \( \rho'_j \) is valid. The remaining conditions for  a perfect sequence depend only on the reduction patterns and dynamic indices, and not on the specific agents involved. 

Note that for each $\ell$, the reduction pattern and the dynamic index of $\rho'$ remain the same as those of $\rho$, since we only reallocate bundles while keeping the goods in the bundles unchanged. 
Suppose, for contradiction, that some reduction in the sequence is not valid. Let \( \rho'_\ell \) be the first such reduction (i.e., the one with the smallest index \( \ell \)). 

The invalidity of \( \rho'_\ell \) implies that there exists another valid reduction \( \rho'' \) with the same reduction pattern as \( \rho'_\ell \) but with a strictly larger dynamic index \( x'' > x_\ell \), applied to the instance obtained after executing \( \rho'_1, \dots, \rho'_{\ell-1} \). Since \( \rho'_\ell \) is the first invalid reduction in the sequence, all reductions in the updated sequence \( (\rho'_1, \dots, \rho'_{\ell-1}, \rho'') \) are valid. However the corresponding sequence 
$
\bigl\langle \reductiontype_1, x_1, \reductiontype_2, x_2, \ldots, \reductiontype_{\ell-1}, x_{\ell-1}, \reductiontype_{\ell}, x'' \bigr\rangle
$
is lexicographically larger than 
$
\bigl\langle \reductiontype_1, x_1, \reductiontype_2, x_2, \ldots, \reductiontype_{r}, x_r \bigr\rangle
$
contradicting the definition of a perfect sequence.

Among all possible matchings that saturate all nodes in the first part of the graph, we choose one that maximizes the number of matched agents in $\pragents^1$. We argue that this matching satisfies the desired properties of \Cref{thm:main}.

Suppose for contradiction that there exists an agent $\pragent_x\in\pragents^1$ who has not received a bundle while there exists $\pragent_y\in\pragents^2$ who has received one, with ${\prvalu}_{x}(\bundle^{y}) \ge \alpha$. Replacing $\pragent_y$ by $\pragent_x$ yields another matching that saturate all nodes in the first part of the graph, covering more agents in $\pragents^1$, contradicting our construction.
\end{proof}

	\section{\boldmath Calibration}\label{sec:virtual-value}

To simplify the analysis, we define \emph{calibration functions} for certain agents. These functions make slight adjustments to the agents’ valuation functions, while ensuring that the change in their maximin share remains small and within a known bound. These modifications are purely analytical and do not affect the actual execution of the algorithm.

\begin{definition}  \label{def:adj}
Let \(\prvalu\) be a valuation over goods \(\prgoods\) such that for all $\prgood \in \prgoods$, we have $\prvalu(\{\prgood\})\le 1$, and let \(\hat\vF: [0, 1] \to [0, 1]\) be a non-decreasing function with \(\hat\vF(x) \le x\) for all \(x\). The \textbf{calibration} of \(\prvalu\) by \(\hat\vF\), denoted by \(\prauxval{\hat\vF}\), is the \textbf{additive} valuation defined by
$
\praux{\hat\vF}{\{\prgood\}} = \hat\vF(\prvalu(\{\prgood\}))$ for all  $\prgood \in \prgoods$.

\end{definition}  

Throughout our analysis, different agents may use different calibration functions. Here we define the specific calibration functions used in this paper.

\begin{definition}\label{def:F}
	Define the functions:
	\[
	\begin{array}{@{}lll}
		(\text{For } 0\leq \lambda \leq \tfrac{4\alpha}{3}-1) \quad &\vF_\lambda(x) = &
		\begin{cases}
			x, & x \in [0,\, \frac{\alpha}{3} - \lambda) \\[0.6em]
			\max(\frac{\alpha}{3} - \lambda,\, x - \lambda), & x \in [\frac{\alpha}{3} - \lambda,\, 1 - \frac{2\alpha}{3}) \\[0.6em]
			\max(1 - \frac{2\alpha}{3} - \lambda,\, x - \frac{3\lambda}{2}), & x \in [1 - \frac{2\alpha}{3},\, 1 - \frac{\alpha}{3} - \frac{\lambda}{2}) \\[0.6em]
			\max(1 - \frac{\alpha}{3} - 2\lambda,\, x - 3\lambda), & x \in [1 - \frac{\alpha}{3} - \frac{\lambda}{2},\, 1]
		\end{cases}
		\\[4em]
		\quad
		&\vG(x) = &
		\begin{cases}
			x, & x \in [0,\, 2 - \frac{7\alpha}{3}) \\[0.6em]
			\max(2 - \frac{7\alpha}{3},\, x - \frac{4\alpha}{3} + 1), & x \in [2 - \frac{7\alpha}{3},\, 2 - \frac{13\alpha}{6}) \\[0.6em]
			\max(3 - \frac{7\alpha}{2},\, x - \frac{8\alpha}{3} + 2), & x \in [2 - \frac{13\alpha}{6},\, 1]
		\end{cases}
\\[3em]
	(\text{For } 0\le\lambda \le \tfrac{1}{2}) \quad
	&\vHA_\lambda(x) =&
	\begin{cases}
		x, & x \in [0,\, \frac{1}{2} - \lambda) \\[0.6em]
		\max(\frac{1}{2} - \lambda,\, x - 2\lambda), \qquad \qquad& x \in [\frac{1}{2} - \lambda,\, 1]
	\end{cases}
\\[2em]
(\text{For } 0\le\lambda \le 2(1-\alpha)) &\vHB_\lambda(x) =&
	\begin{cases}
		x, & x \in [0,\, 2-2\alpha - \lambda) \\[0.6em]
		\max(2-2\alpha - \lambda,\, x - 2\lambda), & x \in [2-2\alpha - \lambda,\, 1] 
	\end{cases}
			\end{array}
	\]
	For convenience, we denote \( \vF_{\tfrac{4\alpha}{3} - 1} \) by $\FJ$ throughout the remainder of the paper.
\end{definition}
\begin{figure}[t]
	\centering
	\scalebox{1.3}{
		\begin{minipage}[b]{0.4\linewidth}
			\includegraphics[width=\linewidth]{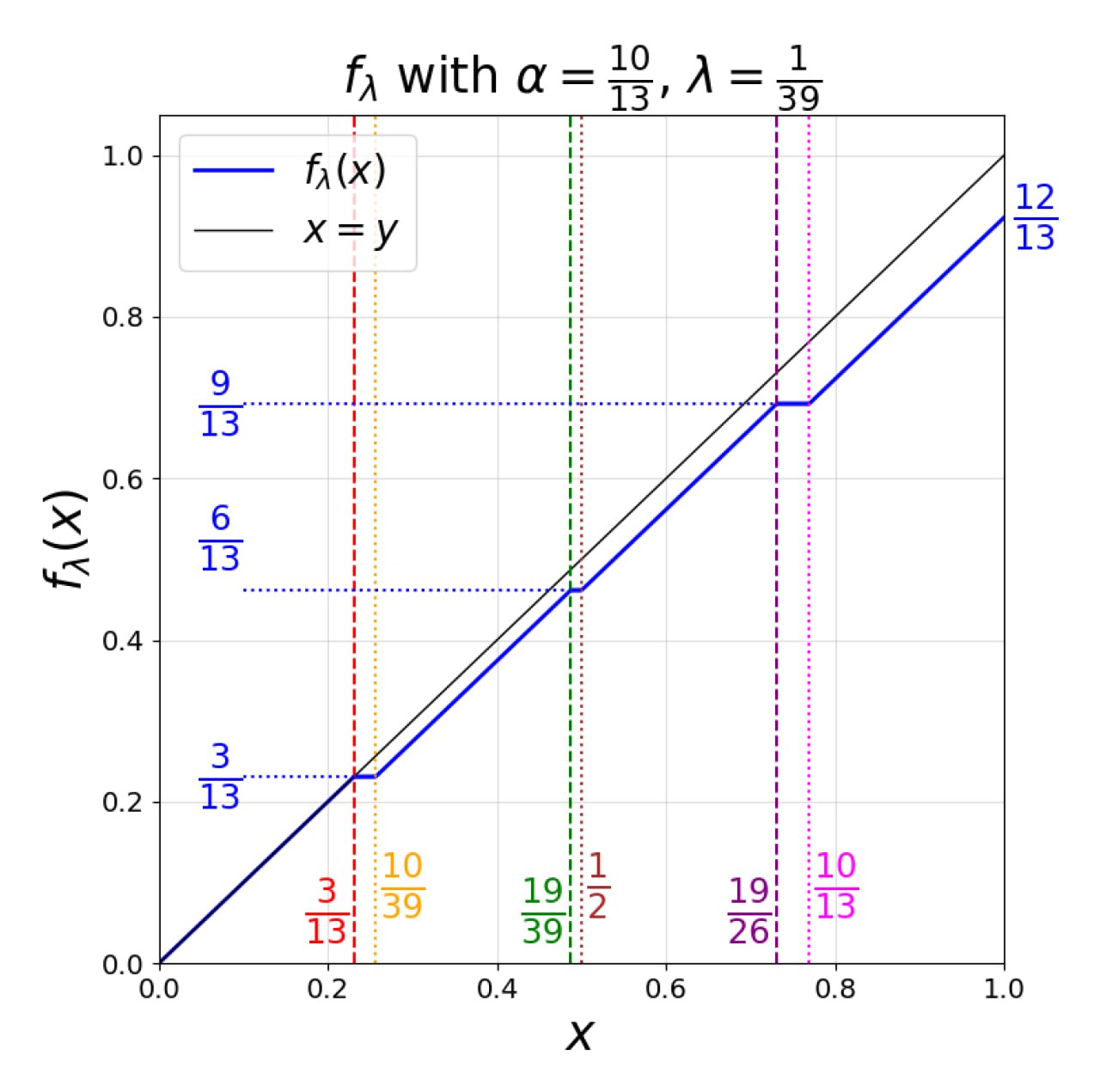}
		\end{minipage}
		\hfill
		\begin{minipage}[b]{0.385\linewidth}
			\includegraphics[width=\linewidth]{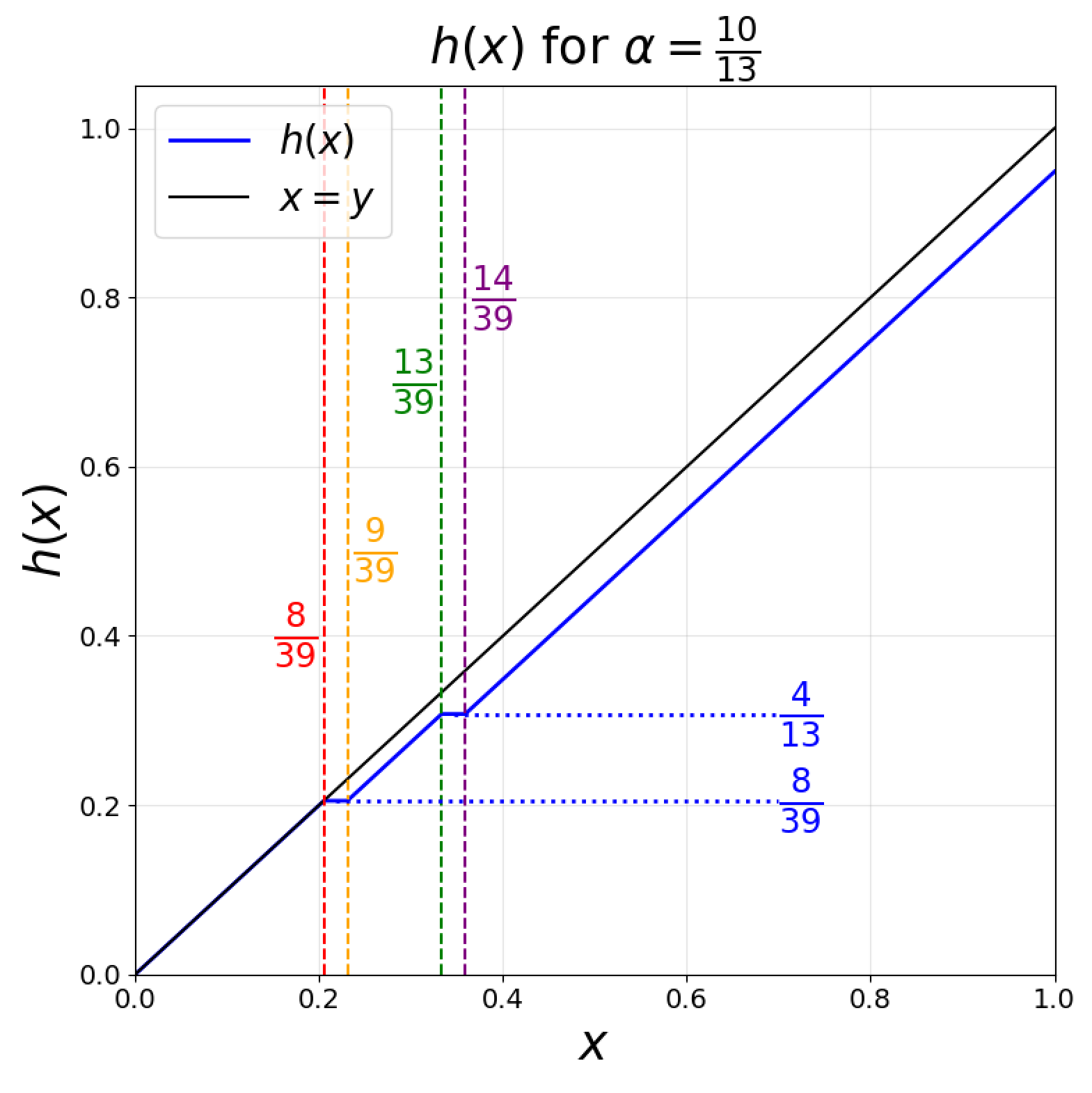}
		\end{minipage}
	}
	\caption{Plot of the calibration functions $f_\lambda$ and $h$ for $\alpha = \tfrac{10}{13}$.}
	\label{fig:plot-fh}
\end{figure}

\begin{figure}[t]
	\centering
	\scalebox{1.3}{
		\hspace{-0.8cm}
		\begin{minipage}[b]{0.4\linewidth}
			\includegraphics[width=\linewidth]{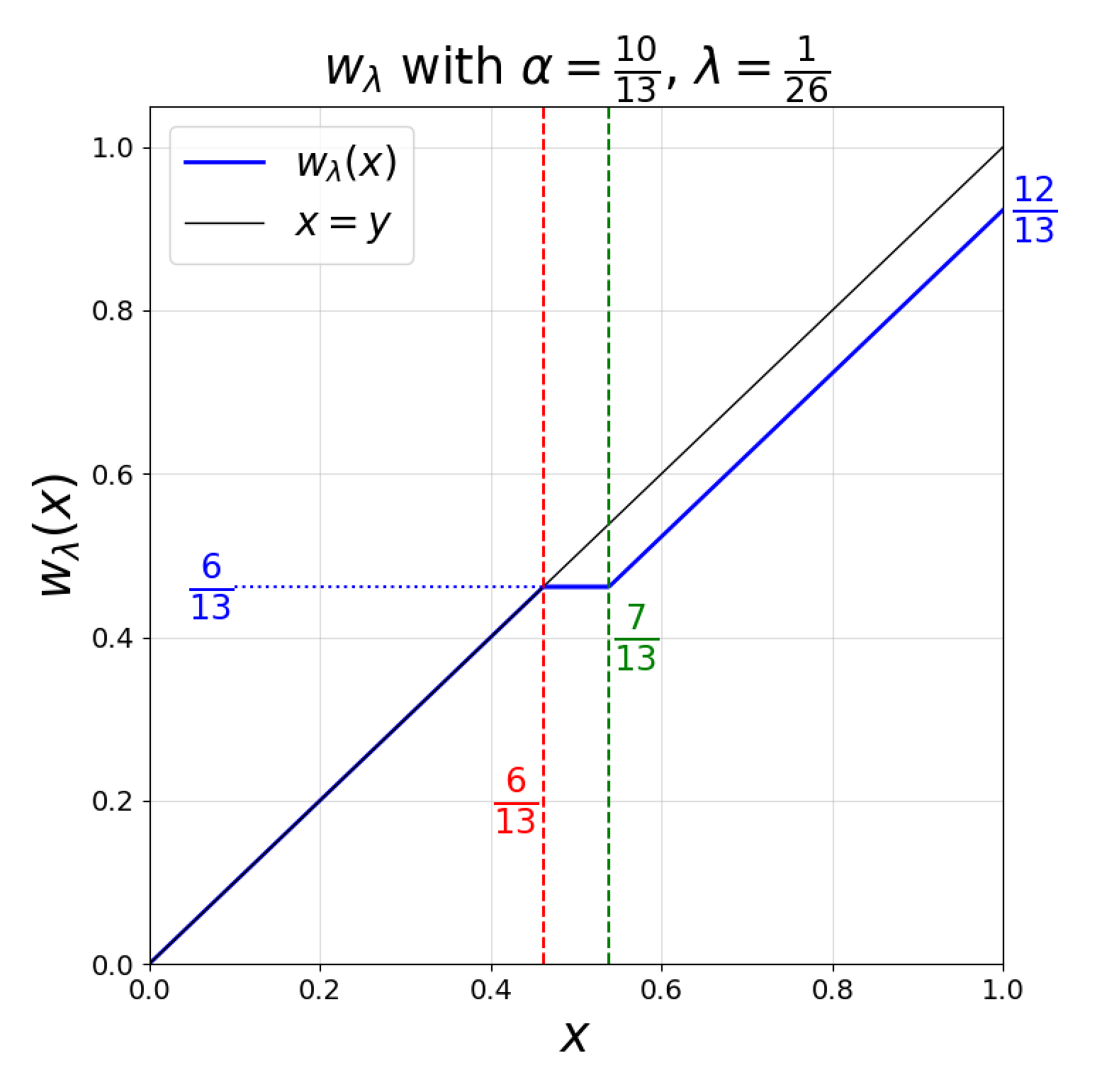}
		\end{minipage}
		\hfill
		\begin{minipage}[b]{0.41\linewidth}
			\includegraphics[width=\linewidth]{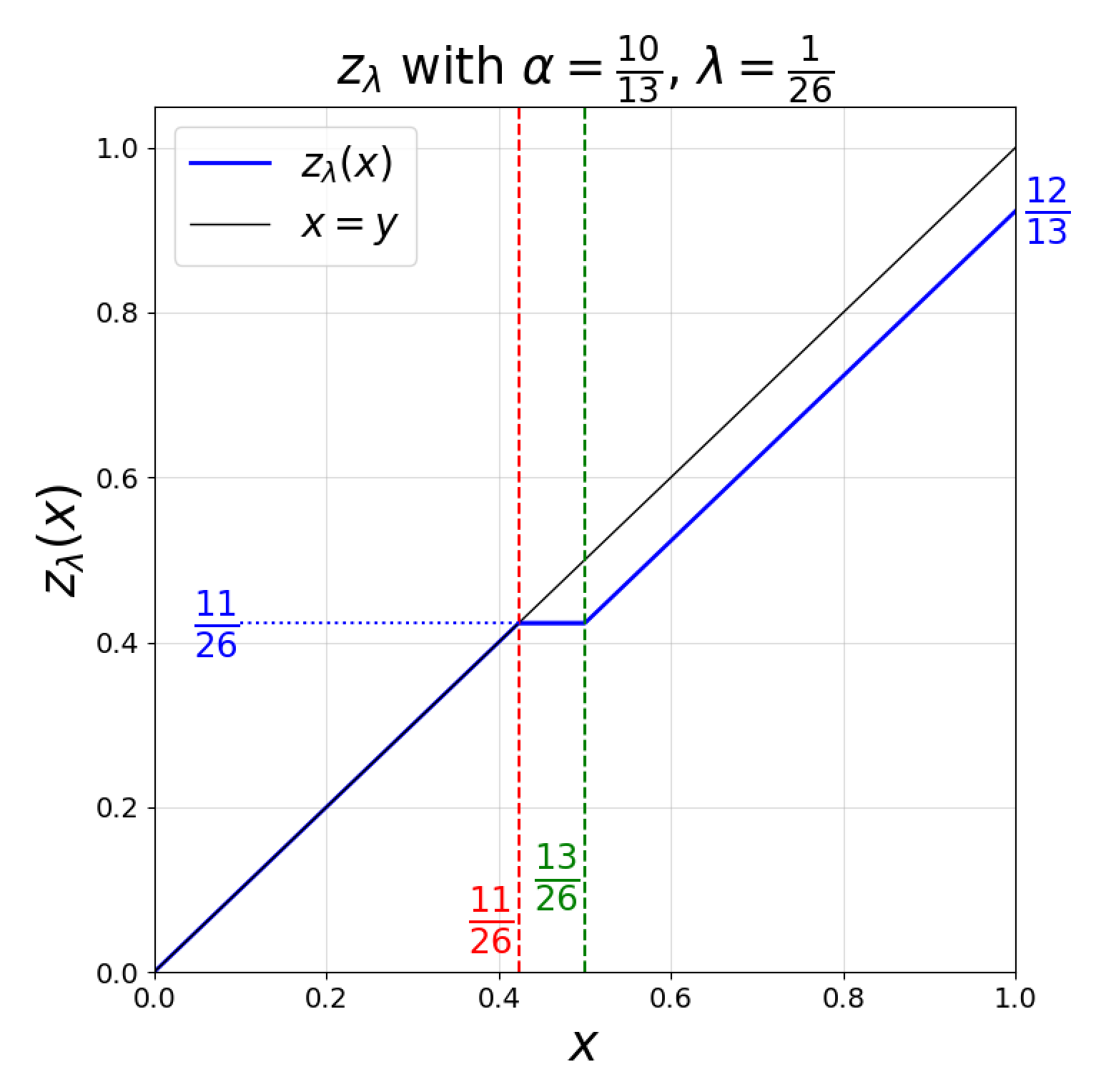}
		\end{minipage}
	}
	\caption{Plot of the calibration functions $w_\lambda$ and $z_\lambda$ for $\alpha = \tfrac{10}{13}$.}
	\label{fig:plot-wz}
\end{figure}

\newpage
\Cref{fig:plot-fh,fig:plot-wz} illustrates the calibration functions. Note that all functions satisfy $\hat\vF(x) \le x$ for every \( x \in [0,1] \). In \cref{apx:adjusting}, we present lemmas establishing lower bounds on the maximin share values under these calibration functions. These bounds are summarized in \Cref{tab:mms-d-bounds,tab:calibrated-mms-bounds}.

\begin{table}[H]
	\centering
	\renewcommand{\arraystretch}{1,6}
	\begin{tabular}{>{\raggedright\arraybackslash}p{1.6cm}>{\centering\arraybackslash}p{2cm}>{\raggedright\arraybackslash}p{5cm}>{\raggedright\arraybackslash}p{5cm}}
		\toprule
		\rowcolor{headerblue!60}
		\textbf{Lemma} & \textbf{Function} & \textbf{Precondition} & \textbf{$\MMS$ Bound} \\ 
		\midrule
		\rowcolor{rowblue1}\cref{lem:virt-mms-before} & 
		$\vF_\lambda$ & 
		$\mms^d_{\prvalu} (\prgoods) \geq 1$ & 
		$\prauxmms{\vF_\lambda}{d}{\prgoods} \geq 1 - 3\lambda$ \\ 
		\rowcolor{rowblue2} \cref{lem:virtG-mms} & 
		$\vG$ & 
		$\mms^d_{\prvalu} (\prgoods) \geq 4(1-\alpha)$ & 
		$\prauxmms{\vG}{\prgoods}{d} \geq 4(2 - \frac{7\alpha}{3})$ \\ 
		\rowcolor{rowblue1} \cref{lem:HA-mms-before} & 
		$\vHA_\lambda$ & 
		$\mms^d_{\prvalu} (\prgoods) \geq 1$ & 
		$\prauxmms{\vHA_\lambda}{d}{\prgoods} \geq 1 - 2\lambda$ \\ 
		\rowcolor{rowblue2} \cref{lem:HB-mms-before} & 
		$\vHB_\lambda$ & 
		$\mms^d_{\prvalu} (\prgoods) \geq 4(1-\alpha)$ & 
		$\prauxmms{\vHB_\lambda}{\prgoods}{d} \geq 4(1 - \alpha) - 2\lambda$ \\ 
		\bottomrule
	\end{tabular}
	\caption{Bounds on the maximin share (MMS) under calibration functions. For any instance satisfying the given preconditions, the corresponding lemma establishes the stated lower bound.}
	\label{tab:mms-d-bounds}
\end{table}

\begin{lemrep}\label{lem:calib-reduction}
	Let $\prinstance = (\pragents,\prgoods)$ be an ordered instance, and let $		 \prreduc_1 = [\reductiontype^0 \succ \reductiontype^1 \succ \reductiontype^2 \succ \rstar^1]$ and $		\prreduc_2 = [\reductiontype^1 \succ \reductiontype^2 \succ \reductiontype^3 \succ \reductiontype^4 \succ\rstar^2]$.  Assume that $\seinstance=(\seagents,\segoods)$ is the result of applyinng a sequence of valid reductions with respect to either $\prreduc_1$ or $\prreduc_2$. Then, the conditions shown in \cref{tab:calibrated-mms-bounds} satisfy. 
	\begin{table}[h!]
		\centering
		\renewcommand{\arraystretch}{2}
\scalebox{0.95}{	
		\begin{tabular}{>{\centering\arraybackslash}p{0.8cm}>{\centering\arraybackslash}p{0.8cm}>{\raggedright\arraybackslash}p{2.2cm}>{\raggedright\arraybackslash}p{2.5cm}>{\raggedright\arraybackslash}p{4cm}>{\raggedright\arraybackslash}p{4.5cm}}
			\toprule
			\rowcolor{headerblue!60}
			\textbf{Func} & \textbf{\small{Red}} & \textbf{Prec 1} & \textbf{Prec 2}&\textbf{Prec 3} & \textbf{MMS Guarantee} \\ 
			\midrule
			\rowcolor{rowblue1}
			$\vF_\lambda$ & $\prreduc_1$ &
			$ \lambda \le \tfrac{4\alpha}{3}-1$ &
			$\hatmms{\prvalu_\agenti} \ge 1$  & 
			$\prvalu_\agenti(\{\prgoodrat{1}\}) \le 1 - \tfrac{\alpha}{3} + \lambda$ &
			$\checkmms{\prauxfun{\vF_\lambda}{\agenti}} \ge 1 - 3\lambda$ \\ 
			\rowcolor{rowblue2} 
			$\vHA_\lambda$ & $\prreduc_2$ &
			$\lambda \le \tfrac{1}{2}$ & 
			$\hatmms{\prvalu_\agenti} \ge 1$
				&$\prvalu_\agenti(\{\prgoodrat{1}\}) \le \tfrac{1}{2} + \lambda$ & 
			$\checkmms{\prauxfun{\vHA_\lambda}{\agenti}} \ge 1 - 2\lambda$ \\ 
			
			\rowcolor{rowblue1} 
			$\vHB_\lambda$ & $\prreduc_2$ &
			$\lambda \le 2(1-\alpha)$ 				&$\hatmms{\prvalu_\agenti} \ge 4(1-\alpha)$ &$\prvalu_\agenti(\{\prgoodrat{1}\}) \le 2(1-\alpha) + \lambda$ & 
			$\checkmms{\prauxfun{\vHB_\lambda}{\agenti}} \ge 4(1-\alpha) - 2\lambda$ \\ 
			\bottomrule
		\end{tabular}}
		\caption{Calibrated $\MMS$ bounds under various reduction sequences. If an instance satisfies both preconditions, the stated guarantee holds for the calibrated $\MMS$ after applying the reductions.}
		\label{tab:calibrated-mms-bounds}
	\end{table}
\end{lemrep}
\begin{appendixproof}
	
\textbf{Case 1:} Since $\hatmms{\hat{v_i}} \ge 1$, by \cref{lem:virt-mms-before}, we have $\hatmms{\prauxfun{\vF_\lambda}{\agenti}} \ge 1 - 3\lambda$. Therefore it suffices to show that no reduction makes the $\MMS$ under $\prauxfun{\vF_\lambda}{\agenti}$ less than $1 - 3 \lambda$. By \cref{reductionlem}, the only reduction that can decrease the $\MMS$ value of an agent is $\rstar^1$. Suppose that the rule allocates goods $\prgood_1$ and $\prgood_x$, where $x \ge 2 \prnumb + 1$. As Precondition 2 holds, we have $\prvalu_\agenti(\prgood_1) \le 1 - \tfrac{\alpha}{3} + \lambda$, and $\vF_\lambda(1 - \tfrac{\alpha}{3} + \lambda) = 1 - \tfrac{\alpha}{3} - 2\lambda$. Furthermore, since $\reductiontype^2$ is not applicable, we have $\prvalu_\agenti(\prgood_x) \le \tfrac{\alpha}{3}$, and $\vF_\lambda(\tfrac{\alpha}{3}) = \tfrac{\alpha}{3} - \lambda$. Hence, $\prauxfun{\vF}{\agenti}(\{g_1, g_x\}) \le 1 - \tfrac{\alpha}{3} - 2\lambda + \tfrac{\alpha}{3} - \lambda = 1 - 3\lambda$, and by \cref{reductionlem2}, this reduction does not decrease $\MMS$ under $\prauxfun{\vF_\lambda}{\agenti}$. Therefore, in the final instance $\segoods$, we have $\checkmms{\prauxfun{\vF_\lambda}{\agenti}} \ge 1 - 3\lambda$.

\textbf{Case 2:} Since $\hatmms{\hat{v_i}} \ge 1$, by \cref{lem:HA-mms-before}, we have $\hatmms{\prauxfun{\vH_\lambda}{\agenti}} \ge 1 - 2\lambda$. Therefore it suffices to show that no reduction makes the $\MMS$ under $\prauxfun{\vH_\lambda}{\agenti}$ less than $1 - 2\lambda$. By \cref{reductionlem}, the only reduction that can decrease the $\MMS$ value of an agent is $\rstar^2$. Suppose that the rule allocates goods $\prgood_1$ and $\prgood_x$, where $x \ge 2$. As Precondition 2 holds, we have $\prvalu_\agenti(\prgood_1) \le \tfrac{1}{2} + \lambda$, and $\vH_\lambda(\tfrac{1}{2} + \lambda) = \tfrac{1}{2} - \lambda$. Hence, $\prauxfun{\vH}{\agenti}(\{g_1, g_x\}) \le 2 (\tfrac{1}{2} - \lambda) = 1 - 2\lambda$, and by \cref{reductionlem2}, this reduction does not decrease $\MMS$ under $\prauxfun{\vH_\lambda}{\agenti}$. Therefore, in the final instance $\segoods$, we have $\checkmms{\prauxfun{\vH_\lambda}{\agenti}} \ge 1 - 2\lambda$.

\textbf{Case 3:} Since $\hatmms{\hat{v_i}} \ge 4(1 - \alpha)$, by \cref{lem:HB-mms-before}, we have $\hatmms{\prauxfun{\vHB_\lambda}{\agenti}} \ge 4(1 - \alpha) - 2\lambda$. Therefore it suffices to show that no reduction makes the $\MMS$ under $\prauxfun{\vHB_\lambda}{\agenti}$ less than $4(1 - \alpha) - 2\lambda$. By \cref{reductionlem}, the only reduction that can decrease the $\MMS$ value of an agent is $\rstar^2$. Suppose that the rule allocates goods $\prgood_1$ and $\prgood_x$, where $x \ge 2$. As Precondition 2 holds, we have $\prvalu_\agenti(\prgood_1) \le 2(1 - \alpha) + \lambda$, and $\vHB_\lambda(2(1 - \alpha) + \lambda) = 2(1 - \alpha) - \lambda$. Hence, $\prauxfun{\vHB}{\agenti}(\{g_1, g_x\}) \le 4(1 - \alpha) - 2\lambda$, and by \cref{reductionlem2}, this reduction does not decrease $\MMS$ under $\prauxfun{\vHB_\lambda}{\agenti}$. Therefore, in the final instance $\segoods$, we have $\checkmms{\prauxfun{\vHB_\lambda}{\agenti}} \ge 4(1 - \alpha) - 2\lambda$.
\end{appendixproof}

Finally, we introduce a calibration function with a structure that differs from the earlier ones. This function will be used in \cref{sec:alg-N2} to simplify the analysis.

\begin{definition}\label{def:norm}
Let \( v \) be a valuation over \(\prgoods\) such that \(\mms^d_v(\prgoods) \geq \lambda\), and let \( P = (P_1, \ldots, P_d) \) be a maximin partition of \(\prgoods\) under \(v\). Define the multiset
\[
\mathcal{M} = \left\{ \frac{v(\{\prgood\}) \cdot \lambda}{v(P_j)} \;\middle|\; 1 \le j \le d,\; \prgood \in P_j \right\}.
\]
Now define \(\text{normalized}^{\,d}_{\lambda}(v, \prgoods)\) as the valuation obtained by assigning the values in \(\mathcal{M}\) to the goods in \(\prgoods\), preserving their original order under \(v\). That is, if \(\prgood_i\) is the \(i\)-th highest-valued good under \(v\), then value of \(\{\prgood_i\}\) in \(\text{normalized}^{\,d}_{\lambda}(v, \prgoods)\)  is the \(i\)-th highest value in \(\mathcal{M}\).
\end{definition}
\begin{observation}\label{lem:norm}
	Let $v$ be a valuation over \(\prgoods\) such that \(\mms^d_v(\prgoods) \geq \lambda\). Then  \(\nu = \text{normalized}^{\,d}_{\lambda}(v, \prgoods)\) satisfies the following conditions: 
	\begin{enumerate}
		\item There exists a partition \((P_1, \ldots, P_d)\) of \(\prgoods\) such that \(\nu (P_j) = \lambda\) for every \(1 \leq j \leq d\).
		\item  \(\nu (\{\prgood\}) \leq v(\{\prgood\})\) for every \(\prgood \in \prgoods\).
		\item  $v$ and $\nu$ rank the goods in the same order.
	\end{enumerate} 

\end{observation}

	\clearpage
\section{A $(\tfrac{10}{13})$-$\MMS$ Allocation Algorithm} \label{sec:mainalg}

Algorithm \ref{alg:main} provides the pseudocode for our allocation algorithm. By the scaling assumption, the maximin share of all the agents in the initial instance is equal to $1$. After applying the reductions, we denote the resulting instance by \( \instance = (\agents, \goods) \), which is irreducible with respect to reduction rules \( \reductiontype^0, \reductiontype^1, \reductiontype^2 \) and \( \rstar^1 \).   Along with the reductions, we also divide the agents into two subsets: green agents (\( \Nyek \)) and red agents (\( \Ndo \)). Depending on the size of these subsets, we consider two cases: whether \( |\Nyek| \geq \tfrac{\Pnumb}{\sqrt 2} \) or not. For each case, we design a separate algorithm and prove that it guarantees an \( \alpha \)-\( \MMS \) allocation.

\begin{algorithm}
	\caption{ $(\tfrac{10}{13})$-$\MMS$ Allocation}
	\textbf{Input:} $\Pinstance = (\Pagents, \Pgoods)$ \\
	\textbf{Output:} Allocation satisfying  $(\tfrac{10}{13})$-$\MMS$
	\begin{algorithmic}[1]
		\State $\instance,\Nyek,\Ndo \leftarrow$ \textsc{Primary-Reductions}($\Pinstance$) \CComment{See \Cref{algo:pr}}
		\If{$|\Nyek| \geq \tfrac{\Pnumb}{\sqrt 2}$}
		\State Run \textsc{Algorithm-Case1}$(\instance, \Nyek,\Ndo)$  \CComment{See \Cref{alg:allocation1}}
		\Else
		\State Run \textsc{Algorithm-Case2}$(\instance, \Nyek,\Ndo)$  \CComment{See \cref{algo:N2}}
		\EndIf
	\end{algorithmic}
	\label{alg:main}
\end{algorithm}
\vspace{-0.5cm}
\subsection{ Primary Reductions}\label{sec:init}

In the primary reductions, we first find a perfect sequence of reductions with respect to 
$
\reductiontype^0 \succ \reductiontype^1 \succ \reductiontype^2 \succ \rstar^1.
$ Let $\goods$ denote the set of goods obtained after applying all reductions, and let $\agents$ be the resulting set of agents. We partition the agents in $\Pagents$ into two subsets, $\Nyek$ and $\Ndo$, as follows:
\[
\Nyek = \left\{ \agent_i \in \Pagents \mid \valu_\agenti(\{\good_{2\numb+1}\}) \ge 1 - \alpha \right\}
\quad \text{and} \quad
\Ndo = \left\{ \agent_i \in \Pagents \mid \valu_\agenti(\{\good_{2\numb+1}\}) < 1 - \alpha \right\}.
\]

Note that we partition agents in $\Pagents$, not $\agents$.  
Using \cref{thm:main}, we can modify the reduction sequence so that the primary reductions give priority to one of \( \Nyek \) or \( \Ndo \), based on the size of \( \Nyek \) as follows:  
(i) If \( |\Nyek| \geq \tfrac{\Pnumb}{\sqrt{2}} \) we follow a perfect reduction sequence that prioritizes agents in \( \Ndo \).
(ii) If \( |\Nyek| < \tfrac{\Pnumb}{\sqrt{2}} \), we choose a perfect reduction sequence that prioritizes agents in \( \Nyek \).  
\cref{algo:pr} presents a pseudo code of our algorithm for the primary reductions.
\begin{algorithm}[h]
	\caption{$\textsc{Primary-Reductions}$}
	\label{algo:pr}
	\textbf{Input:} $\Pinstance = (\Pagents, \Pgoods)$
	
	\textbf{Output:} $\instance, \Nyek, \Ndo$
	\begin{algorithmic}[1]
	    \State $\prreduc \leftarrow [\reductiontype^0 \succ \reductiontype^1 \succ \reductiontype^2 \succ \rstar^1]$ \CComment{Reduction order}
		\State $\boldsymbol{\rho} \leftarrow$ Perfect sequence of reductions on $\Pinstance$ with respect to $\prreduc$  \CComment{See \cref{validseq}}
		\State $\goods \leftarrow$ Set of goods after applying $\boldsymbol{\rho}$ on $\Pinstance$ 
		\State $\numb \leftarrow \Pnumb - |\boldsymbol{\rho}|$ \CComment{Number of remaining agents}
		\State $\Nyek \leftarrow \{\agent_i \in \Pagents \;\;|\;\; \valu_\agenti(\{\good_{2\numb+1}\}) \ge 1 - \alpha\}$ 
		\State $\Ndo \leftarrow \{\agent_i \in \Pagents \;\;|\;\; \valu_\agenti(\{\good_{2\numb+1}\}) < 1 - \alpha\}$ 
		\If{$|\Nyek| \ge \frac{\Pnumb}{\sqrt{2}}$}
		\State $\boldsymbol{\rho}' \leftarrow$ Reallocate $\boldsymbol{\rho}$ with maximum number of red agents 
		\Else
		\State $\boldsymbol{\rho}' \leftarrow$ Reallocate $\boldsymbol{\rho}$ with maximum number of green agents
		\EndIf \CComment{See \cref{thm:main}}
		\State $\instance \leftarrow$ Output of $\boldsymbol{\rho}'$ on $\Pinstance$. 
		\State \Return $\instance, \Nyek, \Ndo$ 
	\end{algorithmic}
\end{algorithm}

Recall that, after the reductions, in instance~$\instance$ the agents’ maximin share values may no longer be at least $1$, since $\rstar^1$ can decrease these values. We prove \Cref{lem:N-minus} for the agents whose maximin share values decrease due to the primary reductions.

\begin{lemma}\label{lem:N-minus}
Let $\agent_i$ be an agent in $\agents$ whose $\dotmms{\valu_\agenti} < 1$ after the primary reductions. Then there exists a real number $s_i > 0$ such that the following conditions hold:
	\begin{align}
		&\valu_\agenti(\{\good_1\}) - (1 - \frac{\alpha}{3}) \le s_\agenti < \frac{4\alpha}{3} - 1,\label[ineq]{conds1}\\
		\forall \good \in \goods,\qquad\qquad &\valu_\agenti(\{\good\}) \notin \Bigl[\frac{4\alpha}{3} - 1 - s_i,\; \frac{\alpha}{3} - s_i\Bigr],\label{conds2}\\
		\mbox{$\forall\,\lambda \in [s_{\agenti},\,\tfrac{4\alpha}{3}-1]$},\qquad\qquad &\dotmms{\auxfun{\LM}{\agenti}} \ge 1 - 3\lambda. \label[ineq]{conds3}
	\end{align}
\end{lemma}
\begin{proof}
	Consider the first reduction $\riki{\prinstance}{\reductiontype}{x}{\agent_j}{\seinstance}$ such that $\agent_i$'s maximin share drops below 1. That is, before the reduction, $\hatmms{v_\agenti} \geq 1$, but after the reduction, $\checkmms{v_\agenti} < 1$.
	By \cref{reductionlem}, among all reduction patterns, only $\rstar^1$ can decrease the maximin share of an agent. So we assume this reduction allocated two goods, $\prgoodrat{1}$ and $\prgoodrat{x}$, and define 
	$
	s_i = \valu_\agenti(\{\prgoodrat{1}\}) - \left(1 - \tfrac{\alpha}{3}\right).
	$
	
	We first show that  \( 0 < s_i < \tfrac{4\alpha}{3} - 1 \). The drop in the maximin share implies that $\valu_\agenti(\{\prgoodrat{1},\,\prgoodrat{x}\}) > 1$, by \cref{reductionlem2}. Meanwhile, because the reduction $\reductiontype^2$ does not apply here, we have $\valu_\agenti(\{\prgoodrat{x}\}) < \tfrac{\alpha}{3}$. Putting these together, we conclude that $\valu_\agenti(\{\prgoodrat{1}\}) > 1 - \tfrac{\alpha}{3}$, so $s_i > 0$.
	On the other hand, the fact that $\reductiontype^0$ does not apply means that $\valu_\agenti(\{\prgoodrat{1}\}) < \alpha$, which gives $s_i < \tfrac{4\alpha}{3} - 1$. Therefore,  $s_i$ is positive and lies in the desired range.
	Since all goods in $\goods$ are drawn from $\prgoods$, we have
	\begin{align*}
	\valu_\agenti(\{\good_1\}) &\le \valu_\agenti(\{\prgoodrat{1}\})\\ &= (1 - \tfrac{\alpha}{3}) + s_i.
\end{align*}
	This confirms that \Cref{conds1} holds.
	
	To show \Cref{conds2}, we use two properties of $\rstar^1$: the pair $\prgoodrat{1}$ and $\prgoodrat{x}$ has total value above 1, while replacing $\prgoodrat{x}$ with $\prgoodrat{x+1}$ \footnote{\label{foot:xm} In the boundary case \(x=\prit\), we define \(\prgoodrat{x+1}\) to be an auxiliary good that is assigned value \(0\) by every agent.} gives a total value below $\alpha$. Plugging the expression for $\valu_\agenti(\{\prgoodrat{1}\})$ into these inequalities, we get 
$	\valu_\agenti(\{\prgoodrat{x}\}) > \tfrac{\alpha}{3} - s_i,$ and 
	$\valu_\agenti(\{\prgoodrat{x+1}\}) < \tfrac{4\alpha}{3} - 1 - s_i.
	$
This implies that no good in $\prgoods$ has a value (according to agent~$\agent_i$) that falls within the interval $$\left[\frac{4\alpha}{3} - 1 - s_i, \frac{\alpha}{3} - s_i\right].$$ Since $\goods \subseteq \prgoods$, the same holds for all goods in $\goods$, implying \Cref{conds2}.
	
	To prove \Cref{conds3}, recall that $\hatmms{\valu_\agenti} \ge 1$. Then for any $s_i \le \lambda \le \tfrac{4\alpha}{3}-1$, we have $\valu_\agenti(\{\prgoodrat{1}\}) \le (1 - \tfrac{\alpha}{3}) + \lambda$, therefore applying \cref{lem:calib-reduction} implies:
	\[
	\dotmms{\auxfun{\LM}{\agenti}} \ge 1 - 3\lambda.
	\]
\end{proof}

	\clearpage
 \section{\boldmath \Cref{alg:allocation1}: Frequent Green Agents} \label{sec:alg-N1}
In this section, we present our algorithm for the case that \( |\Nyek| \geq \lbnone \). The pseudocode for this case is given in \Cref{alg:allocation1}.
This algorithm takes as input an instance that is irreducible with respect to  $\reductiontype^0,\reductiontype^1,\reductiontype^2,\rstar^1$. 
First we run a set of further reductions on the input instance. Next, we run a Bag-filling on the set of remaining agents. Throughout this section, whenever we need to choose between multiple agents, we prioritize agents in \( \Ndo \).

\begin{algorithm}[h]
	\caption{\textsc{Algorithm-Case1}}
	\label{alg:allocation1}
	\textbf{Input:} \(\instance, \Nyek,\Ndo\) \CComment{Assumption: $|\Nyek| \geq \lbnone$}\\
	\textbf{Output:} Allocation satisfying  $(\tfrac{10}{13})$-$\MMS$
	\begin{algorithmic}[1]
		\State $\Sinstance =$  \textsc{Secondary-Reductions}$( \instance)$  \CComment{See \cref{algo:N11}}
		\State  Run \textsc{Bag-filling1}$ ( \Sinstance, \Nyek,\Ndo)$ \CComment{See \cref{algo:N1}}
	\end{algorithmic}
\end{algorithm}

We organize this section in two parts. In \Cref{subsec:alg-N1-reductions}, we present the additional reductions used by our algorithm. Then, in \Cref{subsec:alg-N1-bag-filling}, we describe the Bag-filling process and prove the approximation guarantee of the resulting allocation.

\subsection{Secondary Reductions} \label{subsec:alg-N1-reductions}
In this case, we further apply a sequence of the secondary reductions, following the priority order
$$\reductiontype^1 \succ \reductiontype^2 \succ \reductiontype^3 \succ \reductiontype^4 \succ \widetilde\reductiontype^{2}.$$
The pseudocode for this step is provided in \cref{algo:N11}. We denote by $\Sinstance = (\Sagents, \Sgoods)$ the instance obtained after applying these reductions.

\begin{algorithm}[h]
	\caption{$\textsc{Secondary-Reductions}$}
	\label{algo:N11}
	\textbf{Input:} $\instance = (\agents, \goods)$ \\
	\textbf{Output:} $\Sinstance = (\Sagents, \Sgoods)$ 
	\begin{algorithmic}[1]
		\State 
		$
		\prreduc \leftarrow [\reductiontype^1 \succ \reductiontype^2 \succ \reductiontype^3 \succ \reductiontype^4 \succ \rstar^2] 
		$ \CComment{Reduction order}
		\State $\Sinstance \leftarrow \instance$ 
		\While{there exists a valid reduction from $\prreduc$ on $\Sinstance$}
		\State $\reductiontype \leftarrow$ the highest-priority valid reduction from $\prreduc$ on $\Sinstance$
		
		\State Apply a valid reduction $\rho = \riki{\Sinstance}{\reductiontype}{x}{\agent_i}{\Sinstance'}$ on $\Sinstance$
		\CComment{Priority is given to agents in $\Ndo$}
		\EndWhile
		\State \Return $\Sinstance = (\Sagents, \Sgoods)$ 
	\end{algorithmic}
\end{algorithm}


We now establish some useful bounds. Specifically, in \Cref{lem:N1-upper-bounds}, we show that the values of goods \( \good_1 \), \( \Sgood_2 \), and \( \Sgood_{2\Snumb+1} \) are bounded above under the functions \( \valu_\agenti \), \( (\FJ \star \valu_\agenti) \), and \( (\vG \star \FJ \star \valu_\agenti) \) for green agents. These bounds are later used to prove our  claims.

\newpage
\begin{observation}\label{lem:N1-upper-bounds} 
	Let $\agent_i$ be a green agent in $\Sagents$. Then, for goods \( \good_1 \), \( \Sgood_2 \), and \( \Sgood_{2\Snumb+1} \), the following upper bounds hold:
	\begin{table}[h]
		\centering
		\renewcommand{\arraystretch}{1.5}
		\begin{tabular}{>{\raggedright\arraybackslash}p{2cm}>{\raggedright\arraybackslash}p{4cm}>{\raggedright\arraybackslash}p{4cm}>{\raggedright\arraybackslash}p{5cm}}
			\toprule
			\rowcolor{headerblue!60} \textbf{Goods} & \textbf{$\valu_\agenti(\cdot)$} & \textbf{$\ff{\cdot}$} & \textbf{$(\vG \star \FJ \star \valu_\agenti)(\cdot)$} \\
			\midrule
			\rowcolor{rowblue1} \( \{\good_1\} \) & \( < 2\alpha - 1 \) & \( \le \tfrac{1}{2} \) & \( \le \tfrac{5}{2} - \tfrac{8\alpha}{3} \) \\
			\rowcolor{rowblue2} \( \{\Sgood_2\} \) & \( < \tfrac{\alpha}{2} \) & \( \le 1 - \tfrac{5\alpha}{6} \) & \( \le 3 - \frac{7\alpha}{2} \) \\
			\rowcolor{rowblue1} \( \{\Sgood_{2\Snumb+1}\} \) & \( < \tfrac{\alpha}{3} \) & \( \le 1 - \alpha \) & \( \le 2 - \tfrac{7\alpha}{3} \) \\
			\bottomrule
		\end{tabular}
		\label{tab:n1-upper-bounds}
	\end{table}
\end{observation}
\begin{proof}
	For $\good_1$, since $\agent_i$ is a green agent, she values $\good_{2\numb+1}$ at least $1 - \alpha$. Moreover, since $\rstar^1$ is not applicable, $\valu_i(\{\good_1,\,\good_{2\numb+1}\})\le \alpha$. Together, these imply the first inequality. Noting that $\tfrac{3}{4} < \alpha < \tfrac{5}{6}$, we have
	\begin{align*}
		\ff{\{\good_1\}} &\leq \FJ(2\alpha-1)& \FJ \mbox{ is non-decreasing},
		\\&=\max\Bigl(2(1-\alpha),\, \frac{1}{2}\Bigr) &2\alpha - 1 \in [1 - \tfrac{2\alpha}{3}, 1 -\tfrac{\alpha}{3}-(\tfrac{2\alpha}{3}-\tfrac{1}{2})),
		\\
		&=\frac{1}{2}.
	\end{align*}
	By the definition of \( \vG \) and noting that \( \alpha\ge \tfrac{9}{13}\) we have   
	\begin{align*}
		\gf{\{\good_1\}} &\le
		\vG\left(\frac{1}{2}\right)&\vG\mbox{ is non-decreasing}, \\&= \max\Bigl(3-\frac{7\alpha}{2},\, \frac{5}{2} - \frac{8\alpha}{3}\Bigr) &  \tfrac{1}{2} \in [2 - \tfrac{13\alpha}{6}, 1], \\&=\frac{5}{2} - \frac{8\alpha}{3}.
	\end{align*}
	
	For \(\Sgood_2\), since reduction \(\rstar^2\) is not applicable we have $\valu_\agenti(\{\Sgood_2\}) < \tfrac{\alpha}{2}.$
	Since \(\tfrac{2}{3}<\alpha<\tfrac{6}{7}\)
	, we have
	\begin{align*}
		\ff{\{\Sgood_2\}} &\leq \FJ\Bigl(\frac{\alpha}{2}\Bigr)& \FJ \mbox{ is non-decreasing},
		\\&=\max\Bigl(1-\alpha,\;1-\frac{5\alpha}{6}\Bigr) &\tfrac{\alpha}{2}
		\;\in\;[\tfrac{\alpha}{3}-(\tfrac{4\alpha}{3}-1),\;1-\tfrac{2\alpha}{3}),
		\\
		&=1-\frac{5\alpha}{6}.
	\end{align*}
	By the definition of \( \vG \) and noting that \( \alpha \geq \tfrac{3}{4} \) we have   
	\begin{align*}
		\gf{\{\Sgood_2\}} &\le
		\vG\left(1 - \frac{5\alpha}{6}\right)&\vG\mbox{ is non-decreasing}, \\&= \max\Bigl(3 - \frac{7\alpha}{2},\, 1 - \frac{5\alpha}{6} - \frac{8\alpha}{3} + 2\Bigr) &  1 - \tfrac{5\alpha}{6} \in [2 - \tfrac{13\alpha}{6},\, 1], \\&=3 - \frac{7\alpha}{2}.
	\end{align*}
	For \( \Sgood_{2\Snumb+1} \), since reduction \( \reductiontype^2 \) is not applicable, we have $\valu_\agenti(\{\Sgood_{2\Snumb+1}\}) < \tfrac{\alpha}{3}.$  
	Since \(\alpha>\tfrac{3}{4}\), we have
	\begin{align*}
		\ff{\{\Sgood_{2\Snumb+1}\}} &\leq  \FJ\Bigl(\frac{\alpha}{3}\Bigr) & \FJ \mbox{ is non-decreasing},\\
		&=\max\left(\frac{\alpha}{3} - \Bigl(\frac{4\alpha}{3}-1\Bigr),\, \frac{\alpha}{3} - \Bigl(\frac{4\alpha}{3}-1\Bigr)\right)  &\tfrac{\alpha}{3}
		\in[\tfrac{\alpha}{3} - (\tfrac{4\alpha}{3}-1),\;1 - \tfrac{2\alpha}{3}),
		\\
		&=1-\alpha.
	\end{align*}
	Finally, by the definition of \( \vG \), and noting that \(\tfrac{3}{4} \le \alpha \le \tfrac{6}{7} \), we have  
	\begin{align*}
		\gf{\{\Sgood_{2\Snumb+1}\}} &\le \vG\left(1 - \alpha\right)&\vG\mbox{ is non-decreasing},\\ &= \max\Bigl(2 - \frac{7\alpha}{3},\, (1-\alpha) - \frac{4\alpha}{3} + 1\Bigr) &1 - \alpha \in [2 - \tfrac{7\alpha}{3}, 2 - \tfrac{13\alpha}{6}),\\&=2 - \frac{7\alpha}{3}.
	\end{align*} 
\end{proof}

Also, for \( \Sgood_{3\Snumb+1} \), we establish a strong upper bound in \cref{obs:3n-kam}.

\begin{observation}\label{obs:3n-kam}
	Let $\agent_i \in \Sagents$ be an agent with $\dotmms{\valu_\agenti} < 1$ after the primary reductions. Then $\valu_{\agenti}(\{\Sgood_{3\Snumb+1}\}) < \tfrac{4\alpha}{3}-1$.
\end{observation}

\begin{proof}
	By \cref{lem:N-minus}, for every $\Sgood \in \Sgoods$ we have $\valu_{\agenti}(\{\Sgood\}) \notin [\tfrac{4\alpha}{3}-1-s_i,\, \tfrac{\alpha}{3}-s_i]$. From \cref{conds1}, $0 < s_i < \tfrac{4\alpha}{3}-1$, hence $[\tfrac{4\alpha}{3}-1,\, 1-\alpha] \subseteq [\tfrac{4\alpha}{3}-1-s_i,\, \tfrac{\alpha}{3}-s_i]$. Since $\reductiontype^3$ is not applicable, $\valu_{\agenti}(\{\Sgood_{3\Snumb+1}\}) < \tfrac{\alpha}{4}$, and as $\tfrac{\alpha}{4} \in [\tfrac{4\alpha}{3}-1,\, 1-\alpha]$, the claim follows. 
\end{proof}

As shown in \Cref{reductionlem}, applying \( \reductiontype^1, \reductiontype^2, \reductiontype^3, \reductiontype^4 \) does not reduce the \(\MMS\) value of any agent. However, this is not necessarily true for \( \rstar^2 \). In \cref{lem:plus-two,lem:minus-two}, we provide a set of bounds on the valuation of agents after the secondary reductions.

\begin{lemma}\label{lem:plus-two}
	Let $\agent_i$ be a green agent in $\Sagents$ such that  $\dotmms{\valu_\agenti} \ge 1$ after the primary reductions, and $\ddotmms{\valu_\agenti} < 1$ after the secondary reductions. Then, there exists a positive real number 
	$t_\agenti$ satisfying the following conditions:
	\begin{align}
		&\valu_\agenti(\{\Sgood_1\}) - \frac{1}{2} \le t_\agenti < 2\alpha - \frac{3}{2},\label[ineq]{cond1}\\
		\forall \Sgood \in \Sgoods,\qquad\qquad &\valu_\agenti(\{\Sgood\}) \notin \Bigl[\alpha - \frac{1}{2} - t_\agenti,\; \frac{1}{2} - t_\agenti\Bigr],\label{cond2}\\
		& \ddotmms{\auxfun{\AN{t_\agenti}}{\agenti}} \ge 1 - 2t_\agenti. \label[ineq]{cond3}
	\end{align}
\end{lemma}

\begin{proof}
Consider the first reduction $\rho = \riki{\prinstance}{\reductiontype}{x}{\agent_j}{\seinstance}$ such that $\agent_i$'s maximin share drops below $1$, that is, before the reduction, $\hatmms{v_\agenti} \geq 1$, but after the reduction, $\checkmms{v_\agenti} < 1$. By \cref{reductionlem}, among all reduction patterns, only $\rstar^2$ can decrease the maximin share of an agent. So we assume this reduction allocated goods $\{\prgoodrat{1},\prgoodrat{x}\}$, and define 
$
t_\agenti = \valu_\agenti(\{\prgood_{1}\}) - \tfrac{1}{2}.
$
The drop in the maximin share after this reduction implies  $\valu_\agenti(\{\prgood_{1},\prgood_{x}\}) > 1$, by \cref{reductionlem2}. Meanwhile,  since \( \valu_\agenti(\{\prgood_{1}\}) \geq \valu_\agenti(\{\prgood_{x}\}) \), we have \( \valu_\agenti(\{\prgood_{1}\}) > \mmsfrac{1}{2} \), which means \( t_i > 0 \). By \Cref{lem:N1-upper-bounds}, we have \( \valu_\agenti(\{\prgood_{1}\}) < 2\alpha - 1 \), which gives \( t_i < 2\alpha - \tfrac{3}{2} \).  Therefore,  $t_i$ is positive and lies in the desired range.
Since all goods in $\Sgoods$  are drawn from $\prgoods$, we have
$
\valu_\agenti(\{\Sgood_1\}) \leq \valu_\agenti(\{\prgood_{1}\})$ and hence 
$\valu_\agenti(\{\Sgood_1\})\leq \frac{1}{2} + t_\agenti.$
This proves \Cref{cond1}.

	To show \Cref{cond2}, note that by the construction of $\rstar^2$, we have $$\valu_\agenti(\{\prgoodrat{1},\prgoodrat{x}\})>1 \qquad \text{and} \qquad \valu_{\agenti}(\{\prgoodrat{1},\prgoodrat{x+1}\})<\alpha. \footref{foot:xm}$$  Plugging the bound for $\valu_\agenti(\{\prgoodrat{1}\})$ into these inequalities, we get 
$$
\valu_\agenti(\{\prgood_{x}\}) > \tfrac{1}{2} - t_\agenti 
\qquad \text{and} \qquad
\valu_\agenti(\{\prgood_{x+1}\}) < \alpha - \mmsfrac{1}{2} - t_\agenti.
$$
This implies that no good in $\prgoods$ has a value to agent~$\agent_i$ that falls within the interval $[\alpha - \tfrac{1}{2} - t_\agenti,\; \tfrac{1}{2} - t_\agenti]$. Since $\Sgoods \subseteq \prgoods$, the same holds for all goods in $\Sgoods$, implying \Cref{cond2}.

	To prove \Cref{cond3}, recall that $\hatmms{\valu_\agenti} \ge 1$. Since  $\valu_\agenti(\{\prgood_{1}\}) = \tfrac{1}{2} + t_\agenti$, applying \cref{lem:calib-reduction} implies $\ddotmms{\auxfun{\AN{t_\agenti}}{\agenti}} \ge 1 - 2t_\agenti.$
\end{proof}

\begin{lemma}\label{lem:minus-two}
	Let $\agent_i$ be a green agent in $\Sagents$ such that  $\dotmms{\valu_\agenti} < 1$ after the primary reductions, and $\ddotmms{\auxfun{\FJ}{\agenti}} < 4(1-\alpha)$ after the secondary reductions. Then, there exists a positive real number 
	$t_\agenti$ satisfying the following conditions:
	\begin{align}
		&\ff{\{\Sgood_1\}} - 2(1-\alpha) \le t_\agenti \le 2\alpha - \frac{3}{2},\label[ineq]{condd1}\\
		\forall \Sgood \in \Sgoods,\qquad\qquad &\ff{\{\Sgood\}} \notin \Bigl[\frac{5\alpha}{3} - 1 - t_\agenti,\; 2(1-\alpha) - t_\agenti\Bigr],\label{condd2}\\
		& \ddotmms{\auxfun{\vHB_{t_\agenti}\star\FJ}{\agenti}} \ge 4(1-\alpha) - 2t_\agenti. \label[ineq]{condd3}
	\end{align}
\end{lemma}

\begin{proof}
	By setting \(\lambda = \tfrac{4\alpha}{3} - 1\) in \Cref{conds3}, we obtain $$\dotmms{\auxfun{\FJ}{\agenti}} \;\ge\; 4(1-\alpha).$$ Hence,
	consider the first reduction $\rho =\riki{\prinstance}{\reductiontype}{x}{\agent_j}{\seinstance}$ such that $\agent_i$'s maximin share under $\auxfun{\FJ}{\agenti}$ drops below $4(1-\alpha)$. By \cref{reductionlem}, among all reduction patterns, only $\rstar^2$ can decrease the maximin share value of an agent. Thus, we assume this reduction allocated two goods, $\prgoodrat{1}$ and $\prgoodrat{x}$, and define $t_\agenti = \ff{\{\prgood_{1}\}} - 2(1-\alpha).$
	
	We first show that  $t_\agenti$ is positive and less than $2\alpha - \tfrac{3}{2}$. The drop in the maximin share by $\rho$ implies that $\ff{\{\prgood_{1},\,\prgood_{x}\}} > 4(1 - \alpha),$ by \cref{reductionlem2}. Since $\prgood_{1}$ is more valuable than $\prgood_{x}$, we conclude that \( \ff{\{\prgood_{1}\}} > 2(1-\alpha) \), which means \( t_\agenti > 0 \).	
	On the other hand, by \Cref{lem:N1-upper-bounds}, we have \( \ff{\{\prgood_{1}\}} \le \tfrac{1}{2}\), which gives \( t_\agenti \le 2\alpha - \tfrac{3}{2} \). Therefore,  $t_i$ is positive and lies in the desired range.
	Since all goods in $\Sgoods$ are drawn from $\prgoods$, we have
	\begin{align*}
	\ff{\{\Sgood_1\}} &\leq \ff{\{\prgood_{1}\}} \\ &= 2(1-\alpha) + t_\agenti.
	\end{align*}
	This proves \Cref{condd1}. 
	
	To show \Cref{condd2}, we use two properties of $\rstar^2$:
	$$
	\ff{\{\prgood_{1},\,\prgood_{x}\}} > 4(1-\alpha) 
	\quad \text{and} \quad  
	\valu_\agenti(\{\prgood_{1},\,\prgood_{x+1}\}) < \alpha. \footref{foot:xm}
	$$
	Since $0 < t_\agenti \le 2\alpha - \tfrac{3}{2}$ and $\tfrac{3}{4}<\alpha<\tfrac{6}{7}$ we conclude $1 - \tfrac{2\alpha}{3} + t_\agenti \in [1 - \tfrac{2\alpha}{3}, 1 -\tfrac{\alpha}{3}-(\tfrac{2\alpha}{3}-\tfrac{1}{2}))$. Hence, 
	\begin{align*}
		\FJ(1-\frac{2\alpha}{3} + t_\agenti)&=\max\Bigl(2(1-\alpha),\,  \frac{5}{2} - \frac{8\alpha}{3} + t_{\agenti}\Bigr) &\mbox{\cref{def:F}},\\
		&< 2(1-\alpha) + t_\agenti& \alpha> \frac{3}{4},\\&=\ff{\{\prgood_{1}\}}.
	\end{align*}
	 Using monotonicity of $\FJ$, we obtain $\valu_\agenti(\{\prgood_{1}\}) \ge 1 - \tfrac{2\alpha}{3} + t_\agenti$.  Plugging $\ff{\{\prgood_{1}\}} = 2(1-\alpha) + t_\agenti$ and $\valu_\agenti(\{\prgood_{1}\}) \ge 1 - \tfrac{2\alpha}{3} + t_\agenti$ into these inequalities, we get
	$$\ff{\{\prgood_{x}\}} > 2(1-\alpha) - t_\agenti \qquad \text{and} \qquad 
	\ff{\{\prgood_{x+1}\}}\le \valu_\agenti(\{\prgood_{x+1}\}) < \tfrac{5\alpha}{3} - 1 - t_\agenti.$$
This implies that no good in $\prgoods$ has a value under $\auxfun{\FJ}{\agenti}$ that falls within the interval $[\tfrac{5\alpha}{3} - 1 - t_\agenti,\; 2(1-\alpha) - t_\agenti]$. Since $\Sgoods \subseteq \prgoods$, the same holds for all goods in $\Sgoods$, implying \Cref{condd2}.

To prove \Cref{condd3}, recall that $\hatmms{\auxfun{\FJ}{\agenti}} \ge 4(1-\alpha)$, and since  $\ff{\{\prgood_{1}\}} = 2(1-\alpha) + t_\agenti$, applying \cref{lem:calib-reduction} implies $\ddotmms{\auxfun{\vHB_{t_\agenti}\star\FJ}{\agenti}} \ge 4(1-\alpha) - 2t_\agenti.$
\end{proof}

\subsection{\boldmath Bag-filling} \label{subsec:alg-N1-bag-filling}

After the secondary reductions, we apply the Bag-filling method shown in \Cref{algo:1}. The algorithm begins by constructing \( \Snumb \) bundles \( \bundle_1, \bundle_2, \dots, \bundle_{\Snumb} \), where  
\[
\bundle_k = \{ \Sgood_k,\, \Sgood_{\Snumb+k},\, \Sgood_{3\Snumb-k+1} \}.
\]

\makeatletter

\tikzset{
	curve scale/.initial=1,
	curve scale/.get=\tikz@curvescale,
	curve scale/.store in=\tikz@curvescale
}

\newcommand\curvedraw[6][]{%
	\draw[#1] 
	(#2) 
	.. controls 
	+(#3,{\tikz@curvescale*#4}) 
	and +(-#3,{\tikz@curvescale*#4}) 
	.. (#5) #6;
}
\makeatother

\begin{figure}
	\centering
	\resizebox{\textwidth}{!}{
		\begin{tikzpicture}[x=1.25cm, y=1.3cm, curve scale=1.25, good/.style={draw, rectangle, rounded corners=3pt,minimum width=1cm,minimum height=1cm,inner sep=2pt, font=\small}]
	\foreach \i/\label/\x/\gcol in {
		1/$\Sgood_1$/1/blue,
		2/$\Sgood_2$/2/red,
		3/$\Sgood_k$/3.5/green,
		4/$\Sgood_{\Snumb-1}$/5/orange,
		5/$\Sgood_{\Snumb}$/6/yellow,
		6/$\Sgood_{\Snumb+1}$/7.2/blue,
		7/$\Sgood_{\Snumb+2}$/8.2/red,
		8/$\Sgood_{\Snumb+k}$/9.7/green,
		9/$\Sgood_{2\Snumb-1}$/11.2/orange,
		10/$\Sgood_{2\Snumb}$/12.2/yellow,
		11/$\Sgood_{2\Snumb+1}$/13.4/yellow,
		12/$\Sgood_{2\Snumb+2}$/14.4/orange,
		13/$\Sgood_{3\Snumb-k+1}$/15.9/green,
		14/$\Sgood_{3\Snumb-1}$/17.4/red,
		15/$\Sgood_{3\Snumb}$/18.4/blue} {
		\node[good, fill=\gcol!30] (good\i) at (\x,0) {\label};
	}
	
	\curvedraw[blue, thick]{good1}{1.2}{2.0}{good6}{;}
	\curvedraw[red,  thick]{good2}{1.2}{-2.0}{good7}{;}
	\curvedraw[green,thick]{good3}{1.2}{2.0}{good8}{;}
	\curvedraw[orange,thick]{good4}{1.2}{-2.0}{good9}{;}
	\curvedraw[yellow,thick]{good5}{1.2}{2.0}{good10}{;}
	
	\curvedraw[blue, thick]{good6}{1.2}{2.4}{good15}{;}
	\curvedraw[red,  thick]{good7}{1.2}{-2.2}{good14}{;}
	\curvedraw[green,thick]{good8}{1.2}{2.0}{good13}{;}
	
	\curvedraw[orange,thick]{good9}{0.7}{-1.5}{good12}{;}
	\curvedraw[yellow,thick]{good10}{0.2}{1.0}{good11}{;}

	\foreach \i/\label/\x in {
		16/$\Sgood_{3\Snumb+1}$/19.5,
		17/$\Sgood_{3\Snumb+2}$/20.4,
		18/$\Sgood_{4\Snumb}$/21.6
	} {
		\node[good, scale=0.8, fill=gray!30] (good\i) at (\x,0) {\label};
	}
	
	\foreach \i/\label/\x in {
		19/$\Sgood_{4\Snumb+1}$/22.6,
		20/$\Sgood_{4\Snumb+2}$/23.4,
		21/$\Sgood_{m}$/24.6
	} {
		\node[good, scale=0.6, fill=gray!30] (good\i) at (\x,0) {\label};
	}
	
	\node at (2.75,0) {\(\dots\)};  
	\node at (4.25,0) {\(\dots\)};  
	
	\node at (8.95,0) {\(\dots\)};  
	\node at (10.45,0) {\(\dots\)};  
	
	\node at (15.1,0) {\(\dots\)};  
	\node at (16.7,0) {\(\dots\)};  
	
	\node at (21.0,0) {\(\dots\)};  
	
	\node at (24.0,0) {\(\dots\)};  
	
	\draw[decorate, thick, decoration={brace,mirror}] 
	(19.1, -0.5) -- (22.05, -0.5)
	node[midway, below=6pt, font=\large]{Special Bag-filling goods};
	
	\draw[decorate, thick, decoration={brace}] 
	(22.2, 0.5) -- (24.85, 0.5)
	node[midway, above=6pt, font=\large]{Bag-filling goods};

	\draw[dashed] (6.6,-0.6) -- (6.6,0.6);
	
	\draw[dashed] (12.8,-0.6) -- (12.8,0.6);
	
	\draw[dashed] (19.0,-0.6) -- (19.0,0.6);
	
	\draw[dashed] (22.15,-0.6) -- (22.15,0.6);
\end{tikzpicture}
	}
	\caption{Structure of the bags in \Cref{algo:1}.}
\end{figure}

Then, starting from \( k = \Snumb \) and proceeding downwards to \( k = 1 \), the algorithm attempts to allocate bag \( \bundle_k \) to an agent. At each step, if no remaining agent values \( \bundle_k \) at least \( \alpha \), one additional remaining good is added to the bag until some agent finds its value at least \( \alpha \).  
\textbf{
	We select the next good according to the following priority:}
\begin{enumerate}[label=(\roman*)]
	\item \( \Sgood_{3\Snumb + k} \),
	\item any good from the set \( \{ \Sgood_{4\Snumb+1}, \Sgood_{4\Snumb+2}, \dots, \Sgood_{\Sit} \} \),
	\item the remaining good with the smallest index.
\end{enumerate}
 
\begin{algorithm}[h]
	\caption{$\textsc{Bag-Filling1}$}
	\label{algo:1}
	\textbf{Input:} $\Sinstance = (\Sagents, \Sgoods)$ \\
	\textbf{Output:} Allocation satisfying  $(\tfrac{10}{13})$-$\MMS$ 
	
	\begin{algorithmic}[1]
		\For {$k$ : $1 \rightarrow \Snumb$}
		\State $\finib_k \leftarrow \{\Sgood_k,\, \Sgood_{\Snumb+k},\, \Sgood_{3\Snumb-k+1}\}$
		\EndFor
		\For {$k$ : $\Snumb \rightarrow 1$}
		\While {There does not exist a remaining agent $\agent_i$ s.t. $\valu_\agenti(\bundle_k) \geq \alpha$}
		\If {$\Sgood_{3\Snumb + k}$ is remaining}
		\State Add $\Sgood_{3\Snumb + k}$ to $\bundle_k$
		\ElsIf {$\exists x \ge 4\Snumb+1$ s.t. $\Sgood_x$ is remaining}
		\State Add $\Sgood_x$ to $\bundle_k$ 
		\Else
		\State Add the remaining good with the smallest index to $\finib_k$ 
		\EndIf
		\EndWhile
		\State Allocate $\bundle_k$ to $\agent_i$ with $\valu_\agenti(\finib_k) \ge \alpha$ \CComment{Priority is given to agents in $\Ndo$}
		\EndFor
	\end{algorithmic}
	\label{algo:N1}
\end{algorithm}

To analyze our algorithm, we categorize green agents into four distinct groups based on their maximin share properties after the reduction phases: 
\begin{enumerate}
	\item Agents whose maximin share satisfies \( \dotmms{\valu_i} \geq 1 \) after the primary reductions and \( \ddotmms{\valu_i} \geq 1 \) after the secondary reductions.
	
	\item Agents for whom \( \dotmms{\valu_i} < 1 \) after the primary reductions, but the calibrated share satisfies \( \ddotmms{\auxfun{\FJ}{\agenti}} \geq 4(1-\alpha) \) after the secondary reductions.
	
	\item Agents with \( \dotmms{\valu_i} \geq 1 \) after the primary reductions, but \( \ddotmms{\valu_i} < 1 \) after the secondary reductions.
	
	\item Agents for whom \( \dotmms{\valu_i} < 1 \) after the primary reductions and \( \ddotmms{\auxfun{\FJ}{\agenti}} < 4(1-\alpha) \) after the secondary reductions.
\end{enumerate}

\cref{lem:way} provides general tools for analyzing the first two groups of agents, while \cref{lem:way2} helps with the analysis of the last two groups. In \cref{plus-one}, we show that the agents in the first group receive a bundle. \cref{minus-one} establishes the same for the second group, \cref{plus-two} for the third group, and \cref{minus-two} for the fourth group.

\begin{lemma}\label{lem:way}
	Let \( \agent_i \in \Sagents \) be an agent, and let \( \prvalu \) be a valuation function that ranks the goods in the same order as \( \valu_i \). Assume the following conditions hold:
	\begin{align}
		\forall_{\Sgood \in \Sgoods} \qquad\qquad&\label[ineq]{cond:sefr}\prvalu(\{\Sgood\}) \le \valu_i(\{\Sgood\}) \\ 
		& \alpha + \prvalu(\{\Sgood_{3\Snumb+1}\}) \le \ddotmms{\prvalu}, &\label[ineq]{cond:yek}\\
		\forall_{2 \le k \le \Snumb},\qquad\qquad & \prvalu(\{\Sgood_k,\,\Sgood_{\Snumb+k},\,\Sgood_{3\Snumb+1-k}\}) \le \ddotmms{\prvalu}, &\label[ineq]{cond:do}\\
		& \prvalu(\{\Sgood_1,\,\Sgood_{\Snumb+1},\,\Sgood_{3\Snumb}\}) + \alpha \le 2\,\ddotmms{\prvalu}. &\label[ineq]{cond:se}
	\end{align}
	Then, $\agent_i$ receives a bundle of value at least $\alpha$ in \cref{algo:1}.
\end{lemma}

\begin{proof}
	Suppose, for the sake of contradiction, that agent \(\agent_i\) receives no bundle, and let \(\halt\) be the index at which the algorithm halts while filling bag \(\finib_\halt\).
	
	We claim that for every $2 \le k \le \Snumb$, the value of bag $B_k$ satisfies $\prvalu(B_k) \le \ddotmms{\prvalu}$. 
	If no good is added to bag $k$ during the Bag-filling process—either because it already holds value at least \(\alpha\) for some remaining agent or its turn has not yet come—then this inequality follows directly from \Cref{cond:do}.
	Otherwise, if goods are added to bag $k$, we know that just before the last good was added, the bundle had value less than \(\alpha\) to agent \(\agent_i\). Furthermore, by the construction of the Bag-filling process, the last added good has an index at least \(3\Snumb + 1\). Therefore, its value under $\prvalu$ is at most \(\prvalu(\{\Sgood_{3\Snumb+1}\})\).
	By the additivity of \(\prvalu\) and using \Cref{cond:sefr,cond:yek}, it follows that the total value of \(\finib_k\) does not exceed \(\ddotmms{\prvalu}\), as claimed. Now, assume \(\halt \neq 1\). Since the bags \(\finib_1, \finib_2, \ldots, \finib_{\Snumb}\) form a partition of \(\Sgoods\), we have:
	\begin{align*}
		\prvalu(\Sgoods) &= \sum_{k=1}^{\Snumb} \prvalu(\finib_k) \\
		&= \prvalu(\finib_1) + \prvalu(\finib_{\halt}) + \sum_{k=2}^{\halt-1} \prvalu(\finib_k) + \sum_{k=\halt+1}^{\Snumb} \prvalu(\finib_k) \\
		&= \prvalu(\finib_1) + \prvalu(\finib_{\halt}) + (\Snumb - 2)\cdot \ddotmms{\prvalu} &\\
		&< \prvalu(\finib_1) + \alpha + (\Snumb - 2)\cdot \ddotmms{\prvalu} \\
		&= \prvalu(\{\Sgood_1,\,\Sgood_{\Snumb+1},\,\Sgood_{3\Snumb}\}) + \alpha + (\Snumb - 2)\cdot \ddotmms{\prvalu} &\mbox{$\halt \neq 1$},\\
		&\le 2\cdot \ddotmms{\prvalu} + (\Snumb - 2)\cdot \ddotmms{\prvalu} &\mbox{\cref{cond:se}},\\
		&= \Snumb \cdot \ddotmms{\prvalu}.
	\end{align*}
	which is a contradiction.
	If \(\halt = 1\), since for every $2 \le k \le \Snumb$ we have $\prvalu(B_k) \le \ddotmms{\prvalu}$,
	and $\prvalu(B_1) < \alpha \le \ddotmms{\prvalu}$, it follows that $\prvalu(\Sgoods) < \Snumb\;\ddotmms{\prvalu}$, which is a contradiction. 
\end{proof}

\begin{lemma}\label{lem:way2}
	Let \( \agent_i \in \Sagents \) be an agent, and let \( \prvalu \) be a valuation function that ranks the goods in the same order as \( \valu_i \). Assume the following conditions hold:
	\begin{align}
		\forall_{\Sgood \in \Sgoods}\quad &\prvalu(\{\Sgood\}) \le \valu_i(\{\Sgood\}),&\label[ineq]{cond:zero} \\ 
		& \alpha + \prvalu(\{\Sgood_{4\Snumb+1}\}) \le \ddotmms{\prvalu}, 
		&\label[ineq]{cond:one}\\
		& \alpha + \prvalu(\{\Sgood_1,\,\Sgood_{\Snumb+1},\,\Sgood_{3\Snumb}\}) \le 2\,\ddotmms{\prvalu}, 
		&\label[ineq]{cond:two}\\
		\forall_{2\le k\le\Snumb}\quad &\prvalu(\{\Sgood_k,\,\Sgood_{\Snumb+k},\,\Sgood_{3\Snumb+1-k},\,\Sgood_{3\Snumb+k}\}) \le \ddotmms{\prvalu},
		&\label[ineq]{cond:three}\\
		& 2\alpha + \prvalu(\{\Sgood_{3\Snumb+1}\}) \le 2\,\ddotmms{\prvalu},
		&\label[ineq]{cond:four}\\
		& 2\alpha + \prvalu(\{\Sgood_1,\,\Sgood_{\Snumb+1},\,\Sgood_{3\Snumb},\,\Sgood_{3\Snumb+1}\}) 
		\le 3\,\ddotmms{\prvalu}.
		&\label[ineq]{cond:five}
	\end{align}
	Then, $\agent_i$ receives a bundle of value at least $\alpha$ in \cref{algo:1}.
\end{lemma}

\begin{proof}
	Assume, to reach a contradiction, that agent \(\agent_i \in \Sagents\) receives no bundle. Then the algorithm must enter its third priority at least once. Let the first such moment occur while filling bag~\(\ell\). For each \(1 \le k \le \Snumb\), denote by \(C_k\) the contents of bag~\(k\) at that time. Observe that during the first priority, the algorithm added \(\Sgood_{3\Snumb+k}\) into bag~\(k\), and during the second priority it added \(\Sgood_{x}\) to bag~\(k\) for some \(x \ge 4\Snumb+1\). Moreover, since the algorithm only enters the third priority after exhausting all higher‐priority goods, none of 
	\[
	\Sgood_{4\Snumb+1},\dots,\Sgood_{\Sit}
	\]
	is remaining. And since this is the first time we reach the third priority, for each \(1\le k\le\Snumb\), the good \(\Sgood_{3\Snumb+k}\) has either been placed into bag \(k\) or is still remaining; in particular, for bag~\(\ell\) we have \(\Sgood_{3\Snumb+\ell}\in C_\ell\). Therefore, the sets
	\[
	C_k \;\cup\;\{\Sgood_{3\Snumb+k}\},\quad k=1,\dots,\Snumb,
	\]
	indeed form a partition of all goods in \(\Sgoods\). We will show that
	\[
	\sum_{k=1}^{\Snumb} \prvalu\bigl(C_k \cup \{\Sgood_{3\Snumb+k}\}\bigr)
	< \Snumb\;\ddotmms{\prvalu},
	\]
	contradicting the definition of the maximin share.
	
	First, we claim that for every $k\ge2$,
	$
	\prvalu\bigl(C_k \cup \{\Sgood_{3\Snumb+k}\}\bigr)
	\le \ddotmms{\prvalu}.
	$
	Indeed, if
	$C_k \cup \{\Sgood_{3\Snumb+k}\} = \{\Sgood_k,\Sgood_{\Snumb+k},\Sgood_{3\Snumb+1-k},\Sgood_{3\Snumb+k}\}$,
	then the desired bound follows immediately from \Cref{cond:three}. Otherwise, the algorithm must have reached the second priority: let $\Sgood_x$ be the last good added to~$C_k$. By \Cref{cond:zero} we have $\prvalu(C_k\setminus\{\Sgood_x\})<\alpha$, and since $x\ge4\Snumb+1$, it follows that
	\begin{align*}
		\prvalu\bigl(C_k \cup \{\Sgood_{3\Snumb+k}\}\bigr)
		&=\prvalu(C_k)\\
		&= \prvalu(C_k\setminus \{\Sgood_{x}\}) + \prvalu(\{\Sgood_x\}) \\
		&< \alpha + \prvalu(\{\Sgood_{4\Snumb+1}\}) &\\
		&\le \ddotmms{\prvalu} &\mbox{\cref{cond:one}}.
	\end{align*}
	
	Next, if \(\ell=1\), then since for every \(k\ge2\) we have $\prvalu\bigl(C_k\cup\{\Sgood_{3\Snumb+k}\}\bigr)\le\ddotmms{\prvalu}$, 
	we deduce $\prvalu\bigl(C_1\cup\{\Sgood_{3\Snumb+1}\}\bigr)\ge\ddotmms{\prvalu}$. 
	By \Cref{cond:one}, \(\ddotmms{\prvalu}\ge\alpha\), and since \(\Sgood_{3\Snumb+1}\) remains available, agent \(\agent_i\) would accept \(C_1\cup\{\Sgood_{3\Snumb+1}\}\).
 This contradicts our assumption, so $\ell\ge2$. We now show that for every $j\notin\{1,\ell\}$,
	\[
	\prvalu\bigl(C_j \cup \{\Sgood_{3\Snumb+j}\}\bigr)
	\ge \alpha.
	\]
	Suppose, towards a contradiction, that for some such~$j$ we have
	$\prvalu(C_j \cup \{\Sgood_{3\Snumb+j}\})<\alpha$. Noting that $\Sgood_{3\Snumb+\ell} \in C_\ell$, we have
	\begin{align*}
		&\prvalu\bigl(C_1 \cup \{\Sgood_{3\Snumb+1}\}\bigr)
		+ \prvalu\bigl(C_j \cup \{\Sgood_{3\Snumb+j}\}\bigr)
		+ \prvalu\bigl(C_\ell \cup \{\Sgood_{3\Snumb+\ell}\}\bigr) \\
		&< \prvalu(C_1 \cup \{\Sgood_{3\Snumb+1}\}) + 2\alpha \\
		&= \prvalu(\{\Sgood_1,\Sgood_{\Snumb+1},\Sgood_{3\Snumb},\Sgood_{3\Snumb+1}\}) + 2\alpha
		\le 3\,\ddotmms{\prvalu}&\mbox{\Cref{cond:five}}.
	\end{align*}
	Since each of the other bundles also has value at most $\ddotmms{\prvalu}$, summing yields
	$\prvalu(\Sgoods)<\Snumb\;\ddotmms{\prvalu}$,
	again a contradiction. Hence for all $j\notin\{1,\ell\}$ we have
	$\prvalu(C_j \cup \{\Sgood_{3\Snumb+j}\})\ge\alpha$.
	
	Note that this is the first time at which the algorithm enters the third priority. Hence goods $\Sgood_{3\Snumb+1},\,\Sgood_{3\Snumb+2},\,\dots,\,\Sgood_{3\Snumb+\ell-1}$ 
	are remaining. In the third priority the algorithm selects the remaining good with the smallest index, so it adds \(\Sgood_{3\Snumb+1}\) to bag $\ell$. We show that $\prvalu(C_\ell \cup \{\Sgood_{3\Snumb+1}\}) \ge \alpha$. Suppose, towards a contradiction, that $\prvalu(C_\ell \cup \{\Sgood_{3\Snumb+1}\}) < \alpha$. Then: 
	\begin{align*}
		\prvalu(C_1) +  \prvalu(C_\ell \cup \{\Sgood_{3\Snumb+1}\}) &< \prvalu(C_1) +  \alpha &\\
		&= \prvalu(\{\Sgood_{1},\,\Sgood_{2\Snumb+1},\,\Sgood_{3\Snumb}\}) +  \alpha & \mbox{$\ell \neq 1$},\\
		&\le \ddotmms{\prvalu} &\mbox{\cref{cond:two}}. 
	\end{align*}
	Since each of the other bundles also has value at most $\ddotmms{\prvalu}$, summing yields
	$\prvalu(\Sgoods)<\Snumb\;\ddotmms{\prvalu}$,
	again a contradiction. Therefore, agent \(\agent_i\) would accept bag \(\ell\) upon adding \(\Sgood_{3\Snumb+1}\). The algorithm then fills the remaining bags in descending order from \(\ell-1\) down to \(1\). As shown above, for each \(k\notin\{1,\ell\}\), agent \(\agent_i\) would accept bag \(k\) after adding \(\Sgood_{3\Snumb+k}\). Hence, if the algorithm reaches bag \(1\) with \(\prvalu(C_1)<\alpha\), then we have:
	
	\begin{align*}
		\prvalu(C_1) + \prvalu(C_\ell \cup \{\Sgood_{3\Snumb+1}\}) &< \alpha +  \prvalu(C_\ell \cup \{\Sgood_{3\Snumb+1}\}) &\\
		&\le 2\alpha + \prvalu(\{\Sgood_{3\Snumb+1}\}) &\\
		&\le 2\;\ddotmms{\prvalu} &\mbox{\cref{cond:four}}. 
	\end{align*}
	Since each of the other bundles also has value at most $\ddotmms{\prvalu}$, summing yields
	$\prvalu(\Sgoods)<\Snumb\;\ddotmms{\prvalu}$,
	again a contradiction.
\end{proof}

Using \cref{lem:way}, we analyze two groups of green agents. The first group consists of those who satisfy \(\dotmms{\valu_i} \geq 1\) after the primary reductions and \(\ddotmms{\valu_i} \geq 1\) after the secondary reductions (see \cref{plus-one}). The second group includes agents with \(\dotmms{\valu_i} < 1\) after the primary reductions and \(\ddotmms{\auxfun{\FJ}{\agenti}} \geq 4(1-\alpha)\) after the secondary reductions (see \cref{minus-one}).

\begin{lemrep}\label{plus-one}
	Every green agent $\agent_i \in \Sagents$ with $\dotmms{\valu_i} \ge 1$ after the primary reductions, and $\ddotmms{\valu_i} \ge 1$ after the secondary reductions, receives a bundle in \cref{algo:1}.
\end{lemrep}
\begin{proofsketch}
	We analyze two cases based on the value agent \( \agent_i \) assigns to good \(\Sgood_{3\Snumb}\). If this value is at most \(\tfrac{4\alpha}{3} - 1\), we show that  \(\auxwal{\vGF}\) satisfies all three conditions of \cref{lem:way}. Otherwise, we assume \(\valu_i(\{\Sgood_{3\Snumb}\}) > \tfrac{4\alpha}{3} - 1\), and consider whether for any  $2 \le k \le \Snumb$, the bundle \(\{\Sgood_k,\, \Sgood_{\Snumb+k},\,\Sgood_{3\Snumb-k+1}\}\) has value at least 1. If none do, we again verify all conditions of \cref{lem:way} for \(\valu_i\). If for some $2 \le k \le \Snumb$, the bundle $\{\Sgood_k,\, \Sgood_{\Snumb+k},\,\Sgood_{3\Snumb-k+1}\}$ has value at least 1, we prove that all such bundles have value at least \(\alpha\), ensuring the agent receives a bundle. 
\end{proofsketch}
\begin{proof}
	Note that by \cref{lem:virt-mms-before}, we have $\ddotmms{\auxfun{\FJ}{\agenti}} \ge 4(1-\alpha)$, therefore by \cref{lem:virtG-mms}, we conclude  $\ddotmms{\auxfun{\vGF}{\agenti}} \ge 4(2-\tfrac{7\alpha}{3})$. 
	
	If $\boldsymbol{\valu_\agenti(\{\Sgood_{3\Snumb}\}) \leq \tfrac{4\alpha}{3}-1}$, we show that function $\auxwal{\vGF}$ satisfies conditions of \cref{lem:way}: 
	\begin{itemize}
		\item \cref{cond:sefr} follows directly by \cref{def:adj}.
		\item For \cref{cond:yek} we have:
		\begin{align*}
			\alpha + \gf{\{\Sgood_{3\Snumb+1}\}} &< \alpha + \frac{4\alpha}{3} - 1 & \\
			&\le 4\left(2-\frac{7\alpha}{3}\right) &\alpha \le \frac{27}{35},\\
			&\le \ddotmms{\auxfun{\vGF}{\agenti}}. 
		\end{align*}
		\item For \cref{cond:do} we have:
		\begin{align*}
			&\gf{\{\Sgood_{k},\,\Sgood_{\Snumb+k} ,\,\Sgood_{3\Snumb-k+1}\}}\\
			&\le \Bigl(3-\frac{7\alpha}{2}\Bigr) + \Bigl(3-\frac{7\alpha}{2}\Bigr) + \Bigl(2-\frac{7\alpha}{3}\Bigr) &\mbox{\cref{lem:N1-upper-bounds}},\\
			&= 4\left(2-\frac{7\alpha}{3}\right).
		\end{align*}
		\item To prove \cref{cond:se}, first we show $\gf{\{\Sgood_{1}, \Sgood_{\Snumb+1}, \Sgood_{3\Snumb}\}} < \tfrac{7\alpha}{3}-1$:
		\begin{align*}
			&\gf{\{\Sgood_{1}, \Sgood_{\Snumb+1}, \Sgood_{3\Snumb}\}} &\\
			&< \alpha + \gf{\{\Sgood_{3\Snumb}\}} &\mbox{$\rstar^2$ is not applicable},\\
			&\le \alpha + \left(\frac{4\alpha}{3}-1\right) &\\
			&= \frac{7\alpha}{3}-1.
		\end{align*}
		Since $\alpha \le \tfrac{17}{22}$, we have $\tfrac{7\alpha}{3}-1 + \alpha \le 2(4(2-\tfrac{7\alpha}{3}))$, therefore we can verify \cref{cond:se}. 
	\end{itemize}
	Thus, we assume $\boldsymbol{\valu_\agenti(\{\Sgood_{3\Snumb}\}) > \tfrac{4\alpha}{3}-1}$, in the rest of the proof. 
	
	Now, consider for all $2 \le k \le \Snumb$, $\valu_\agenti(\{\Sgood_{k},\, \Sgood_{\Snumb+k},\, \Sgood_{3\Snumb-k+1}\}) < 1$ holds. Under this assumption, we show that function $\valu_i$ satisfies conditions of \cref{lem:way}: 
	\begin{itemize}
		\item \cref{cond:sefr} trivially holds.
		\item To prove \cref{cond:yek}, we have: 
		\begin{align*}
			\alpha + \valu_\agenti(\{\Sgood_{3\Snumb+1}\}) &< \alpha + \frac{\alpha}{4} & \mbox{$\reductiontype^3$ is not applicable},\\
			&\le 1 &\alpha\le \frac{4}{5}, \\
			&\le \ddotmms{\valu_\agenti}. 
		\end{align*}
		
		\item \cref{cond:do} holds by the assumption that for each $k > 1$ we have $\valu_i(\{\Sgood_{k}\, \Sgood_{\Snumb+k}\, \Sgood_{3\Snumb-k+1}\}) < 1$. 
		
		\item To prove \cref{cond:se}, first we show $\valu_\agenti(\{\Sgood_{1}\, \Sgood_{\Snumb+1}\, \Sgood_{3\Snumb}\}) < \tfrac{4\alpha}{3}$:
		\begin{align*}
			& \valu_\agenti(\{\Sgood_{1},\;\Sgood_{\Snumb+1},\;\Sgood_{3\Snumb}\}) &\\
			&= \valu_\agenti(\{\Sgood_{1},\,\Sgood_{\Snumb+1}\})+\valu_\agenti(\{\Sgood_{3\Snumb}\}) &\\
			&< \alpha + \valu_\agenti(\{\Sgood_{3\Snumb}\}) &\mbox{$\rstar^2$ is not applicable},\\
			&\le \alpha + \frac{\alpha}{3} &\mbox{\cref{lem:N1-upper-bounds}},\\
			&= \frac{4\alpha}{3}. 
		\end{align*}
		For $\alpha \le \tfrac{6}{7}$ we have $\tfrac{4\alpha}{3} + \alpha \le 2$. 
	\end{itemize}
	
	Next we show if for some $2 \le k \le \Snumb$, $\valu_\agenti(\{\Sgood_{k},\;\Sgood_{\Snumb+k},\;\Sgood_{3\Snumb-k+1}\}) \ge 1$ holds, then for all $1 \le k \le \Snumb$ we have $\valu_\agenti(\{\Sgood_{k},\;\Sgood_{\Snumb+k},\;\Sgood_{3\Snumb-k+1}\}) \ge \alpha$ which implies agent $\agent_i$ receives a bundle. 
	We first show that if $\valu_\agenti(\{\Sgood_{k}\, \Sgood_{\Snumb+k}\, \Sgood_{3\Snumb-k+1}\}) \geq 1$, then
	$$\valu_\agenti(\{\Sgood_{\Snumb+k}\}) \geq 1-\tfrac{5\alpha}{6} \quad \text{and} \quad \valu_\agenti(\{\Sgood_{3\Snumb-k+1}\}) \geq 1-\alpha.$$
	
	Since $\rstar^2$ is not applicable, we have 
	$\valu_\agenti(\{\Sgood_k,\,\Sgood_{\Snumb+k}\}) \leq \alpha.$
	Therefore,
	$\valu_\agenti(\{\Sgood_{3\Snumb-k+1}\}) \geq 1-\alpha.$
	By \cref{lem:N1-upper-bounds} we have 
	$$\valu_\agenti(\{\Sgood_k\}) \le \tfrac{\alpha}{2} \quad \text{and} \quad \valu_\agenti(\{\Sgood_{3\Snumb-k+1}\}) \le \tfrac{\alpha}{3}.$$
	Hence,
	$\valu_\agenti(\{\Sgood_{\Snumb+k}\}) \geq 1-\tfrac{5\alpha}{6}.$
	
	Next, we prove that for every index $j > k$ we have $\valu_\agenti(\{\Sgood_{j},\,\Sgood_{\Snumb+j},\,\Sgood_{3\Snumb-j+1}\})\ge \alpha$: 
	\begin{align*}
		&\valu_\agenti(\{\Sgood_{j},\;\Sgood_{\Snumb+j},\;\Sgood_{3\Snumb-j+1}\}) &\\
		&\geq \valu_\agenti(\{\Sgood_{\Snumb+k},\;\Sgood_{3\Snumb-k+1},\;\Sgood_{3\Snumb-k+1}\})  &\\
		&\geq \Bigl(1-\frac{5\alpha}{6}\Bigr) + (1-\alpha) + (1-\alpha)&\\
		&\geq \alpha &\alpha\le \frac{18}{23}.
	\end{align*}
	
	Finally, we show that for every index $j < k$ it holds that 
	\begin{align*}
		&\quad \valu_\agenti(\{\Sgood_{j},\,\Sgood_{\Snumb+j},\,\Sgood_{3\Snumb-j+1}\})\\
		&\geq \valu_\agenti(\{\Sgood_{k},\,\Sgood_{\Snumb+k},\,\Sgood_{3\Snumb}\})\\
		&= \Bigl(\valu_\agenti(\{\Sgood_{k},\, \Sgood_{\Snumb+k},\, \Sgood_{3\Snumb-k+1}\}) - \valu_\agenti(\{\Sgood_{3\Snumb+1-k}\})\Bigr) + \valu_\agenti(\{\Sgood_{3\Snumb}\})\\ 
		&\geq \Bigl(1 - \valu_\agenti(\{\Sgood_{3\Snumb+1-k}\})\Bigr) + \valu_\agenti(\{\Sgood_{3\Snumb}\})\\
		&\geq \Bigl(1 - \valu_\agenti(\{\Sgood_{3\Snumb+1-k}\})\Bigr) + \frac{4\alpha}{3}-1\\
		&\geq \Bigl(1 - \frac{\alpha}{3}\Bigr) + \frac{4\alpha}{3}-1 &\mbox{\cref{lem:N1-upper-bounds}},\\
		&= \alpha.
	\end{align*}
	Completing the proof. 
\end{proof}

\begin{lemrep}\label{minus-one}
	Every green agent $\agent_i \in \Sagents$ with $\dotmms{\valu_i} < 1$ after the primary reductions, and $\ddotmms{\auxfun{\FJ}{\agenti}} \ge 4(1-\alpha)$ after the secondary reductions, receives a bundle in \cref{algo:1}.
\end{lemrep}
\begin{proofsketch}
	We analyze two cases based on the value agent \( \agent_i \) assigns to good \(\Sgood_{3\Snumb}\). If this value is at most \(\tfrac{4\alpha}{3} - 1\), we show that  \(\auxwal{\vGF}\) satisfies all conditions of \cref{lem:way}. Otherwise, we assume \(\valu_i(\{\Sgood_{3\Snumb}\}) > \tfrac{4\alpha}{3} - 1\), and consider whether for any $2 \le k \le \Snumb$, the bundle \(\{\Sgood_k,\, \Sgood_{\Snumb+k},\,\Sgood_{3\Snumb-k+1}\}\) has value at least $4(1-\alpha)$, under the function $(\FJ\star\valu_\agenti)$. If none do, we again verify all conditions of \cref{lem:way} for \((\FJ\star\valu_\agenti)\). If for some $2 \le k \le \Snumb$, $$\auxagent{\FJ}{\{\Sgood_k,\, \Sgood_{\Snumb+k},\,\Sgood_{3\Snumb-k+1}\}}{\agenti} \ge 4(1-\alpha),$$ we prove that all such bundles have value at least \(\alpha\), ensuring the agent receives a bundle. 
\end{proofsketch}
\begin{proof}
	If $\boldsymbol{\valu_\agenti(\{\Sgood_{3\Snumb}\}) \leq \tfrac{4\alpha}{3}-1}$, we show that function $\auxwal{\vGF}$ satisfies conditions of \cref{lem:way}. 
	
	\begin{itemize}
		\item \cref{cond:sefr} follows directly by \cref{def:adj}.
		\item For \cref{cond:yek} we have:
		\begin{align*}
			\alpha + \gf{\{\Sgood_{3\Snumb+1}\}} &< \alpha + \frac{4\alpha}{3} - 1 & \\
			&\le 4\left(2-\frac{7\alpha}{3}\right) &\alpha \le \frac{27}{35},\\
			&\le \ddotmms{\auxfun{\vGF}{\agenti}} & \mbox{\cref{lem:virtG-mms}}. 
		\end{align*}
		\item For \cref{cond:do} we have:
		\begin{align*}
			&\quad \gf{\{\Sgood_{k},\,\Sgood_{\Snumb+k},\,\Sgood_{3\Snumb-k+1}\}} &\\
			&\le \Bigl(3-\frac{7\alpha}{2}\Bigr) + \Bigl(3-\frac{7\alpha}{2}\Bigr) + \Bigl(2-\frac{7\alpha}{3}\Bigr) &\mbox{\cref{lem:N1-upper-bounds}},\\
			&= 4\left(2-\frac{7\alpha}{3}\right)\\
			&\le \ddotmms{\auxfun{\vGF}{\agenti}} & \mbox{\cref{lem:virtG-mms}}.
		\end{align*}
		\item To prove \cref{cond:se}, first we show $\gf{\{\Sgood_{1},\,\Sgood_{\Snumb+1},\,\Sgood_{3\Snumb}\}} < \tfrac{7\alpha}{3}-1$:
		\begin{align*}
			&\quad \gf{\{\Sgood_{1},\,\Sgood_{\Snumb+1},\,\Sgood_{3\Snumb}\}}&\\
			&< \alpha + \gf{\{\Sgood_{3\Snumb}\}} &\mbox{$\rstar^2$ is not applicable},\\
			&\le \alpha + \left(\frac{4\alpha}{3}-1\right) &\\
			&= \frac{7\alpha}{3}-1.
		\end{align*}
		Since $\alpha \le \tfrac{17}{22}$ we have $\tfrac{7\alpha}{3}-1 + \alpha \le 2(4(2-\tfrac{7\alpha}{3}))$, therefore by \cref{lem:virtG-mms} we can verify \cref{cond:se}.
	\end{itemize}
	Thus, we assume $\boldsymbol{\valu_\agenti(\{\Sgood_{3\Snumb}\}) > \tfrac{4\alpha}{3}-1}$ in the rest of the proof. 
	
	Now, consider for all $2 \le k \le \Snumb$, $\auxagent{\FJ}{\{\Sgood_{k},\,\Sgood_{\Snumb+k},\,\Sgood_{3\Snumb}\}}{\agenti} < 4(1-\alpha)$ holds. Under this assumption, we show that function $\auxval{\FJ}$ satisfies conditions of \cref{lem:way}: 
	\begin{itemize}
		\item \cref{cond:sefr} follows directly by \cref{def:adj}.
		\item For \cref{cond:yek} we have:
		\begin{align*}
			&\quad \alpha+\auxagent{f}{\{\Sgood_{3\Snumb+1}\}}{\agenti} &\\
			&<\alpha+\frac{4\alpha}{3}-1
			&\mbox{\cref{obs:3n-kam}},\\
			&\le 4(1-\alpha)
			&\alpha\le\frac{15}{19},\\
			&\le \ddotmms{\auxfun{\FJ}{\agenti}}. 
		\end{align*}
		
		\item \cref{cond:do} holds by the assumption 
		$
		\auxagent{\FJ}{\{\Sgood_{k},\,\Sgood_{\Snumb+k},\,\Sgood_{3\Snumb-k+1}\}}{\agenti}<4(1-\alpha). 
		$
		
		\item To prove \cref{cond:se}, we show first that
		\[
		\ff{\{\Sgood_{1},\,\Sgood_{\Snumb+1},\,\Sgood_{3\Snumb}\}}
		=\ff{\{\Sgood_1,\,\Sgood_{\Snumb+1}\}}+\ff{\{\Sgood_{3\Snumb}\}}
		<\alpha+(1-\alpha)=1,
		\]
		where the inequality uses the fact that $\rstar^2$ is not applicable and \cref{lem:N1-upper-bounds}. Since $\alpha\le\tfrac{7}{9}$, it follows that
		\[
		\alpha + \ff{\{\Sgood_{1},\,\Sgood_{\Snumb+1},\,\Sgood_{3\Snumb}\}}\;\le\;8(1-\alpha).
		\]
	\end{itemize}
	
	Now, we show that if for some $2 \le k \le \Snumb$, $\auxagent{\FJ}{\{\Sgood_{k},\,\Sgood_{\Snumb+k},\,\Sgood_{3\Snumb-k+1}\}}{\agenti} \geq 4(1-\alpha)$ holds, then for all $1 \le j \le \Snumb$ we have $\auxagent{\FJ}{\{\Sgood_{j},\,\Sgood_{\Snumb+j},\,\Sgood_{3\Snumb-j+1}\}}{\agenti} \geq \alpha$ which implies agent $\agent_i$ receives a bundle. 
	We first show that 
	$\auxagent{\FJ}{\{\Sgood_{\Snumb+k}\}}{\agenti} \geq 2 - \tfrac{13\alpha}{6}.$ 
	By \cref{lem:N1-upper-bounds} we have 
	$$\auxagent{\FJ}{\{\Sgood_{k}\}}{\agenti} \le 1- \tfrac{5\alpha}{6} \quad \text{and} \quad \auxagent{\FJ}{\{\Sgood_{3\Snumb-k+1}\}}{\agenti} \le 1-\alpha.$$ 
	Therefore 
	\begin{align*}
		\auxagent{\FJ}{\{\Sgood_{\Snumb+k}\}}{\agenti}  &\geq 4(1-\alpha)-\Bigl(1- \frac{5\alpha}{6}\Bigr)-(1-\alpha) \\
		&= 2 - \frac{13\alpha}{6}.
	\end{align*}
	
	By the assumption  $\valu_\agenti(\{\Sgood_{3\Snumb}\})>\tfrac{4\alpha}{3}-1$, and by \cref{lem:N-minus}, we have $\valu_\agenti(\{\Sgood_{3\Snumb}\})>1-\alpha$, which by definition of $\FJ$ we can conclude that $\auxagent{\FJ}{\{\Sgood_{3\Snumb}\}}{\agenti}\ge 1-\alpha$. Thus, for any $1 \le j \le \Snumb$, we have:
	\begin{align*}
		&\quad  \ff{\{\Sgood_{j},\,\Sgood_{\Snumb+j},\,\Sgood_{3\Snumb-j+1}\}} \\
		&\geq \ff{\{\Sgood_{\Snumb+k},\,\Sgood_{3\Snumb},\,\Sgood_{3\Snumb}\}} \\
		&\geq \Bigl(2 - \frac{13\alpha}{6}\Bigr) + (1-\alpha) + (1-\alpha) &\\
		&\geq \alpha &&\alpha\le \frac{24}{31}.
	\end{align*}
	Completing the proof. 
\end{proof}

We handle the remaining two groups using similar arguments: agents with $\dotmms{\valu_i} \geq 1$ after the primary reductions but $\ddotmms{\valu_i} < 1$ after the secondary reductions (\cref{plus-two}), and those with $\dotmms{\valu_i} < 1$  after the primary reductions and $\ddotmms{\auxfun{\FJ}{\agenti}} < 4(1-\alpha)$ after the secondary reductions (\cref{minus-two}).

\begin{lemrep}\label{plus-two}
	Every green agent $\agent_i$ in $\Sagents$ such that  $\dotmms{\valu_\agenti} \ge 1$ after the primary reductions and $\ddotmms{\valu_\agenti} < 1$ after the secondary reductions, receives a bundle in \cref{algo:1}.
\end{lemrep}
\begin{proofsketch}
	Noting that by \cref{lem:plus-two}, there exists a number $t_\agenti$ satisfying \Cref{cond1}, \Cref{cond2}, and \Cref{cond3}, we verify that valuation function $\valu_\agenti$ satisfies all conditions required by \cref{lem:way2}. 
\end{proofsketch}
\begin{proof}
	Note that by \cref{lem:plus-two}, there exists a number $t_\agenti$ satisfying conditions \Cref{cond1}, \Cref{cond2}, and \Cref{cond3}. We verify that valuation function $\valu_\agenti$ satisfies all conditions required by \cref{lem:way2}:
	
	\begin{itemize}
		\item \cref{cond:zero} holds trivially.
		\item For \cref{cond:one}, we have: 
		\begin{align*}
			\alpha + \valu_\agenti(\{\Sgood_{4\Snumb+1}\}) 
			&< \alpha + \frac{\alpha}{5} &\mbox{ $\reductiontype^4$ is not applicable},\\[6pt]
			&\le 1 - 2\left(2\alpha - \frac{3}{2}\right) &\alpha \le \frac{10}{13},\\[6pt]
			&\le 1 - 2 t_\agenti &\mbox{\cref{cond1}},\\[6pt]
			&\le \ddotmms{\valu_\agenti} &\mbox{ \cref{cond3}}.
		\end{align*}
		
		\item For \cref{cond:two}, first, we show that 
		$\valu_\agenti(\{\Sgood_{2}\}) \le \alpha - \tfrac{1}{2} - t_\agenti.$
		
		\begin{align*}
			\valu_\agenti(\{\Sgood_{2}\}) &\le \frac{\alpha}{2}&\mbox{\cref{lem:N1-upper-bounds}},\\
			&\le \frac{1}{2} - \Bigl(2\alpha-\frac{3}{2}\Bigr) & \alpha \le \frac{4}{5},\\
			&\le \frac{1}{2} - t_\agenti &\mbox{\cref{cond1}}.
		\end{align*}
		
		\Cref{cond2} ensures no good's value lies in interval 
		$\left[\alpha - \tfrac{1}{2} - t_\agenti,\, \tfrac{1}{2} - t_\agenti\right].$
		Thus, since $\tfrac{\alpha}{2}$ is within this interval, we conclude $\valu_\agenti(\{\Sgood_{2}\}) \le \alpha - \tfrac{1}{2} - t_\agenti$.
		
		Now we bound $\valu_\agenti(\{\Sgood_{1},\,\Sgood_{\Snumb+1},\,\Sgood_{3\Snumb}\})$:
		\begin{align*}
			&\quad \valu_\agenti(\{\Sgood_{1},\,\Sgood_{\Snumb+1},\,\Sgood_{3\Snumb}\}) &\\
			&= \valu_\agenti(\{\Sgood_{1}\}) + \valu_\agenti(\{\Sgood_{\Snumb+1},\,\Sgood_{3\Snumb}\}) &\\[6pt]
			&\le \valu_\agenti(\{\Sgood_{1}\}) + 2\,\valu_\agenti(\{\Sgood_{2}\}) &\\[6pt]
			&\le \left(t_\agenti + \frac{1}{2}\right) + 2\left(\alpha - \frac{1}{2} - t_\agenti \right) &\mbox{\cref{cond1}},\\[6pt]
			&= 2\alpha - \frac{1}{2} - t_\agenti.
		\end{align*}
		Therefore
		\begin{align*}
			&\quad \alpha+\valu_\agenti(\{\Sgood_{1},\,\Sgood_{\Snumb+1},\,\Sgood_{3\Snumb}\}) &\\
			&\le 3\alpha - \frac{1}{2} - t_\agenti\\
			&\le 2-3\Bigl(2\alpha-\frac{3}{2}\Bigr)- t_\agenti & \alpha \le \frac{7}{9},\\
			&\le 2(1-2t_i)&\mbox{\cref{cond1}},\\
			&\le 2\ddotmms{\valu_\agenti}&\mbox{\cref{cond3}}.
		\end{align*}

		\item To prove \cref{cond:three}, for all $2 \le k \le \Snumb$, we have:
		\begin{align*}
			&\quad \valu_\agenti(\{\Sgood_{k},\,\Sgood_{\Snumb+k},\,\Sgood_{3\Snumb-k+1},\,\Sgood_{3\Snumb+k}\}) &\\
			&\le 3\left(\alpha - \frac{1}{2} - t_\agenti\right) + \valu_\agenti(\{\Sgood_{3\Snumb+k}\}) &\\[6pt]
			&\le 3\left(\alpha - \frac{1}{2} - t_\agenti\right) + \frac{\alpha}{4} &\mbox{$\reductiontype^3$ is not applicable},\\[6pt]
			&= 1 - 2 t_\agenti - \left(\frac{5}{2} - \frac{13\alpha}{4} + t_\agenti\right) &\\[6pt]
			&< 1 - 2 t_\agenti - \left(\frac{5}{2} - \frac{13\alpha}{4} \right) &t_\agenti>0,\\[6pt]
			&\le 1 - 2 t_\agenti &\alpha \le \frac{10}{13},\\
			&\le \ddotmms{\valu_\agenti} &\mbox{ \cref{cond3}}.
		\end{align*}
		
		\item For \cref{cond:four}, we have: 
		\begin{align*}
			2\alpha + \valu_\agenti(\{\Sgood_{3\Snumb+1}\}) &< 2\alpha + \frac{\alpha}{4} & \mbox{$\reductiontype^3$ is not applicable},\\
			&= 2(1 - 2t_\agenti) - \Bigl(2(1 - 2t_\agenti) - \frac{9\alpha}{4}\Bigr) & \\
			&\leq 2(1 - 2t_\agenti) - \Bigl(2\Bigl(1 - 2\Bigl(2\alpha - \frac{3}{2}\Bigr)\Bigr) - \frac{9\alpha}{4}\Bigr) & \mbox{\cref{cond1}},\\
			&\leq 2(1 - 2t_\agenti) & \alpha \le \frac{32}{41}, &\\
			&\le 2\;\ddotmms{\valu_\agenti} &\mbox{ \cref{cond3}}.
		\end{align*}
		
		\item To prove \cref{cond:five}, note that we have already shown that $\valu_i(\{\Sgood_{1},\,\Sgood_{\Snumb+1},\,\Sgood_{3\Snumb}\}) \le 2\alpha - \tfrac{1}{2} - t_\agenti$. Therefore we have: 
		\begin{align*}
			&\quad 2\alpha + \valu_\agenti(\{\Sgood_{1},\,\Sgood_{\Snumb+1},\,\Sgood_{3\Snumb},\,\Sgood_{3\Snumb+1}\})&\\
			&\le 2\alpha + (2\alpha - \frac{1}{2} - t_\agenti) + \valu_\agenti(\{\Sgood_{3\Snumb+1}\}) &\\
			&\le 2\alpha + (2\alpha - \frac{1}{2} - t_\agenti) + \frac{\alpha}{4} &\mbox{$\reductiontype^3$ is not applicable},\\
			&= 3(1 - 2t_\agenti) - \Bigl(\frac{7}{2} - 5t_\agenti - \frac{17\alpha}{4}\Bigr) \\
			&\leq 3(1 - 2t_\agenti) - \Bigl(\frac{7}{2} - 5\Bigl(2\alpha - \frac{3}{2}\Bigr) - \frac{17\alpha}{4}\Bigr) & \mbox{\cref{cond1}}, \\
			&\leq 3(1 - 2t_\agenti) & \alpha \le \frac{44}{57}, &\\
			&\le 3\;\ddotmms{\valu_\agenti} &\mbox{ \cref{cond3}}.
		\end{align*}
	\end{itemize}
	
	Since all conditions of \cref{lem:way2} are satisfied, the proof is complete.
\end{proof}

\begin{lemrep}\label{minus-two}
	Every green agent $\agent_i$ in $\Sagents$ such that  $\dotmms{\valu_i} < 1$ after the primary reductions and $\ddotmms{\auxfun{\FJ}{\agenti}} < 4(1-\alpha)$ after the secondary reductions, receives a bundle in \cref{algo:1}.
\end{lemrep}
\begin{proofsketch}
Noting that by \cref{lem:minus-two}, there exists a number $t_\agenti$ satisfying \cref{condd1}, \Cref{condd2}, and \cref{condd3}, we verify that valuation function $\auxval{\vHB_{t_\agenti}\star\FJ}$ satisfies all conditions required by \cref{lem:way2}. 
\end{proofsketch}
\begin{proof}
	Note that by \cref{lem:plus-two}, there exists a number $t_\agenti$ satisfying conditions \cref{condd1}, \Cref{condd2}, and \cref{condd3}. We verify that valuation function $\auxval{\vHB_{t_\agenti}\star\FJ}$ satisfies all conditions required by \cref{lem:way2}:
	
	\begin{itemize}
		\item \cref{cond:zero} follows directly by \cref{def:adj}.
		\item To prove \cref{cond:one}, we show: 
		\begin{align*}
			&\quad \alpha + \auxwal{\vHB_{t_\agenti}\star\FJ}(\{\Sgood_{4\Snumb+1}\})&\\\ 
			&< \alpha + \frac{4\alpha}{3}-1 & \mbox{\cref{obs:3n-kam}},\\
			&\le 4(1-\alpha)-2\Bigl(2\alpha-\frac{3}{2}\Bigr) &\alpha \le \frac{24}{31},\\
			&\le 4(1-\alpha)-2t_i &\mbox{\cref{condd1}},\\
			&\le \ddotmms{\auxfun{\vHB_{t_\agenti}\star\FJ}{\agenti}} &\mbox{\cref{condd3}}.
		\end{align*}
		
		\item To prove \cref{cond:two}, first, we show that 
		$\auxwal{\FJ}(\{\Sgood_{2}\}) \le \tfrac{5\alpha}{3} - 1 - t_\agenti$. 
		
		\begin{align*}
			\auxwal{\FJ}(\{\Sgood_{2}\}) &\le 1-\frac{5\alpha}{6}&\mbox{\cref{lem:N1-upper-bounds}},\\
			&\le 2(1-\alpha) - \Bigl(2\alpha-\frac{3}{2}\Bigr) & \alpha \le \frac{15}{19},\\
			&\le 2(1-\alpha) - t_\agenti &\mbox{\cref{condd1}}.
		\end{align*}

		\Cref{condd2} ensures no good's value lies in interval 
		$\Bigl[\tfrac{5\alpha}{3} - 1 - t_\agenti,\; 2(1-\alpha) - t_\agenti\Bigr].$
		Thus, since $1-\tfrac{5\alpha}{6}$ is within this interval, we conclude $\auxwal{\FJ}(\{\Sgood_{2}\}) \le \tfrac{5\alpha}{3} - 1 - t_\agenti$.
		
		Now we bound $\auxwal{\vHB_{t_\agenti}\star\FJ}(\{\Sgood_{1},\,\Sgood_{\Snumb+1},\,\Sgood_{3\Snumb}\})$:
		\begin{align*}
			&\quad \auxwal{\vHB_{t_\agenti}\star\FJ}(\{\Sgood_{1},\,\Sgood_{\Snumb+1},\,\Sgood_{3\Snumb}\})&\\
			&\le \auxwal{\vHB_{t_\agenti}\star\FJ}(\{\Sgood_{1}\}) + \auxwal{\FJ}(\{\Sgood_{\Snumb+1}\}) + \auxwal{\FJ}(\{\Sgood_{3\Snumb}\})\\[6pt]
			&\le \auxwal{\vHB_{t_\agenti}\star\FJ}(\{\Sgood_{1}\}) + \auxwal{\FJ}(\{\Sgood_{2}\})+(1-\alpha)& \mbox{\cref{lem:N1-upper-bounds}},\\[6pt]
			&\le \auxwal{\vHB_{t_\agenti}\star\FJ}(\{\Sgood_{1}\}) + \Bigl(\frac{5\alpha}{3} - 1 - t_\agenti\Bigr)+(1-\alpha) &\\[6pt]
			&\le \vHB_{t_\agenti}(2(1-\alpha)+t_\agenti) + \Bigl(\frac{5\alpha}{3} - 1 - t_\agenti\Bigr)+(1-\alpha) &\mbox{\cref{condd1}},\\[6pt]
			&= (2(1-\alpha)-t_\agenti) + \Bigl(\frac{5\alpha}{3} - 1 - t_\agenti\Bigr)+(1-\alpha) &\\[6pt]
			&= 2-\frac{4\alpha}{3}-2t_i.
		\end{align*}
		Therefore
		\begin{align*}
			&\quad \alpha + \auxwal{\vHB_{t_\agenti}\star\FJ}(\{\Sgood_{1},\,\Sgood_{\Snumb+1},\,\Sgood_{3\Snumb}\})&\\
			&\le 2-\frac{\alpha}{3}-2t_i\\
			&=2(4(1-\alpha)-2t_i)+\Bigl(-6 + \frac{23\alpha}{3} + 2t_i\Bigr)\\
			&\le 2(4(1-\alpha)-2t_i)+\Bigl(-6 + \frac{23\alpha}{3} + 2\Bigl(2\alpha - \frac{3}{2}\Bigr)\Bigr)&\mbox{\cref{condd1}},\\
			&\le2(4(1-\alpha)-2t_i) &\alpha \le \frac{27}{35},\\
			&\le2\ddotmms{\auxfun{\vHB_{t_\agenti}\star\FJ}{\agenti}} &\mbox{\cref{condd3}}.
		\end{align*}
		
		\item To prove \cref{cond:three}, for all $2 \le k \le \Snumb$, we show:
		\begin{align*}
			&\quad \auxwal{\vHB_{t_\agenti}\star\FJ}(\{\Sgood_{k},\,\Sgood_{\Snumb+k},\,\Sgood_{3\Snumb-k+1},\,\Sgood_{3\Snumb+k}\})&\\
			&\le 2\ff{\{\Sgood_{2}\}} + \ff{\{\Sgood_{3\Snumb-k+1}, \Sgood_{3\Snumb+k}\}} & k \ge 2, \\
			&\le 2\Bigl(\frac{5\alpha}{3} - 1 - t_\agenti\Bigr) + \ff{\{\Sgood_{3\Snumb-k+1}, \Sgood_{3\Snumb+k}\}} &\\
			&\le 2\Bigl(\frac{5\alpha}{3} - 1 - t_\agenti\Bigr) + 1-\alpha + \ff{\{\Sgood_{3\Snumb+k}\}} & \mbox{\cref{lem:N1-upper-bounds}}, \\
			&\le 2\Bigl(\frac{5\alpha}{3} - 1 - t_\agenti\Bigr) + 1-\alpha + \frac{4\alpha}{3}-1 & \mbox{\cref{obs:3n-kam}}, \\
			&\le 4(1-\alpha) - 2t_\agenti &\alpha \le \frac{18}{23},\\
			&\le \ddotmms{\auxfun{\vHB_{t_\agenti}\star\FJ}{\agenti}}&\mbox{ \cref{condd3}}.
		\end{align*}
		
		\item To prove \cref{cond:four}, we show: 
		\begin{align*}
			&\quad 2\alpha + \auxwal{\vHB_{t_\agenti}\star\FJ}(\{\Sgood_{3\Snumb+1}\})&\\
			&< 2\alpha + \frac{4\alpha}{3}-1 & \mbox{\cref{obs:3n-kam}},\\
			&= 2(4(1-\alpha) - 2t_\agenti) +\frac{34\alpha}{3} - 9 + 4\,t_{\agenti}  & \\
			&\le2(4(1-\alpha) - 2t_\agenti) +\frac{34\alpha}{3} - 9 + 4\Bigl(2\alpha - \frac{3}{2}\Bigr)& \mbox{\cref{condd1}},\\
			&\le 2(4(1-\alpha) - 2t_\agenti) & \alpha \le \frac{45}{58}, \\
			&\le 2\ddotmms{\auxfun{\vHB_{t_\agenti}\star\FJ}{\agenti}} &\mbox{ \cref{condd3}}.
		\end{align*}
		
		\item To prove \cref{cond:five}, note that we have already shown that $\auxwal{\vHB_{t_\agenti}\star\FJ}(\{\Sgood_{1},\,\Sgood_{\Snumb+1},\,\Sgood_{3\Snumb}\}) \le 2-\tfrac{4\alpha}{3}-2t_i$. Therefore we have:
		\begin{align*}
			&\quad 2\alpha+\auxwal{\vHB_{t_\agenti}\star\FJ}(\{\Sgood_{1},\,\Sgood_{\Snumb+1},\,\Sgood_{3\Snumb} ,\,\Sgood_{3\Snumb+1}\})&\\
			&\le 2\alpha + (2-\frac{4\alpha}{3}-2t_i) + \valu_\agenti(\{\Sgood_{3\Snumb+1}\}) &\\
			&\le 2\alpha + (2-\frac{4\alpha}{3}-2t_i) + \frac{4\alpha}{3}-1 &\mbox{\cref{obs:3n-kam}},\\
			&=(1-\alpha)-6t_\agenti+3\alpha+4t_\agenti&\\
			&\le (1-\alpha)-6t_\agenti+3\alpha+4(2\alpha - \frac{3}{2})&\mbox{\cref{condd1}},\\
			&= 3(4(1-\alpha) - 2t_\agenti)+(22\alpha-17)\\
			&\leq 3(4(1-\alpha) - 2t_\agenti) & \alpha \le \frac{17}{22},\\
			&\le 3\;\ddotmms{\auxfun{\vHB_{t_\agenti}\star\FJ}{\agenti}} &\mbox{ \cref{condd3}}.
		\end{align*}
	\end{itemize}
	
	Since all conditions of \cref{lem:way2} are satisfied, the proof is complete.
\end{proof}

\begin{lemma}\label{lem:N2-in-N1}
	Every red agent \(\agent_i \in \Sagents\) receives a bundle in \cref{algo:1}.
\end{lemma}

\begin{proof}
	We argue by contradiction. Suppose there exists a red agent \(\agent_i \in \Sagents\) who does not receive any bundle in \cref{algo:1}. Let the algorithm terminate while filling bag \(\finib_\halt\). This means $\Pnumb - \halt$ agents have received a bundle of value at least $\alpha$, and bags $\finib_1, \finib_2, \ldots, \finib_\halt$ are unallocated. Note that by \Cref{plus-one}, \cref{minus-one}, \cref{plus-two} and \cref{minus-two}, all green agents have received a bundle, and only red agents remain. 
	Our goal is to show $v_\agenti(\Pgoods) < \Pnumb$, 
	contradicting \(\allmms{\valu_\agenti}=1\).
	
	For every agent $\agent_j$ such that $\agent_j$ received a bundle during the primary reductions, secondary reductions, or Bag-filling, we denote $A_j$ by the bundles allocated to her. By the prioritization of red agents in \cref{algo:pr,algo:N11,algo:1}, every bundle allocated to a green agent satisfies
	$v_\agenti(A_j) < \alpha$. 
	We now show that every bundle allocated to a red agent satisfies
	$v_\agenti(A_j) \le 2\alpha$. 
	
\paragraph{Primary Reductions.}
Consider a red agent $\agent_j$ such that we allocated a bundle to her, during the primary reductions, and let $\prinstance = (\pragents, \prgoods)$ be the instance before applying the reduction. 
If \(\agent_j\) receives a bundle via reduction \(\reductiontype^0\), then \(v_\agenti(A_j) \le 1 < 2\alpha\). 
If the bundle is assigned via \(\reductiontype^1\), note that \(\reductiontype^0\) was not applicable at that point, meaning \(v_\agenti(\{\prgood_1\}) < \alpha\). Since \(\reductiontype^1\) allocates two goods, it follows that \(v_\agenti(A_j) < 2\alpha\). The same reasoning applies to \(\rstar^1\).
Finally, if \(\agent_j\) receives a bundle via \(\reductiontype^2\), the inapplicability of \(\reductiontype^1\) at that time implies \(v_\agenti(\{\prgood_{\prnumb+1}\}) < \tfrac{\alpha}{2}\). Since \(\reductiontype^2\) allocates three goods, we get
$
v_\agenti(A_j) < \tfrac{3\alpha}{2} < 2\alpha.
$

\paragraph{Secondary Reductions.} 
Consider a red agent $\agent_j$ such that we allocated a bundle to her, during the secondary reductions, and let $\prinstance = (\pragents, \prgoods)$ be the instance before applying the reduction. 
The bounds for \(\reductiontype^1\) and \(\reductiontype^2\) were already established in the primary reductions. 
For \(\rstar^2\), the same reasoning as in \(\reductiontype^1\) applies, since the bundle size is also 2, leading to \(v_\agenti(A_j) < 2\alpha\). For \(\reductiontype^3\), since \(\reductiontype^2\) was not applicable at the time of allocation, we have $
v_\agenti(\{\prgood_{2\prnumb+1}\}) < \tfrac{\alpha}{3},
$
and hence,
$
v_\agenti(A_j) < \tfrac{4\alpha}{3} < 2\alpha.
$
Similarly, for \(\reductiontype^4\), the inapplicability of \(\reductiontype^3\) implies
$
v_\agenti(\{\prgood_{3\prnumb+1}\}) < \tfrac{\alpha}{4},
$
which leads to
$
v_\agenti(A_j) < \tfrac{5\alpha}{4} < 2\alpha.
$
 \paragraph{Bag‐filling.} For each \(1\le k\le \Snumb\), $\finib_k$ is either $\{\Sgood_k,\,\Sgood_{\Snumb+k},\,\Sgood_{3\Snumb-k+1}\}$ or at least one good is added to it. In the first case, since $\rstar^2$ is not applicable, $\valu_\agenti(\{\Sgood_k,\,\Sgood_{\Snumb+k}\}) < \alpha$, so $\valu_\agenti(\{\Sgood_k,\,\Sgood_{\Snumb+k},\,\Sgood_{3\Snumb-k+1}\}) < 2\alpha$. In the second case, let $\Sgood_x$ be the last good added to it. Since both $\valu_\agenti(\finib_k \setminus \{\Sgood_x\}) < \alpha$ and $\valu_\agenti(\{\Sgood_x\}) < \alpha$, we have \(v_\agenti(\finib_k)<2\alpha\). 
	
	Now we calculate sum of all bundles $\finib_1, \finib_2, \ldots, \finib_\halt$ and $A_j$ for all satisfied agents $\agent_j$. Since all green agents are satisfied, we have: 
	\[
	v_\agenti(\Pgoods) < \alpha\,|\Nyek| + 2\alpha\,(|\Ndo| - \halt) + 2\alpha\,\halt.
	\]
	Using $|\Nyek| \ge \tfrac{n}{\sqrt{2}}$
	we obtain
	\begin{align*}
		v_\agenti(\Pgoods) &< \left(\alpha\,\frac{1}{\sqrt{2}} + 2\alpha\left(1-\frac{1}{\sqrt{2}}\right)\right)\Pnumb &\\
		&< \Pnumb & \alpha < \frac{4+\sqrt{2}}{7} \approx 0.7735.
	\end{align*} 
	This contradicts $\allmms{\valu_\agenti} = 1$.
	Therefore, every red agent in \(\Sagents\) receives a bundle.
\end{proof}

\clearpage
\section{\boldmath \cref{algo:N2}: Less Frequent Green Agents}\label{sec:alg-N2}

In this section, we consider the case where \( |\Nyek| < \lbnone \). To handle this scenario, we use \cref{algo:N2}, which gives priority to green agents. We first show that every red agent receives a bundle. Then, we prove that, due to the prioritization of green agents, each green agent also receives a bundle.

\begin{algorithm}[h]
	\caption{$\textsc{Algorithm-Case2}$}
	\label{algo:N2}
	\textbf{Input:} $\instance, \Nyek, \Ndo$\\
	\textbf{Output:} Allocation satisfying  $(\tfrac{10}{13})$-$\MMS$
	\begin{algorithmic}[1]
		\For {$k$ : $1 \rightarrow \numb$}
		\State $\finib_k \leftarrow \{\good_k,\, \good_{\numb+k}\}$
		\EndFor
		\For {$k : 1 \rightarrow \numb$}
		\While {There does not exist a remaining agent $\agent_i$ s.t. $v_\agenti(B_k) \geq \alpha$}
		\State Add an arbitrary remaining good to $B_k$ 
		\EndWhile
		\State Allocate $B_k$ to $\agent_i$ with $\valu_\agenti(B_k) \ge \alpha$ \CComment{Priority is given to agents in $\Nyek$}
		\EndFor
	\end{algorithmic}
\end{algorithm}

\begin{figure}[t]
	\centering
	\resizebox{\textwidth}{!}{
		\begin{tikzpicture}[x=1.25cm, y=1.3cm, curve scale=1.25, good/.style={draw, rectangle, rounded corners=3pt,minimum width=1cm,minimum height=1cm,inner sep=2pt, font=\small}]
	\foreach \i/\label/\x/\gcol in {
		1/$\good_1$/1/blue,
		2/$\good_2$/2/red,
		3/$\good_k$/3.5/green,
		4/$\good_{\numb-1}$/5/orange,
		5/$\good_{\numb}$/6/yellow,
		6/$\good_{\numb+1}$/7.2/blue,
		7/$\good_{\numb+2}$/8.2/red,
		8/$\good_{\numb+k}$/9.7/green,
		9/$\good_{2\numb-1}$/11.2/orange,
		10/$\good_{2\numb}$/12.2/yellow} {
		\node[good, fill=\gcol!30] (good\i) at (\x,0) {\label};
	}
	
	\curvedraw[blue, thick]{good1}{1.2}{2.0}{good6}{;}
	\curvedraw[red,  thick]{good2}{1.2}{-2.0}{good7}{;}
	\curvedraw[green,thick]{good3}{1.2}{2.0}{good8}{;}
	\curvedraw[orange,thick]{good4}{1.2}{-2.0}{good9}{;}
	\curvedraw[yellow,thick]{good5}{1.2}{2.0}{good10}{;}

	\foreach \i/\label/\x in {
		11/$\good_{2\numb+1}$/13.3,
		12/$\good_{2\numb+2}$/14.2,
		13/$\good_{m-1}$/15.6,
		14/$\good_{m}$/16.5
	} {
		\node[good, scale=0.75, fill=gray!30] (good\i) at (\x,0) {\label};
	}
	
	\node at (2.75,0) {\(\dots\)};  
	\node at (4.25,0) {\(\dots\)};  
	
	\node at (8.95,0) {\(\dots\)};  
	\node at (10.45,0) {\(\dots\)};  
	
	\node at (14.9,0) {\(\dots\)};  
	
	\draw[decorate, thick, decoration={brace,mirror}] 
	(12.9, -0.5) -- (16.9, -0.5)
	node[midway, below=4pt, font=\small]{Bag-filling goods};
	
	\draw[dashed] (6.6,-0.6) -- (6.6,0.6);
	
	\draw[dashed] (12.8,-0.6) -- (12.8,0.6);
\end{tikzpicture}
	}
	\caption{Structure of the bags in \Cref{algo:N2}.}
\end{figure}

In \cref{algo:N2}, we use a simple Bag-filling algorithm. We begin by forming initial bags of the form 
\[
\{ \good_{k},\, \good_{\numb + k} \} \quad \text{for } k = 1, \dots, \numb.
\]
Next, we sequentially add the remaining goods \(\good_{2\numb+1}, \good_{2\numb+2}, \dots, \good_{\It}\) to the bags, one good at a time. As soon as the total value of a bag reaches at least \(\alpha\) for some agent \(\agent_i\), we allocate that bag to \(\agent_i \in \agents\). If multiple agents are eligible at the same time, priority is given to those in \(\Nyek\).

To analyze our algorithm, we categorize red agents into two groups based on their maximin share after the primary reductions:

\begin{enumerate}
	\item Agents with \( \dotmms{\valu_i} \geq 1 \) after the primary reductions.
	
	\item Agents with \( \dotmms{\valu_i} < 1 \) after the primary reductions.
\end{enumerate}

\cref{lem:N2bag} provides general tools for analyzing these agents. In \cref{lem:N21J}, we show that the agents in the first group receive a bundle. \cref{lem:N22J} establishes the same for the second group.

\begin{lemma}\label{lem:N2bag}
	Let $\agent_i$ be an agent in $\agents$, consider integers $0 \le x \le y \le \numb$ and define
	$$\Jgoods=\goods \setminus \bigcup_{k=x+1}^{y} \{\good_k,\,\good_{\numb+k}\}\quad  \text{and}\quad  \Jnumb=\numb-(y-x).$$
	Let $\prvalu$ be a valuation function on $\Jgoods$ that ranks the goods in the same order as \( \valu_i \). Assume the following conditions hold:
	\begin{align}
		&\forall {\good \in \Jgoods}\qquad &\prvalu(\{\good\}) \le \valu_i(\{\good\}),\label[ineq]{wcond} \\ 
		&\forall 1 \le k \le y \qquad &\valu_i(\{\good_k,\,\good_{\numb+k}\}) \ge \alpha, 
		\label[ineq]{firstcond}\\
		&\forall y < k \le \numb \qquad &\prvalu(\{\good_k,\,\good_{\numb+k}\}) < \alpha, 
		\label[ineq]{secondcond}\\
		&&x\prvalu(\{\good_1\}) + \sum_{k=1}^{x} \prvalu(\{\good_{\numb+k}\}) + (\Jnumb-x)(\alpha + \prvalu(\{\good_{2\numb+1}\})) < \Jnumb\mms^{\Jnumb}_{\prvalu}(\Jgoods).\label[ineq]{thirdcond}
	\end{align}
	Then $\agent_i$ receives a bundle in \cref{algo:N2}.
\end{lemma}

\begin{proof}
	Assume, for contradiction, that agent \(\agent_i\) does not receive any bundle and \cref{algo:N2} terminates while filling bag \(\finib_\halt\).
	Note that the bags \(\finib_1, \finib_2, \ldots, \finib_{\numb}\) form a partition of \(\goods\). For \(1 \le k \le y\), we know from \cref{firstcond} that \(\valu_i(\{\good_k,\, \good_{\numb+k}\}) \ge \alpha\), so the algorithm does not add any additional goods to these bags. Thus, \(\finib_k = \{\good_k,\, \good_{\numb+k}\}\) for all \(k \le y\).
	It follows that the following collection of bundles forms a partition of \(\Jgoods\):
	\[
	\{\good_1,\, \good_{\numb+1}\},\, \{\good_2,\, \good_{\numb+2}\},\, \ldots,\, \{\good_x,\, \good_{\numb+x}\},\, \finib_{y+1},\, \finib_{y+2},\, \ldots,\, \finib_{\numb}.
	\]
	Therefore,
	\[
	\prvalu(\Jgoods) = \sum_{k=1}^{x} \prvalu(\{\good_k,\, \good_{\numb+k}\}) + \sum_{k=y+1}^{\numb} \prvalu(\finib_k).
	\]
	
	We now show that
	\[
	\sum_{k=1}^{x} \prvalu(\{\good_k,\, \good_{\numb+k}\}) + \sum_{k=y+1}^{\numb} \prvalu(\finib_k) < \Jnumb \cdot \mms^{\Jnumb}_{\prvalu}(\Jgoods),
	\]
	which contradicts the assumption that \(\prvalu(\Jgoods) \ge \Jnumb \cdot \mms^{\Jnumb}_{\prvalu}(\Jgoods)\).

	For each \( y < k \le \Jnumb \), the bag \(\finib_k\) is either left as \(\{\good_k,\, \good_{\numb+k}\}\), or at least one additional good is added to it during the Bag-filling process. In the first case, by \cref{secondcond}, we have
	$
	\prvalu(\{\good_k,\, \good_{\numb+k}\}) < \alpha.
	$
	In the second case, let \(\good_x\) denote the last good added to the bag (with \(x \ge 2\numb + 1\)). Then by \cref{wcond},
	$
	\prvalu(\finib_k \setminus \{\good_x\}) < \alpha.
	$
	Thus, in both cases,
	\[
	\prvalu(\finib_k) < \alpha + \prvalu(\{\good_{2\numb+1}\}).
	\]
	
	We can now upper bound the total value of \(\Jgoods\) as follows:
	\begin{align*}
		\prvalu(\Jgoods) &= \sum_{k=1}^{x} \prvalu(\{\good_k,\, \good_{\numb+k}\}) + \sum_{k=y+1}^{\numb} \prvalu(\finib_k) \\
		&\le \sum_{k=1}^{x} \prvalu(\{\good_k\}) + \sum_{k=1}^{x} \prvalu(\{\good_{\numb+k}\}) + (\Jnumb - x)(\alpha + \prvalu(\{\good_{2\numb+1}\})) \\
		&\le x \cdot \prvalu(\{\good_1\}) + \sum_{k=1}^{x} \prvalu(\{\good_{\numb+k}\}) + (\Jnumb - x)(\alpha + \prvalu(\{\good_{2\numb+1}\})) \\
		&< \Jnumb \cdot \mms^{\Jnumb}_{\prvalu}(\Jgoods) &\text{ \cref{thirdcond}}.
	\end{align*}
	
	This contradicts the assumption that \(\prvalu(\Jgoods) \ge \Jnumb \cdot \mms^{\Jnumb}_{\prvalu}(\Jgoods)\), completing the proof.
\end{proof}

We now prove Lemmas \ref{lem:N21J} and \ref{lem:N22J}, which together establish that the red agents receive a bundle in \cref{algo:N2}. To maintain the flow of the paper, we present only a brief proof sketch here and defer the full proofs to the appendix.

\begin{lemrep}\label{lem:N21J}
	Every red agent $\agent_i$ in $\agents$ with $\dotmms{\valu_i} \ge 1$ after the primary reductions, receives a bundle in \cref{algo:N2}.
\end{lemrep}
\begin{proofsketch}
	We begin by defining an index \(0 \le y \le \numb\) such that:
	\begin{itemize}
		\item \(v_\agenti(\{\good_k,\,\good_{\numb+k}\}) \ge \alpha\) for all \(1 \le k \le y\), and
		\item \(v_\agenti(\{\good_k,\,\good_{\numb+k}\}) < \alpha\) for all \(y < k \le \numb\).
	\end{itemize}

	Next, we identify the smallest index \(x \le y\) such that the maximin share with respect to the remaining \(\numb - (y - x)\) bundles is at least 1:
	\[
	\mms^{\numb-(y-x)}_{v_\agenti}\left(\goods \setminus \bigcup_{k=x+1}^{y} \{\good_k,\,\good_{\numb+k}\}\right) \ge 1.
	\]
	Such an \(x\) must exist because \(\dotmms{\valu_i} \ge 1\) by assumption.
	Now, define:
	\[
	\Jnumb = \numb - (y - x), \quad \Jgoods = \goods \setminus \bigcup_{k=x+1}^{y} \{\good_k,\,\good_{\numb+k}\},
	\]
	and let by \cref{def:norm} \(\Jvalu = \text{normalized}^{\Jnumb}_{1}(v_\agenti, \Jgoods)\), the normalized valuation of agent \(\agent_i\) over \(\Jgoods\) for \(\Jnumb\) bundles.
	
	We observe that \(0 \le x \le y \le \numb\) and verify that \(\Jvalu\) satisfies all the conditions required in \cref{lem:N2bag}.  \cref{firstcond,secondcond} follow directly from the definition of \(y\).
	 To verify \cref{thirdcond}, we estimate the sum
		\[
		\sum_{k=1}^{x} \Jvalu\left(\{\good_{\numb+k}\}\right),
		\]
		and provide separate bounds depending on the ratio \(\tfrac{x}{\Jnumb}\).
\end{proofsketch}

\begin{proof}
	Considering the following setup. Let \(0 \le y \le \numb\) be index such that
	\begin{itemize}
		\item $v_\agenti(\{\good_k,\,\good_{\numb+k}\}) \ge \alpha \quad \text{for } 1 \le k \le y.$
		\item $v_\agenti(\{\good_k,\,\good_{\numb+k}\}) < \alpha \quad \text{for } y < k \le \numb.$
	\end{itemize}
	
	Next, let \(x\le y\) be the smallest index satisfying
	$
	\mms^{\numb-(y-x)}_{v_\agenti}\Bigl(\goods \setminus \bigcup_{k=x+1}^{y} \{\good_k,\,\good_{\numb+k}\}\Bigr) \ge 1.
	$
	Since $\dotmms{\valu_i} \ge 1$ such $x$ exist. Let $$\Jnumb=\numb - (y - x) \quad\text{and} \quad\Jgoods = \goods \setminus \bigcup_{k=x+1}^{y} \{\good_k,\,\good_{\numb+k}\},$$ define
	$
	\Jvalu = \text{normalized}^{\Jnumb}_{1}(\valu_\agenti, \Jgoods),
	$
	We verify that \(0 \le x \le y \le \numb\) and valuation function $\Jvalu$ satisfies all conditions required by \cref{lem:N2bag}. \cref{firstcond} and \cref{secondcond} hold by definition of \(y\). Before verifying \cref{thirdcond}, we establish the following claim:
	\[
	\Jvalu(\{\good_k,\,\good_{\numb+k}\}) \;>\; 1
	\quad\text{for all }1 \le k \le x.
	\]
	Indeed, if for some \(k\le x\) we had \(\Jvalu(\{\good_k,\,\good_{\numb+k}\}) \le 1\), then \(\Jvalu(\{\good_x,\,\good_{\numb+x}\})\le 1 \), since 
	$\mms^{\Jnumb}_{\Jvalu}(\Jgoods)\;\ge 1,$
	we would get
	$
	\mms^{\Jnumb-1}_{\Jvalu}\Bigl(\Jgoods\setminus\{\good_x,\,\good_{\numb+x}\}\Bigr)\;\ge\;1,
	$
	therefore
	$
	\mms^{\numb-(y-(x-1))}_{\valu_i}\Bigl(\goods\setminus\bigcup_{k=x}^{y}\{\good_k,\,\good_{\numb+k}\}\Bigr)\;\ge\;1,
	$
	contradicting the minimality of \(x\). Now we are ready to verify \cref{thirdcond}.  We distinguish four cases according to the ratio $\tfrac{x}{\Jnumb}$:
	\begin{itemize}
		\item \textbf{Case 1.} $x = 0$: 
		\begin{align*}
			&x\,\Jvalu(\{\good_1\}) + \sum_{k=1}^{x} \Jvalu(\{\good_{\numb+k}\}) + (\Jnumb-x)(\alpha + \Jvalu(\{\good_{2\numb+1}\}))\\
			&= \Jnumb(\alpha + \Jvalu(\{\good_{2\numb+1}\}))\\
			&< \Jnumb(\alpha + (1 - \alpha)) & \agent_i \in \Ndo, \\
			&= \Jnumb\mms^{\Jnumb}_{\Jvalu}\Bigl(\Jgoods\Bigr).
		\end{align*}
		
		\item \textbf{Case 2.} $0 < \tfrac{x}{\Jnumb} \le \tfrac{3}{5}$:
		\begin{align*}
			&x\,\Jvalu(\{\good_1\}) + \sum_{k=1}^{x} \Jvalu(\{\good_{\numb+k}\}) + (\Jnumb-x)(\alpha + \Jvalu(\{\good_{2\numb+1}\}))\\
			&< x(\Jvalu(\{\good_1\}) + \frac{\alpha}{2}) + (\Jnumb-x)(\alpha + \Jvalu(\{\good_{2\numb+1}\}))& \text{$\reductiontype^1$ is not applicable}, \\
			&\le x(\alpha-\Jvalu(\{\good_{2\numb+1}\}) + \frac{\alpha}{2}) + (\Jnumb-x)(\alpha + \Jvalu(\{\good_{2\numb+1}\}))& \text{$\rstar^1$ is not applicable},\\[1ex]
			&=(\Jnumb-2x)\Jvalu(\{\good_{2\numb+1}\}) + \frac{\alpha(2\Jnumb+x)}{2}.
		\end{align*} 
		First, consider the case where \( x < \tfrac{\Jnumb}{2} \). Since the coefficient of \(\Jvalu(\{\good_{2\numb+1}\})\) is \((\Jnumb - 2x) > 0\), we have:
		\begin{align*}
			&x\,\Jvalu(\{\good_1\}) + \sum_{k=1}^{x} \Jvalu(\{\good_{\numb+k}\}) + (\Jnumb-x)(\alpha + \Jvalu(\{\good_{2\numb+1}\}))\\
			&< (\Jnumb-2x)\Jvalu(\{\good_{2\numb+1}\}) + \frac{\alpha(2\Jnumb+x)}{2} \\
			&\le (\Jnumb-2x)(1-\alpha) + \frac{\alpha(2\Jnumb+x)}{2} & \agent_i \in \Ndo, \\
			&= x\Bigl(\frac{5\alpha}{2} - 2\Bigr) + \Jnumb \\
			&\le \Jnumb &\alpha \le \frac{4}{5},\\
			&= \Jnumb\mms^{\Jnumb}_{\Jvalu}\Bigl(\Jgoods\Bigr).
		\end{align*}
		Next, consider the case where \(x \ge \tfrac{\Jnumb}{2}\):
		\begin{align*}
			&x\,\Jvalu(\{\good_1\}) + \sum_{k=1}^{x} \Jvalu(\{\good_{\numb+k}\}) + (\Jnumb-x)(\alpha + \Jvalu(\{\good_{2\numb+1}\}))\\
			&< (\Jnumb-2x)\Jvalu(\{\good_{2\numb+1}\}) + \frac{\alpha(2\Jnumb+x)}{2} \\
			&\le \frac{\alpha(2\Jnumb+x)}{2} \\
			&\le \frac{\alpha\Bigl(2\Jnumb+\frac{3\Jnumb}{5}\Bigr)}{2} & x \le \frac{3\Jnumb}{5}, \\
			&= \frac{13\alpha \Jnumb}{10} \\
			&\le \Jnumb &\alpha \le \frac{10}{13},\\
			&= \Jnumb\mms^{\Jnumb}_{\Jvalu}\Bigl(\Jgoods\Bigr).
		\end{align*}
		\item \textbf{Case 3.} $\tfrac{3}{5} < \tfrac{x}{\Jnumb} \le \tfrac{2}{3}$: Let \((P_1,\dots,P_{\Jnumb})\) be a partition of \(\Jgoods\) with 
		\[
		\Jvalu(P_k)=1\quad\text{for all }k.
		\]  Note that since 
		$\Jvalu(\{\good_x,\,\good_{\numb+x}\}) > 1,$  
		the goods \(\good_{1}\) to \(\good_{x}\) lie in \(x\) distinct bundles and the goods \(\good_{y+1}\) to \(\good_{\numb+x}\) lie in \(\Jnumb-x\) remaining bundles. We claim that 
		\begin{align}
			\sum_{k=1}^{x} \Jvalu(\{\good_{\numb+k}\}) \le \frac{x}{3}.\label[ineq]{sum3}
		\end{align}
		To prove this claim, we aim to show that among the goods \(\good_{y+1}\) to \(\good_{\numb+x}\), at least \(x\) goods belong to at most \(x/3\) bundles; therefore, that sum is at most \(\tfrac{x}{3}\)—since the goods \(\good_{\numb+1}\) to \(\good_{\numb+x}\) are the smallest \(x\) goods among these \(\Jnumb\) goods, the claim follows. Denote the \(\Jnumb-x\) bundles by \(P_1, P_2, \ldots, P_{\Jnumb-x}\) and suppose that in \(P_k\) there are \(c_k\) goods among \(\good_{y+1}\) to \(\good_{\numb+x}\) (with \(c_1 \ge c_2 \ge \cdots \ge c_{\Jnumb-x}\)). Let \(l\) be the largest index such that 
		$
		\sum_{k=1}^{l} c_k \ge 3l.
		$
		
		Note that by definition, \(l\) is the largest index such that the first \(l\) bundles satisfy \(\sum_{k=1}^{l} c_k \ge 3l\). This means that for the first \(l\) bundles, the total number of goods among \(\good_{y+1}\) to \(\good_{\numb+x}\), is at least \(3l\), but when we include the \((l+1)\)-th bundle, we no longer have this property; in fact, the total number of goods among \(\good_{y+1}\) to \(\good_{\numb+x}\), in the first \(l+1\) bundles is less than \(3(l+1)\). Moreover, for every bundle from index \(l+1\) onward, each bundle can have at most 2 goods—if any such bundle contained 3 or more goods, then we could increase \(l\), contradicting its maximality.

		Thus, an upper bound for the total number of goods is given by assuming that the first \(l+1\) bundles contain at most \(3(l+1)-1\) goods among \(\good_{y+1}\) to \(\good_{\numb+x}\), and that each of the remaining \((\Jnumb-x-l-1)\) bundles contains at most 2 goods among \(\good_{y+1}\) to \(\good_{\numb+x}\). This leads to the inequality
		$
		3(l+1)-1 + 2(\Jnumb-x-l-1) \ge \Jnumb,
		$
		which represents an upper bound on the total number of goods among \(\good_{y+1}\) to \(\good_{\numb+x}\). 
		We simplify it to obtain
		\begin{align*}
			l &\ge 2x - \Jnumb\\
			&>2x-\frac{5x}{3} & x>\frac{3\Jnumb}{5},\\
			&=\frac{x}{3}.
		\end{align*}
		
		Thus, we conclude that the sum of the values of the \(x\) smallest goods among \(\good_{y+1}\) to \(\good_{\numb+x}\), is at most \(\tfrac{x}{3}\), which implies \cref{sum3}. We have: 
		\begin{align*}
			&x\,\Jvalu(\{\good_1\}) + \sum_{k=1}^{x} \Jvalu(\{\good_{\numb+k}\}) + (\Jnumb-x)(\alpha + \Jvalu(\{\good_{2\numb+1}\}))\\
			&\le x\,\Jvalu(\{\good_1\}) + \frac{x}{3} + (\Jnumb-x)(\alpha + \Jvalu(\{\good_{2\numb+1}\}))&\mbox{\cref{sum3}}, \\
			&< x(\alpha-\Jvalu(\{\good_{2\numb+1}\})) + \frac{x}{3} + (\Jnumb-x)(\alpha + \Jvalu(\{\good_{2\numb+1}\}))&\text{$\rstar^1$ is not applicable}, \\
			&= \Jvalu(\{\good_{2\numb+1}\})(\Jnumb-2x) + \Bigl(\Jnumb\alpha + \frac{x}{3}\Bigr)\\[1ex]
			&\le \Jnumb\alpha + \frac{x}{3} & \Jnumb-2x<0,\\[1ex]
			&\le \Jnumb\alpha +\frac{2\Jnumb}{9} & x<\frac{2\Jnumb}{3},\\[1ex]
			&\le \Jnumb &\alpha \le \frac{7}{9},\\
			&= \Jnumb\mms^{\Jnumb}_{\Jvalu}\Bigl(\Jgoods\Bigr).
		\end{align*}
		\item \textbf{Case 4.} $\tfrac{2}{3} < \tfrac{x}{\Jnumb} \le 1$:
		Let \((P_1,\dots,P_{\Jnumb})\) be a partition of \(\Jgoods\) with 
		\[
		\Jvalu(P_k)=1\quad\text{for all }k.
		\]  Note that since 
		$\Jvalu(\{\good_x,\,\good_{\numb+x}\}) > 1,$  
		the goods \(\good_{1}\) to \(\good_{x}\) lie in \(x\) distinct bundles and the goods \(\good_{y+1}\) to \(\good_{\numb+x}\) lie in \(\Jnumb-x\) remaining bundles. Since the goods \(\good_{\numb+1}\) to \(\good_{\numb+x}\) form the smallest \(x\)-element subset among these, their total sum under $\Jvalu$ is at most \(\tfrac{x(\Jnumb-x)}{\Jnumb}.\) We have: 
		\begin{align*}
			&x\,\Jvalu(\{\good_1\}) + \sum_{k=1}^{x} \Jvalu(\{\good_{\numb+k}\}) + (\Jnumb-x)(\alpha + \Jvalu(\{\good_{2\numb+1}\}))\\
			&\le x\,\Jvalu(\{\good_1\}) + \frac{x(\Jnumb-x)}{\Jnumb} + (\Jnumb-x)(\alpha + \Jvalu(\{\good_{2\numb+1}\}))\\
			&< x(\alpha-\Jvalu(\{\good_{2\numb+1}\})) + \frac{x(\Jnumb-x)}{\Jnumb} + (\Jnumb-x)(\alpha + \Jvalu(\{\good_{2\numb+1}\}))&\text{$\rstar^1$ is not applicable}, \\
			&= \Jvalu(\{\good_{2\numb+1}\})(\Jnumb-2x) + \Bigl(\Jnumb\alpha + \frac{x(\Jnumb-x)}{\Jnumb}\Bigr)\\[1ex]
			&\le \Jnumb\alpha + \frac{x(\Jnumb-x)}{\Jnumb} & \Jnumb-2x<0,\\[1ex]
			&\le \Jnumb\alpha + \frac{\frac{2\Jnumb}{3}\Bigl(\Jnumb-\frac{2\Jnumb}{3}\Bigr)}{\Jnumb} &x>\frac{2\Jnumb}{3},\\[1ex]
			&= \Jnumb\left(\alpha + \frac{2}{9}\right)\\[1ex]
			&< \Jnumb &\alpha < \frac{7}{9},\\
			&=\Jnumb\mms^{\Jnumb}_{\Jvalu}\Bigl(\Jgoods\Bigr).
		\end{align*} 
	\end{itemize}
	Completing the proof in all 4 cases. 
\end{proof}
By \cref{lem:N-minus}, for every red agent \(\agent_i \in \agents\) with \(\dotmms{\valu_i} < 1\), there exists a value \(s_\agenti\) such that conditions \cref{conds1}, \Cref{conds2}, and \cref{conds3} are satisfied. Fixing these values \(s_i\), we now proceed to prove \cref{lem:N22values}.

\begin{observation}\label{lem:N22values}
	For every red agent $\agent_i$ in $\agents$ with $\dotmms{\valu_i} < 1$,  we have 
		\begin{align}
		&\auxwal{\vF_{s_\agenti}}({\{\good_1\}}) \le 1-\frac{\alpha}{3} - 2 s_\agenti,
		&\label[ineq]{fcond}\\
		& \auxwal{\vF_{s_\agenti}}(\{\good_{\numb+1}\})  \le \frac{\alpha}{2} - s_\agenti,
		&\label[ineq]{scond}\\
		&\auxwal{\vF_{s_\agenti}}(\{\good_{2\numb+1}\})  < \frac{4\alpha}{3}-1 - s_\agenti.&\label[ineq]{thcond}
	\end{align}
\end{observation}
\begin{proof}
	We prove it one by one. 

		By \cref{lem:N-minus}, $\valu_\agenti(\{\good_{1}\}) \le 1-\tfrac{\alpha}{3} + s_\agenti$, and by the definition of $\vF_{s_\agenti}$, we have
		\begin{align*}
			\auxwal{\vF_{s_\agenti}}(\{\good_1\}) &\le \fs{1-\frac{\alpha}{3} +  s_\agenti}\\
			&=1-\frac{\alpha}{3} - 2 s_\agenti.
		\end{align*}
		
		Since $\reductiontype^1$ is not applicable, it follows that $\valu_\agenti(\{\good_{\numb+1}\}) < \tfrac{\alpha}{2}.$
		Now, by the definition of $\vF_{s_\agenti}$, since $\tfrac{\alpha}{2} \ge \tfrac{\alpha}{3}$
		it follows that
		\begin{align*}
			\fs{\{\good_{\numb+1}\}} & \le\auxwal{\vF_{s_\agenti}}\Bigl(\frac{\alpha}{2}\Bigr)\\
			&\le \frac{\alpha}{2} - s_\agenti.
		\end{align*}
		
		By definition of $\Ndo$, we have $\valu_i(\{\good_{2\numb+1}\})  \le 1-\alpha$ and by \cref{lem:N-minus} there is no good in $\goods$ with value in $$\Bigl[\frac{4\alpha}{3}-1 - s_\agenti, 1-\alpha\Bigr]$$. Therefore  $\fs{\{\good_{2\numb+1}\}}< \tfrac{4\alpha}{3}-1 - s_\agenti$. 
\end{proof}

\begin{lemrep}\label{lem:N22J}
	Every red agent $\agent_i$ in $\agents$ with $\dotmms{\valu_i} < 1$ after the primary reductions, receives a bundle in \cref{algo:N2}.
\end{lemrep}

\begin{proofsketch}
		We begin by defining an index \(0 \le y \le \numb\) such that:
	\begin{itemize}
		\item \(\auxwal{\vF_{s_\agenti}}(\{\good_k,\,\good_{\numb+k}\}) \ge \alpha\) for all \(1 \le k \le y\), and
		\item \(\auxwal{\vF_{s_\agenti}}(\{\good_k,\,\good_{\numb+k}\}) < \alpha\) for all \(y < k \le \numb\).
	\end{itemize}

	Next, we identify the smallest index \(x \le y\) such that the maximin share with respect to the remaining \(\numb - (y - x)\) bundles is at least $1-3s_i$:
	\[
	\mms^{\numb-(y-x)}_{\auxwal{\vF_{s_\agenti}}}\Bigl(\goods \setminus \bigcup_{k=x+1}^{y} \{\good_k,\,\good_{\numb+k}\}\Bigr) \ge 1-3s_i.
	\]
	Such an \(x\) must exist because \(\dotmms{\auxwal{\vF_{s_\agenti}}} \ge 1-3s_i\) by assumption.
	Now, define:
	\[
	\Gnumb=\numb - (y - x), \quad\quad\Ggoods = \goods \setminus \bigcup_{k=x+1}^{y} \{\good_k,\,\good_{\numb+k}\},
	\]
	and let by \cref{def:norm} \(\Gvalu = \text{normalized}^{\Gnumb}_{1-3s_i}(\auxwal{\vF_{s_\agenti}}, \Ggoods)\), the normalized valuation of agent \(\agent_i\) over \(\Ggoods\) for \(\Gnumb\) bundles.
	We observe that \(0 \le x \le y \le \numb\) and verify that \(\Gvalu\) satisfies all the conditions required in \cref{lem:N2bag}.  \cref{firstcond,secondcond} follow directly from the definition of \(y\).
	To verify \cref{thirdcond}, we estimate the sum
	\[
	\sum_{k=1}^{x} \Gvalu\left(\{\good_{\numb+k}\}\right),
	\]
	and provide separate bounds depending on the ratio \(\tfrac{x}{\Gnumb}\).	
\end{proofsketch}

\begin{proof}
	
	Considering the following setup. Let \(0 \le y \le \numb\) be index such that
	\begin{itemize}
		\item $\auxwal{\vF_{s_\agenti}}(\{\good_k,\,\good_{\numb+k}\}) \ge \alpha \quad \text{for } 1 \le k \le y.$
		\item $\auxwal{\vF_{s_\agenti}}(\{\good_k,\,\good_{\numb+k}\}) < \alpha \quad \text{for } y < k \le \numb.$
	\end{itemize}
	
	Next, let \(x\le y\) be the smallest index satisfying
	$
	\mms^{\numb-(y-x)}_{\auxwal{\vF_{s_\agenti}}}\Bigl(\goods \setminus \bigcup_{k=x+1}^{y} \{\good_k,\,\good_{\numb+k}\}\Bigr) \ge 1-3s_i.
	$
	Since $\dotmms{\auxwal{\vF_{s_\agenti}}} \ge 1-3s_i$ such $x$ exist. Let $$\Gnumb=\numb - (y - x) \quad\text{and} \quad\Ggoods = \goods \setminus \bigcup_{k=x+1}^{y} \{\good_k,\,\good_{\numb+k}\}$$ define
	$
	\Gvalu = \text{normalized}^{\Gnumb}_{1-3s_i}(\auxwal{\vF_{s_\agenti}}, \Ggoods),
	$
	We verify that \(0 \le x \le y \le \numb\) and valuation function $\Gvalu$ satisfies all conditions required by \cref{lem:N2bag}. \cref{firstcond} and \cref{secondcond} hold by definition of \(y\). Before verifying \cref{thirdcond}, we establish the following claim:
	\[
	\Gvalu(\{\good_k,\,\good_{\numb+k}\}) \;>\; 1-3s_i
	\quad\text{for all }1 \le k \le x.
	\]
	Indeed, if for some \(k\le x\) we had \(\Gvalu(\{\good_k,\,\good_{\numb+k}\}) \le 1-3s_i\), then \(\Gvalu(\{\good_x,\,\good_{\numb+x}\})\le 1-3s_i \), since 
	$\mms^{\Gnumb}_{\Gvalu}(\Ggoods)\;\ge 1-3s_i,$
	we would get
	$
	\mms^{\Gnumb-1}_{\Gvalu}\Bigl(\Ggoods\setminus\{\good_x,\,\good_{\numb+x}\}\Bigr)\;\ge\;1-3s_i,
	$
	therefore
	$
	\mms^{\numb-(y-(x-1))}_{\auxwal{\vF_{s_\agenti}}}\Bigl(\goods\setminus\bigcup_{k=x}^{y}\{\good_k,\,\good_{\numb+k}\}\Bigr)\;\ge\;1-3s_i,
	$
	contradicting the minimality of \(x\). Now we are ready to verify \cref{thirdcond}.  We distinguish four cases according to the ratio $\tfrac{x}{\Gnumb}$:
	\begin{itemize}
		\item \textbf{Case 1.} $\tfrac{x}{\Gnumb} \le \tfrac{1}{2}$: 
		\begin{align*}
			&x\,\Gvalu(\{\good_1\}) + \sum_{k=1}^{x} \Gvalu(\{\good_{\numb+k}\})
			\\& + (\Gnumb-x)(\alpha + \Gvalu(\{\good_{2\numb+1}\}))\\
			&\le x\Gvalu(\{\good_1,\,\good_{\numb+1}\}) + (\Gnumb-x)(\alpha + \Gvalu(\{\good_{2\numb+1}\}))\\
			&\le x\Bigl(1-\frac{\alpha}{3} - 2 s_\agenti + \frac{\alpha}{2}-s_\agenti\Bigr) + (\Gnumb-x)\Bigl(\alpha + \frac{4\alpha}{3}-1 - s_\agenti\Bigr) &\mbox{ \cref{lem:N22values}},\\
			&= \Gnumb(1-3s_\agenti) + (2\Gnumb-2x)s_\agenti - \frac{2\Gnumb(6-7\alpha) + x(13\alpha-12)}{6}\\
			&\le \Gnumb(1-3s_\agenti) + (2\Gnumb-2x)\Bigl(\frac{4\alpha}{3}-1\Bigr) - \frac{2\Gnumb(6-7\alpha) + x(13\alpha-12)}{6} & \mbox{\cref{conds1}},\\
			&= \Gnumb(1-3s_\agenti)+x\Bigl(4-\frac{29\alpha}{6}\Bigr) + \Gnumb(5\alpha-4)\\
			&\le \Gnumb(1-3s_\agenti)+\Gnumb\Bigl(2-\frac{29\alpha}{12}\Bigr) + \Gnumb(5\alpha-4) & x\le \frac{\Gnumb}{2},\\
			&= \Gnumb(1-3s_\agenti)+\frac{\Gnumb(31\alpha-24)}{12}\\
			&< \Gnumb(1-3s_\agenti) & \alpha < \frac{24}{31},\\
			&=\Gnumb\mms^{\Gnumb}_{\Gvalu}\Bigl(\Ggoods\Bigr).
		\end{align*}
		\item \textbf{Case 2.} $\tfrac{1}{2} < \tfrac{x}{\Gnumb} \le \tfrac{3}{5}$: Let \((P_1,\dots,P_{\Gnumb})\) be a partition of \(\Ggoods\) with 
		\[
		\Gvalu(P_k)=1-3s_i\quad\text{for all }k.
		\]  Note that since 
		$\Gvalu(\{\good_x,\,\good_{\numb+x}\}) > 1-3s_i,$  
		the goods \(\good_{1}\) to \(\good_{x}\) lie in \(x\) distinct bundles and the goods \(\good_{y+1}\) to \(\good_{\numb+x}\) lie in \(\Gnumb-x\) remaining bundles. We now show that it is impossible for any of these $\Gnumb-x$ bundles to contain 4 goods among the $\Gnumb$ goods. Suppose that one of these bundles does contain at least 4 goods among \(\good_{y+1}\) to \(\good_{\numb+x}\). Then
		$
		\Gvalu{\{\good_{\numb+x}\}} \le \tfrac{1-3 s_\agenti}{4}.
		$
		It follows that:
		\begin{align*}
			1-3s_\agenti&< \Gvalu{\{\good_x,\,\good_{\numb+x}\}}\\
			&\le (1 - \frac{\alpha}{3} - 2 s_\agenti) + \frac{1-3 s_\agenti}{4}&\mbox{\cref{lem:N22values}}.\\
		\end{align*}
		Which implies $s_\agenti > \tfrac{4\alpha}{3}-1,$ that's a contradiction. Hence, it is impossible for any of the \(\Gnumb-x\) bundles to contain 4 goods among  \(\good_{y+1}\) to \(\good_{\numb+x}\).
		Assume that among $\Gnumb-x$ bundles, exactly \(r\) of them contains exactly 3 goods among \(\good_{y+1}\) to \(\good_{\numb+x}\). Then we must have
		$
		2(\Gnumb-x-r)+3r\ge \Gnumb,
		$
		which implies
		$
		r\ge 2x-\Gnumb.
		$
		Since the total value of each bundle is exactly \(1-3 s_\agenti\), it follows that the total value of the $3(2x - \Gnumb)$ minimum goods (among  \(\good_{y+1}\) to \(\good_{\numb+x}\)) is at most 
		$
		(1-3 s_\agenti)(2x-\Gnumb).
		$
		On the other hand, since \(x\le \tfrac{3\Gnumb}{5}\) we have
		$3(2x-\Gnumb)\le x$, hence these \(3(2x-\Gnumb)\) goods fall in the range from $\good_{\numb+1}$ to $\good_{\numb+x}$. 
		
		Therefore we can obtain 
		\begin{align*}
			\sum_{k=1}^{x} \Gvalu({\{\good_{\numb+k}\}})&\le(2x - \Gnumb)(1 - 3 s_\agenti) + \Bigl(x - 3(2x - \Gnumb)\Bigr)\Bigl(\frac{\alpha}{2}-s_\agenti\Bigr) &\mbox{\cref{lem:N22values}},\\[1ex]
			& =(2x - \Gnumb)(1 - 3 s_\agenti) + \Bigl(3\Gnumb-5x\Bigr)\Bigl(\frac{\alpha}{2}-s_\agenti\Bigr)\\[1ex]
			&=\frac{\alpha(3\Gnumb-5x)}{2} + (2x-\Gnumb) - x\, s_\agenti.
		\end{align*}
		Therefore
		\begin{align*}
			&x\,\Gvalu(\{\good_1\}) + \sum_{k=1}^{x} \Gvalu(\{\good_{\numb+k}\})\\
			& + (\Gnumb-x)(\alpha + \Gvalu(\{\good_{2\numb+1}\}))\\
			&\le
			x(\Gvalu(\{\good_1\})+\Bigl(\frac{\alpha(3\Gnumb-5x)}{2} + (2x-\Gnumb) - x\, s_\agenti\Bigr)\\
			&\quad + (\Gnumb-x)(\alpha + \Gvalu(\{\good_{2\numb+1}\}))\\
			&\le x\Bigl(1 - \frac{\alpha}{3} - 2 s_\agenti\Bigr) + \frac{\alpha(3\Gnumb-5x)}{2} \\
			&\quad+ (2x-\Gnumb) - x\, s_\agenti + (\Gnumb - x)\Bigl(\alpha + \frac{4\alpha}{3} - 1 - s_\agenti\Bigr)&\mbox{\cref{lem:N22values}},\\[1ex]
			&= \Gnumb(1-3s_\agenti) + (2\Gnumb-2x)s_\agenti - \frac{\Gnumb(18-23\alpha)+x(31\alpha-24)}{6}\\[1ex]
			&\le \Gnumb(1-3s_\agenti) + (2\Gnumb-2x)\Bigl(\frac{4\alpha}{3}-1\Bigr) - \frac{\Gnumb(18-23\alpha)+x(31\alpha-24)}{6} & \mbox{\cref{conds1}},\\[1ex]
			&= \Gnumb(1-3s_\agenti) + \frac{39\Gnumb-47x}{6}\,\alpha + \bigl(6x-5\Gnumb\bigr)\\[1ex]
			&\le \Gnumb(1-3s_\agenti) + \frac{39\Gnumb-47x}{6}\,\frac{10}{13} + \bigl(6x-5\Gnumb\bigr) & \alpha\le\frac{10}{13},\\[1ex]
			&= \Gnumb(1-3s_\agenti)-\frac{x}{39}\\[1ex]
			&< \Gnumb(1-3s_\agenti)\\
			&=\Gnumb\mms^{\Gnumb}_{\Gvalu}\Bigl(\Ggoods\Bigr).
		\end{align*}
		
		\item \textbf{Case 3.} $\tfrac{3}{5} < \tfrac{x}{\Gnumb} \le \tfrac{2}{3}$:
		Let \((P_1,\dots,P_{\Gnumb})\) be a partition of \(\Ggoods\) with 
		\[
		\Gvalu(P_k)=1-3s_i\quad\text{for all }k.
		\] 
		In Case 2, we showed that among \(\good_{y+1}\) to \(\good_{\numb+x}\) there are at least 
		\(3(2x - \Gnumb)\) goods, which are grouped into at most \((2x - \Gnumb)\) bundles (each bundle containing exactly 3 goods among \(\good_{y+1}\) to \(\good_{\numb+x}\)). Since $x > \tfrac{3\Gnumb}{5}$, we have $3(2x-\Gnumb) > x$, it follows that the total value of the goods 
		\(\good_{\numb+1}\) to \(\good_{\numb+x}\) is at most 
		$
		\frac{x(1-3s_\agenti)}{3}.
		$
		We have: 
		
		\begin{align*}
			&x\,\Gvalu(\{\good_1\}) + \sum_{k=1}^{x} \Gvalu(\{\good_{\numb+k}\})\\ 
			&+ (\Gnumb-x)(\alpha + \Gvalu(\{\good_{2\numb+1}\}))\\
			&\le x(\Gvalu(\{\good_1\}) + \frac{x(1-3 s_\agenti)}{3} + (\Gnumb-x)(\alpha + \Gvalu(\{\good_{2\numb+1}\})) \\
			&< x\Bigl(1-\frac{\alpha}{3}-2 s_\agenti\Bigr) + \frac{x(1-3 s_\agenti)}{3} + (\Gnumb-x)\Bigl(\alpha+\frac{4\alpha}{3}-1-s_\agenti\Bigr) &\mbox{\cref{lem:N22values}},\\[1ex]
			&= \Gnumb(1-3s_\agenti) + (2\Gnumb-2x)s_\agenti - \left[\Gnumb - \frac{\Gnumb(7\alpha-3)+x(7-8\alpha)}{3}\right]\\[1ex]
			&\le \Gnumb(1-3 s_\agenti) + (2\Gnumb-2x)\Bigl(\frac{4\alpha}{3}-1\Bigr) - \left[\Gnumb - \frac{\Gnumb(7\alpha-3)+x(7-8\alpha)}{3}\right] & \mbox{\cref{conds1}},\\[1ex]
			&= \Gnumb(1-3s_\agenti)+\frac{13-16\alpha}{3}\,x + \Gnumb(5\alpha-4)\\[1ex]
			&\le \Gnumb(1-3s_\agenti)+\frac{13-16\alpha}{3}\,\frac{2\Gnumb}{3} + \Gnumb(5\alpha-4) & x\le\frac{2\Gnumb}{3},\\[1ex]
			&= \Gnumb(1-3s_\agenti)+\frac{\Gnumb(13\alpha-10)}{9}\\[1ex]
			&\le \Gnumb(1-3s_\agenti) & \alpha \le \frac{10}{13},\\
			&=\Gnumb\mms^{\Gnumb}_{\Gvalu}\Bigl(\Ggoods\Bigr).
		\end{align*}
		\item \textbf{Case 4.} $\tfrac{2}{3} < \tfrac{x}{\Gnumb} \le 1$:
		Let \((P_1,\dots,P_{\Gnumb})\) be a partition of \(\Ggoods\) with 
		\[
		\Gvalu(P_k)=1-3s_i\quad\text{for all }k.
		\] As established in case 2, among  \(\good_{y+1}\) to \(\good_{\numb+x}\), no four of them can be placed together in any of the \(\Gnumb-x\) bundles. This implies that the total number of goods (among  \(\good_{y+1}\) to \(\good_{\numb+x}\)) in these bundles must satisfy
		$
		3(\Gnumb - x) \geq \Gnumb.
		$
		Rearranging gives
		$
		x \leq \tfrac{2\Gnumb}{3}.
		$
		However, this contradicts the assumption of Case 4 that \(x > 2\Gnumb/3\). Hence, this case is impossible.
	\end{itemize}
	‌Completing the proof. 
\end{proof}
Finally, in \cref{lem:N1-in-N2}, we show that every green agent is allocated a bundle during the Bag-Filling process.
\begin{lemma}\label{lem:N1-in-N2}
	Every green agent $\agent_i \in \agents$ receives a bundle in \cref{algo:N2}.
\end{lemma}

\begin{proof}
We argue by contradiction. Suppose there exists a green agent \(\agent_i \in \agents\) who does not receive any bundle in \cref{algo:N2}. Let the algorithm terminate while filling bag \(\finib_\halt\). This means $\Pnumb - (\numb-\halt+1)$ agents have received a bundle of value at least $\alpha$, and bags \(\finib_\halt, \finib_{\halt+1}, \ldots, \finib_{\numb}\) are unallocated. Note that by \Cref{lem:N21J} and \cref{lem:N22J}, all red agents have received a bundle, and only green agents remain. 
Our goal is to show $v_\agenti(\Pgoods) < \Pnumb$, 
contradicting \(\allmms{\valu_\agenti}=1\).
For every agent $\agent_j$ such that $\agent_j$ received a bundle during the primary reductions or Bag-filling, we denote $A_j$ by the bundles allocated to her. By the prioritization of green agents in \cref{algo:pr,algo:N2}, every bundle allocated to a red agent satisfies
$v_\agenti(A_j) < \alpha$. 
We now show that every bundle allocated to a green agent satisfies
$v_\agenti(A_j) \le 4\alpha-2$. 
	
	\paragraph{Primary Reductions.} Consider a green agent \(\agent_j \in \Nyek\), and let \(A_j\) be the bundle allocated to \(\agent_j\) through a reduction on the instance \(\prinstance = (\pragents, \prgoods)\).  
	Let \(y\) be the index such that \(\prgoodrat{y} = \good_{2\numb+1}\) (i.e., \(\good_{2\numb+1}\) is the \(y\)-th good in $\prgoods$).
	By the definition of green agents, the value of $\prgoodrat{y}$ for \(\agent_j\) satisfies
	\[
	1 - \alpha \le v_\agenti(\{\prgoodrat{y}\}) \le \frac{\alpha}{3}.
	\]
We now consider several cases based on the pattern of reduction applied in this step.

If the reduction is $\rho = \riki{\prinstance} {\reductiontype^0}{x}{\agent_j}{\seinstance}$, then the allocated bundle satisfies \(v_\agenti(A_j) \le 1 < 4\alpha - 2\), since \(\alpha > \tfrac{3}{4}\).
Now, suppose the reduction is $\rho = \riki{\prinstance} {\reductiontype^1}{x}{\agent_j}{\seinstance}$, so the allocated bundle is \(A_j = \{\prgoodrat{\prnumb}, \prgoodrat{x}\}\). If the second good in the bundle has value at most \(\tfrac{\alpha}{3}\), then the total value is at most \(\alpha + \tfrac{\alpha}{3} \le 4\alpha - 2\).

Otherwise, \(v_\agenti(\{\prgoodrat{x}\}) > \tfrac{\alpha}{3}\), which implies \(x < y\). By definition of $\reductiontype^1$, we have \(v_\agenti(\{\prgoodrat{\prnumb}, \prgoodrat{y}\}) < \alpha\), and since \(\prgoodrat{y}\) has value at least \(1 - \alpha\), we get \(v_\agenti(\{\prgoodrat{\prnumb}\}) < 2\alpha - 1\). Therefore, the value of the bundle is at most
\[
v_\agenti(A_j) \le 2\,v_\agenti(\{\prgoodrat{\prnumb}\}) < 4\alpha - 2.
\]

Now consider the case where the reduction is $\rho = \riki{\prinstance} {\reductiontype^2}{x}{\agent_j}{\seinstance}$, which assigns the bundle \(A_j = \{\prgoodrat{2\prnumb-1}, \prgoodrat{2\prnumb}, \prgoodrat{x}\}\). Since $\reductiontype^1$ is not applicable, we have 
\[
v_\agenti(\{\prgoodrat{2\prnumb-1}\}) + v_\agenti(\{\prgoodrat{2\prnumb}\}) < \alpha.
\]
If the value of the third good is small, say \(v_\agenti(\{\prgoodrat{x}\}) \le \tfrac{\alpha}{3}\), then the total value is bounded by \(\alpha + \tfrac{\alpha}{3} \le 4\alpha - 2\).
Otherwise,  \(v_\agenti(\{\prgoodrat{x}\}) > \tfrac{\alpha}{3}\), and must have \(x < y\), by definition of $\reductiontype^2$ we can conclude
\[
v_\agenti(\{\prgoodrat{2\prnumb-1}\}) + v_\agenti(\{\prgoodrat{2\prnumb}\}) + v_\agenti(\{\prgoodrat{y}\}) < \alpha,
\quad\text{and}\quad
v_\agenti(\{\prgoodrat{y}\}) \ge 1 - \alpha.
\]
This implies
\[
v_\agenti(\{\prgoodrat{2\prnumb-1}\}) + v_\agenti(\{\prgoodrat{2\prnumb}\}) < 2\alpha - 1.
\]
Since \(x > 2\prnumb\), we get
\[
v_\agenti(A_j) \le \tfrac{3}{2}\bigl(v_\agenti(\{\prgoodrat{2\prnumb-1}\}) + v_\agenti(\{\prgoodrat{2\prnumb}\})\bigr) < 3\alpha - \tfrac{3}{2} < 4\alpha - 2.
\]

Finally, if the reduction is $\rho = \riki{\prinstance} {\rstar^1}{x}{\agent_j}{\seinstance}$, then \(\reductiontype^0\) and \(\reductiontype^2\) are not applicable, which means
\[
v_\agenti(\{\prgoodrat{1}\}) < \alpha \quad \text{and} \quad v_\agenti(\{\prgoodrat{2\prnumb+1}\}) < \tfrac{\alpha}{3}.
\]
Hence, the total value of the bundle is at most
\[
v_\agenti(A_j) \le \alpha + \tfrac{\alpha}{3} \le 4\alpha - 2, \quad\text{since }\alpha > \tfrac{3}{4}.
\]
	
	\paragraph{Bag-filling}
	For each \(1\le k\le \numb\), $\finib_k$ is either $\{\good_k,\,\good_{\numb+k}\}$ or at least one good is added to it. In the first case, by definition of green agents, $v_i(\{\good_{2\numb+1}\}) \ge 1-\alpha$, and since $\rstar^1$ is not applicable, we have $v_i(\{\good_1\}) < 2\alpha - 1$. Therefore 
	\begin{align*}
		v_i(\{\good_k,\,\good_{\numb+k}\}) &< (2\alpha - 1) + \frac{\alpha}{2} < 4\alpha-2 & \alpha \ge \tfrac{2}{3}. 
	\end{align*}
	For the second case, we have 
	\begin{align*}
		v_i(\finib_k) &= v_i(\finib_k \setminus \{\good_x\}) + v_i(\{\good_x\}) &\\
		&< \alpha + \tfrac{\alpha}{3} & \mbox{$\reductiontype^2$ is not applicable},\\
		&< 4\alpha - 2 & \alpha > \tfrac{3}{4}. 
	\end{align*}
	
	In both cases $v_i(B_k) < 4\alpha-2$ holds. 
	
Now we calculate sum of all bundles $\finib_{\halt}, \ldots,\finib_{\numb}$ 
and $A_j$ for all satisfied agents $\agent_j$. Since all red agents are satisfied, we have:
	$
	\Pnumb < \alpha\,|\Ndo| + (4\alpha-2)\,(|\Nyek|-(\numb-\halt+1))+(4\alpha-2)\,(\numb-\halt+1).
	$
	Using $|\Nyek| < \tfrac{n}{\sqrt{2}}$
	we obtain
	\begin{align*}
		\Pnumb &< \alpha\,|\Ndo| + (4\alpha-2)\,|\Nyek| &\\[1mm]
		&\le (4\alpha-2)\,\frac{1}{\sqrt2}\,\Pnumb + \alpha\left(1-\frac{1}{\sqrt2}\right)\Pnumb &\\[1mm]
		&\le \Pnumb & \alpha\le\frac{4+\sqrt{2}}{7}\approx 0.7735.
	\end{align*}
	a contradiction. Therefore, every green agent in \(\agents\) receives a bundle.
\end{proof}

\begin{table}[H]
	\centering 
	\begin{tabular}{|l|l|l|l|}
		\hline
		\textbf{Case} & \textbf{Agent} & \textbf{Condition} & \textbf{Lemma} \\ \hline
		\multirow{5}{*}{Case 1: $|\Nyek| \ge \lbnone$} 
		& \multirow{4}{*}{green agents}
		& $\dotmms{v_i} \ge 1, \ddotmms{v_i} \ge 1$ & \cref{plus-one} \\ \cline{3-4}
		& & $\ddotmms{v_i} < 1, \ddotmms{\auxfun{\FJ}{\agenti}} \ge 4(1-\alpha)$ & \cref{minus-one} \\ \cline{3-4}
		& & $\dotmms{v_i} \ge 1, \ddotmms{v_i} < 1$ & \cref{plus-two} \\ \cline{3-4}
		& & $\ddotmms{v_i} < 1, \ddotmms{\auxfun{\FJ}{\agenti}} < 4(1-\alpha)$ & \cref{minus-two} \\ \cline{2-4}
		& red agents & — & \cref{lem:N2-in-N1} \\  \hline
		\multirow{3}{*}{Case 2: $|\Nyek| < \lbnone$} 
		& \multirow{2}{*}{red agents}
		& $\dotmms{v_i} \ge 1$ & \cref{lem:N21J} \\ \cline{3-4}
		& & $\dotmms{v_i} < 1$ & \cref{lem:N22J} \\ \cline{2-4}
		& green agents & — & \cref{lem:N1-in-N2} \\ \hline
	\end{tabular}
	\caption{Categorization for \cref{thm:mms}.}
	\label{table:main}
\end{table}

\newpage
\section{Putting the Pieces Together}
Finally, in this section, we bring all the components together to show that our algorithm guarantees a \((\mmsfrac{10}{13})\)-\(\MMS\) allocation.

\maintheorem*
\begin{proof} 
	
	First, by \cref{order_reduction}, without loss of generality, we can transform any instance into an ordered and normalized instance. After running the primary reductions, the algorithm branches into two cases. We analyze each case separately.
	
\textbf{Case 1.} When \(|N^g| \ge n/\sqrt{2}\), we first apply the secondary reductions from \Cref{algo:N11}, and then run the Bag-filling process from \Cref{algo:1}. By \Cref{plus-one}, \Cref{minus-one}, \Cref{plus-two},and \Cref{minus-two}, every green agent receives a bundle worth at least \(\alpha\). In addition, \Cref{lem:N2-in-N1} ensures that the red agents also recieve a bundle in this case. Since these groups together include all agents, it follows that every agent receives a bundle of value at least \(\alpha\).
	
	\textbf{Case 2.} When \(|N^g| < n/\sqrt{2}\), we perform the Bag-filling procedure in \cref{algo:N2}. By \cref{lem:N21J} and \cref{lem:N22J}, we have that all red receive a bundles of value at least \(\alpha\). Moreover, in \cref{lem:N1-in-N2}, we show that the green agents also recieve a bundle of value at least \(\alpha\) in this case.
	
	In both cases, the lemmas collectively ensure that every agent is allocated a bundle of value at least \(\alpha\). Moreover, all the constraints and assumptions imposed on \(\alpha\) throughout the analysis are indeed satisfied when \(\alpha = \mmsfrac{10}{13}\), and no larger value of \(\alpha\) satisfies all these conditions simultaneously. Table~\ref{table:main} summarizes the lemmas that cover all agent categories and conditions. Therefore, the algorithm guarantees a \((\mmsfrac{10}{13})\)-\(\MMS\) allocation.
\end{proof}
Finally, as a consequence of \Cref{thm:mms}, we show that for some constant \(\varepsilon > 0\), a \((\mmsfrac{10}{13} - \varepsilon)\)-\(\MMS\) allocation can be computed in polynomial time.

\begin{theorem}
	For every constant \(\varepsilon > 0\), we can find a \(\bigl(\mmsfrac{10}{13} - \varepsilon \bigr)\)-\(\MMS\) allocation in polynomial time.
\end{theorem}

\begin{proof}
All steps of our algorithm run in polynomial time, except for the normalization step, which requires computing the exact \(\MMS\) of each agent which is NP-hard. However, a PTAS due to \cite{woeginger1997polynomial} provides a \((1 - \varepsilon)\)-approximation for \(\MMS\) for a constant $\varepsilon$ in polynomial time. Using this approximation, we can estimate each agent’s \(\MMS\) value closely enough to ensure an overall \((\mmsfrac{10}{13} - \varepsilon)\)-approximation guarantee in polynomial time for constant $\varepsilon$.

The polynomial time implementation of the rest of the algorithm is mostly straightforward, except for identifying a perfect sequence of reductions. In this step, we iteratively select reductions one by one, always choosing the highest-priority reduction that still allows for a perfect matching between the resulting bundles and agents. The existence of such a matching can be verified in polynomial time using standard bipartite matching algorithms. Moreover, we can enforce agent priorities by treating the problem as a weighted matching: assign weight \(n\) to ordinary edges and \(n+1\) to prioritized ones. A maximum-weight matching under this scheme maximizes the number of matched prioritized agents and is computable in polynomial time.

Therefore, for any constant \(\varepsilon > 0\), we can find a \((\mmsfrac{10}{13} - \varepsilon)\)-\(\MMS\) allocation in polynomial time.
\end{proof}

	\newpage
	\bibliographystyle{alpha}
	\bibliography{ref}

	\newpage
	\appendix
	
\end{document}